\documentclass[11pt]{article}

\usepackage{amsmath}
\usepackage{amssymb}
\usepackage{enumitem}
\usepackage{comment}
\usepackage{cancel}
\usepackage{tikz}
\usepackage{caption}
\usepackage{url}
\usepackage{graphicx}
\usepackage{xcolor}
\usepackage{hyperref}
\usepackage{booktabs}
\usepackage{longtable}
\usetikzlibrary{shapes.geometric, arrows.meta}

\usepackage[authoryear]{natbib}

\newtheorem{theorem}{Theorem}

\newtheorem{algorithm}{Algorithm}
\newtheorem{assumption}{Assumption}

\newtheorem{definition}{Definition}
\newtheorem{remark}{Remark}

\newtheorem{example}{Example}

\newtheorem{proposition}{Proposition}

\newcommand{\ind}{\perp \!\!\!\perp}

\newenvironment{proof}[1][Proof]{\noindent\textbf{#1.} }{\ \rule{0.5em}{0.5em}}
\begin{document}     

\title{Filtering without recursion \\ and some of its uses in financial economics}

\author{
  Simon Donker van Heel\textsuperscript{a,b,*} and
  Neil Shephard\textsuperscript{c}
}
\date{\today}
\maketitle

\begin{center}
  \small
  \textsuperscript{a}\textit{Econometric Institute, Erasmus University Rotterdam, Rotterdam, The Netherlands} \\
  \textsuperscript{b}\textit{Tinbergen Institute, Amsterdam, The Netherlands} \\
  \textsuperscript{c}\textit{Dept of Economics \& Dept of Statistics, Harvard University, Cambridge, MA 02138, USA} 
\end{center}

\renewcommand{\thefootnote}{\fnsymbol{footnote}}
\footnotetext[1]{Corresponding author. 
E-mail addresses: \texttt{donkervanheel@ese.eur.nl} (S.W.\ Donker van Heel),
\texttt{shephard@fas.harvard.edu} (N.\ Shephard).  Donker van Heel thanks the Fulbright Program for financial support under their scholarship program.  We thank Rutger-Jan Lange and Dick van Dijk for comments on an earlier draft.  
{\tt Refine.ink} was used to review an earlier draft for consistency and clarity.}
\renewcommand{\thefootnote}{\arabic{footnote}}

\begin{abstract}
We develop a filter for time series, defined at each time $t$ as the minimizer of a discounted convex combination of observed and expected losses. The filter can be estimated by simulation to an arbitrary level of accuracy in $O(1)$ flops at each time point $t$ \& can be run for all values $t=1,...,T$ in parallel. These methods are applied to robustly compute a preaveraged price process from the more than 1.5 million trades made on a single financial asset in a single day where the noise's variance is infinite. It yields a flat ``volatility signature'' plot, down to the 1 second level, so the microstructure noise no longer biases the volatility estimate. 
This is not true when linear methods are employed.      
\end{abstract}


\baselineskip=20pt

\section{Introduction}

This paper defines weighted estimands for filtering, smoothing \& prediction \citep[e.g.][]{Whittle:83,Harvey(89)} based on observed \& expected loss functions. This covers, e.g., check loss for quantile filtering, squared error for linear filtering, scoring rules, utility functions for economic decision making \& heavy-tailed log-likelihoods. The filter at each time $t$ is computed by minimizing an empirical loss of resampled data. Much of the literature on filtering is recursive, but our approach is not. Even so, using the device we introduce, the filter can be computed across all steps of $t=1,...,T$ in an aggregate $O(T)$ flops.  

A classic simple e.g. is the exponentially weighted moving median (EWMM), which only uses observed losses.  The EWMM of a scalar time series $Y_{1:T}$, is defined as 
$$
\theta_{t}=\underset{\theta}{\arg }\min \ \sum_{j=0}^{t-1} \lambda^{j}\, |Y_{t-j}-\theta|,\quad t=1,...,T.
$$
where $\lambda \in [0,1]$ is a discount rate. 
It is robust to outlying data, yet it has no known recursion \citep{LuxenbergBoyd(24)}, so computing $\theta_{1:T}$ directly costs $O(T^2)$. Our approach is to produce a large number $B$ (e.g. $B=10{,}000$) of random lag lengths $j_1^*,...,j_B^*$ as i.i.d.\ draws with $P(j_b^*=k)\propto \lambda^k$, for $k=0,...,t-1$ (using the inverse of the quantile function of a truncated geometric random variable) \& approximate $\theta_t$ with $\hat{\theta}_t$, the sample median of $Y_{t-j_1^*},....,Y_{t-j_B^*}$. The approximation error $\hat{\theta}_t-\theta_t$ is classically behaved \& can be made arbitrarily small by selecting $B$.
Section \ref{sect:preav} reports $\hat{\theta}_{1:T}$ run over the roughly $T=1.5$ million trades of a single day of a heavily traded financial asset. It has total computational cost of $O(T)$ not $O(T^2)$. 

For first order stationary time series, the EWMM extends to 
$$
\theta_{t}=\underset{\theta}{\arg }\min \ \sum_{j=0}^{t-1} \lambda^{j}\, \{(1-\alpha) \mathbb{E}[|Y_1-\theta|] + \alpha|Y_{t-j}-\theta|\},\quad t=1,...,T,\quad  \alpha \in [0,1],
$$
a convex combination of observed \& expected loss, calling $\alpha$ the ``anchor'', a tuning parameter.  A 2-step approach replaces $\mathbb{E}[|Y_1-\theta|]$ by $\frac{1}{T}\sum_{s=1}^T|Y_s-\theta|$, then draws $j_1^*,...,j_B^*$ as i.i.d.\ from
$$
P(j_b^*=k) = \frac{1-\alpha}{T} + \frac{\alpha \lambda^k}{n_t}\,1(k\ge 0),\quad k=t-T,...,t-1,\quad n_t = \sum_{j=0}^{t-1} \lambda^j, 
$$    
\& approximates $\theta_t$ with $\hat{\theta}_t$, the sample median of $Y_{t-j_1^*},....,Y_{t-j_B^*}$. The 1st term yields a uniform draw from $Y_{1:T}$ that pulls $\hat{\theta}_t$ towards the long-run median; the 2nd is the discounted recent past. Replacing $\frac{1}{T}\sum_{s=1}^T|Y_s-\theta|$ by $\frac{1}{t}\sum_{s=1}^t|Y_s-\theta|$ keeps the filter real time, with $\frac{1-\alpha}{t}$ in place of $\frac{1-\alpha}{T}$ \& $k=0,...,t-1$.

The median is not special. Our algorithm computes, at every observation, filters with losses \& weight functions of the researcher's choosing; before, that was infeasible for long series outside rolling windows of fixed length (uniform weights), linear filters \& exponential family likelihoods. With exponential weights, as in the EWMM, but any convex loss, \cite{LuxenbergBoyd(24)} approximate the estimand sequentially, replacing the loss on the older data by a quadratic term. Others side step the $O(T^2)$ cost by changing the model, e.g. through the use of the dynamic conditional score (DCS) filter \citep{harvey2013dynamic}, the generalized autoregressive score (GAS) filter \citep{CrealKoopmanLucas(13)} \& the implicit score-driven (ISD) filter \citep{lange2024robust}, which track a latent time-varying parameter with the score of an assumed parametric density. More broadly, \cite{PenaYohai(23)} review robust time series methods. 

In our application we keep every trade, so the filter must resist the outliers. The noise of intensively traded stocks is then so heavy tailed that its variance is infinite, yet the classical high frequency volatility estimators require it to be finite \citep{ZhangMyklandAitSahalia(05),BarndorffNielsenHansenLundeShephard(08realised),JacodLiMyklandPodolskijVetter(07)}. Once we robustly filter the outliers using a Huber loss, the resulting integrated \& spot volatility estimators are well behaved.

Section \ref{sect:loss} defines the filtering, smoothing \& prediction estimands for a general loss \& weight function, with the anchor as an option for first order stationary series. Section \ref{sect:simFilter} gives the simulation estimator \& the theory of its error $\theta^*_t-\theta_t$, conditional on the data. In the square loss case the error is unbiased and asymptotically Gaussian in $B$, in the median case we give its exact distribution, and for a general loss standard $M$-estimation theory applies. A stratified version computes the heaviest weights exactly \& samples only the rest. Each $t$ is computed separately, so the sweep over $t=1,...,T$ can run in parallel.

Section \ref{sect:preav} robustly filters the high-frequency price record of an asset traded on a financial market, driving a robust preaveraged estimator of integrated and spot volatility --- following the tradition of \cite{AndersenBollerslevDieboldLabys(01)},  \cite{BarndorffNielsenShephard(02realised)}, \cite{JacodLiMyklandPodolskijVetter(07)},  \cite{MyklandZhang(16)}.  In total we use 5 different loss functions in Section \ref{sect:preav} to tackle various high frequency challenges.  
Section \ref{sect:conc} details some conclusions. The Appendix contains the results from various robustness checks, while a web Appendix reports the corresponding empirical results from another 54 assets.

\section{Filtering via weighted losses}\label{sect:loss}

Think of time series data $y_{1:T} = \{y_1,...,y_T\}$ and a time series model for the corresponding random variables $Y_{1:T} = \{Y_1,...,Y_T\}$.  Here $T$ is the length of the time series.  

\begin{definition}[Loss function] Write a deterministic loss function 
$
L(\theta, y),$ where $\theta \in \Theta{\subseteq \mathbb{R}^d}$, the parameter space, \& $y \in \mathcal{Y}$ the sample space of $Y_t$. $L$ could be, e.g., a loss function, minus a log-likelihood or a scoring rule.  For simplicity, we refer to it as a loss function.  Time-varying losses $L_t(\theta,y)$ appear in Section~\ref{sect:preav-path}.   
\end{definition}

Some assumptions are made about the form of the loss function.

\begin{assumption}\label{assum:start1} (a) $\Theta$ is convex.  (b) $L(\theta, y)$ is convex in $\theta \in \Theta$ for all $y \in \mathcal{Y}$. (c)  $\mathbb{E}[L(\theta,Y_t)]$ exists for all $\theta\in \Theta$. (d)  $\mathbb{E}[L(\theta,Y_t)]$ is strictly convex for all $\theta\in \Theta$.\footnote{It is often technically more convenient to view the loss function as the differences in loss $L(\theta,y)-L(\theta_0,y)$, compared to some benchmark $\theta_0 \in \Theta$ as this difference can often have an expectation, even though the loss itself does not.  The leading case of this using $|y-\theta|-|y|$, not simply $|y-\theta|$, when using absolute loss.  }   
\end{assumption}

\begin{example}[Median] The loss 
$
L(\theta,y) = |y - \theta| - |y|,
$
is convex in $\theta$ and $|L(\theta,y)|\le |\theta|$ by the triangular inequality. 
The $\mathbb{E}[L(\theta,Y_t)]$ exists for every $Y_t$; if $Y_t$ has a positive density everywhere the expected loss is strictly convex, minimized at the unique median of $Y_t$, e.g. \cite{Koenker(05)}.
\end{example}

\subsection{Estimands}

This paper is based on weighted estimands: yielding filters, smoothers \& predictors. 

\begin{definition}[Weighted estimands]\label{def:ewm}
Assumption \ref{assum:start1} holds, the anchor $\alpha \in [0, 1]$, \& $\{w_{t,j}\}_{j=0}^{t-1}$ \& $\{w_{t|T,j}\}_{j=t-T}^{t-1}$ are non-stochastic, non-negative weights at time $t$ \& lag $j$. The time-$t$ filter, predictor \& smoother are 
\begin{equation*}
\theta_{t}=\underset{\theta \in \Theta}{\arg }\min \ Q_{t}(\theta),\quad 
\theta_{t|t-1} = \arg\min_{\theta \in \Theta} Q_{t|t-1}(\theta),\quad 
\theta_{t|T} = \underset{\theta \in \Theta}{\arg }\min \ Q_{t|T}(\theta),
\end{equation*}%
respectively, where 
$$
Q_{t}(\theta) = \sum_{j=0}^{t-1} w_{t,j} \{(1-\alpha)\mathbb{E}[L(\theta, Y_{t-j})] + \alpha L(\theta, Y_{t-j})\},\quad \text{and} \quad \sum_{j=0}^{t-1} w_{t,j}=1,
$$
$$
Q_{t|t-1}(\theta) =  w_{t,0}(1-\alpha)\mathbb{E}[L(\theta, Y_t)] + \sum_{j=1}^{t-1} w_{t,j} \{(1-\alpha)\mathbb{E}[L(\theta, Y_{t-j})] + \alpha L(\theta, Y_{t-j})\},
$$
$$
Q_{t|T}(\theta) = \sum_{j=t-T}^{t-1} w_{t|T,j} \left\{(1-\alpha) \mathbb{E}[L(\theta, Y_{t-j})] + \alpha L(\theta, Y_{t-j})\right\},\quad \text{where} \quad \sum_{j=t-T}^{t-1} w_{t|T,j}=1.
$$
\end{definition}

\begin{remark} 

(a) The $\mathbb{E}[L(\theta,Y_{t-j})]$ are expectations under the time series model for $Y_{1:T}$, so the estimands are functions of the model as well as the data. When $\alpha=1$ they are functions of the data alone, the case of our empirical work in Section \ref{sect:preav}. When $\alpha<1$ the expected losses can be estimated by averages over $Y_{1:T}$ under first-order strict stationarity (Remark \ref{rem:Fprop}(d)). Either way, Section \ref{sect:simFilter} estimates the estimands to an arbitrary level of accuracy controlled by the researcher.
 
(b) Definition~\ref{def:ewm}, Assumption~\ref{assum:start1} \& $\alpha \in [0,1)$ imply that $Q_{t}(\theta)$, $Q_{t|t-1}(\theta)$
\& $Q_{t|T}(\theta)$ are globally strictly convex, \& so $\theta_t$, $\theta_{t|t-1}$ \& $\theta_{t|T}$ are guaranteed to be unique if they exist.  Sufficient for existence is that  $\Theta$ is closed and each criterion is coercive \& lower semicontinuous in $\theta$.
 
(c) The cost of computing $Q_{t}(\theta)$ is $O(t)$, so computing it for an entire sweep over $t=1,...,T$ costs $O(T^2)$.  The cost of computing $Q_{t|T}(\theta)$ is $O(T)$, so the entire sweep costs, again, $O(T^2)$.  For large $T$ this is a problem, which we will solve in a moment.  

(d) Negative weights might  benefit modeling cyclical behavior, but that is ruled out here.  

(e) Flexible weights connect to the linear  distributed lag models \citep{Griliches(67)}, Almon lags \citep{Almon(65)} \& MIDAS regressions \citep{GhyselsSinkoValkanov(07),BaiGhyselsWright(13)}. In particular, with squared loss and covariance stationarity, the time-invariant version of the resulting predictor directly produces a linear distributed lag model.

(f) Write $W_t(k):= \sum_{j=0}^k w_{t,j}$ and $W_t^{-1}(p):= \min(k: W_t(k) \ge p)$ as the $p$-th quantile of the weights of the lags  $0,1,2...$ for $p\in (0,1)$. $W_t^{-1}(1/2)$ is the median lag of the weights. For the exponential weights of Ex.\ref{ex:exponential} it is 1 less than the half-life, the lag at which $\lambda^j$ halves.

(g) When the loss is minus the log-likelihood from an exponential family model with exponential weights, this reproduces the \cite{donkershephard2025CEF} analytic filter.

\end{remark}

\subsection{Models of weights}

Here we discuss four models of weights, which are all used in Section~\ref{sect:preav}.

\begin{example}[Uniform weights] 
A uniform weight, with natural number parameter $\ell$, is   
$$
w_{t,j} = \frac{1(j \le \ell)}{\min(t,1+\ell)},\quad j\in \{0,1,...,t-1\},\quad t\in\{1,...,T\}.
$$
Then $W_t^{-1}(p)=\lceil p\min(t,1+\ell) \rceil-1$. The smoothing version has $w_{t|T,j}\propto I(|j|\le \ell)$.  
Uniform weights are classical, yielding, e.g., rolling quantiles and averages.    
\end{example}

Now focus only on detailing the filtering case, as the smoothing case follows immediately.  
\begin{example}[Exponential weights]\label{ex:exponential} Think of $\lambda \in [0,1)$, a discount rate. Then 
$$
w_{t,j} = \frac{1}{n_t}\lambda^j,\quad \text{where},\quad j\in \{0,1,...,t-1\},\quad n_{t} = \sum_{j=0}^{t-1} \lambda^j = \frac{1-\lambda^t}{1-\lambda}.  
$$
The smoothing version has
$$
w_{t|T,j} = \frac{\lambda^{|j|}}{n_{t|T}},\quad j\in \{t-T,...,t-1\},\quad n_{t|T} = \sum_{j=t-T}^{t-1} \lambda^{|j|} = n_t + n_{T-t+1} -1.
$$
Notice that the 
$$
\sum_{j=0}^k w_{t,j} = \frac{n_{k+1}}{n_t} = \frac{1-\lambda^{k+1}}{1-\lambda^t},\quad k \in [0,1,...,t-1].    
$$
Writing 
$
1+\kappa_p = \frac{\log\{1-p(1-\lambda^t)\}}{\log\lambda},
$
so $W_t^{-1}(p)=\lceil \kappa_p \rceil$.  
Exponential weights are classical, e.g. exponentially weighted moving averages (EWMA).   
\end{example}

\begin{example}[Hyperbolic weights]\label{ex:hyperbolic} A tractable hyperbolic model for weights is $$
w_{t,j} = \frac{1}{n_t}\frac{1}{j+a},\quad j\in \{0,1,...,t-1\},\quad a>0,\quad n_t =\Psi(t+a) -\Psi(a), 
$$
where $\Psi$ is the digamma function (by telescoping, recalling that $\Psi(a+1)-\Psi(a)=1/a$, e.g. (6.3.5) in \cite{AbramowitzStegun(70)}). This is inspired by the  \cite{Zipf(32)}–\cite{Mandelbrot(53)} distribution (taking the power of the denominator to be 1). 
The ``$a$'' term is the ``Mandelbrot parameter.''  For $x\rightarrow \infty$ the $\Psi(x) - \log(x) \rightarrow 0$, so each individual weight will go to zero as $t$ gets large logarithmically.  Notice that the 
$$
\sum_{j=0}^k w_{t,j} = \frac{n_{k+1}}{n_t} = \frac{\Psi(k+1+a)-\Psi(a)}{\Psi(t+a)-\Psi(a)},\quad k \in \{0,1,...,t-1\}.    
$$
Write   
$
1+\kappa_p = \Psi^{-1}\{\Psi(a) + p \times n_t \} -a,
$
so $W_t^{-1}(p)=\lceil \kappa_p \rceil$.  
Hyperbolic decay relates to long memory \citep{Hurst(51),GrangerJoyeux(80)}.  
\end{example}

\begin{example}[Superposition weights]\label{ex:superposition} Flexibility can be built using a convex combination of $P\ge 2$ simpler weighting schemes, with time $t$, lag $j$ weights being $\{w_{t,j,p}\}_{p=1}^P$, yielding 
$$
w_{t,j} = \sum_{p=1}^P \alpha_p w_{t,j,p},\quad \alpha_p\ge 0,\quad \sum_{p=1}^P \alpha_p =1,\quad w_{t,j,p} \ge 0,\quad \sum_{j=0}^{t-1}w_{t,j,p}=1,
$$
where $\alpha_{1:P}$ scales each component \citep{Granger(80),BarndorffNielsenShephard(01jrssb)}.   

An e.g. is 
$w_{t,j,p} = \lambda_p^j/\sum_{i=0}^{t-1} \lambda_p^i,$ with $\lambda_p \in (0,1).$
The $\alpha_{1:P}$ \& $\lambda_{1:P}$ are parameters. 
Mimicking the  \cite{BollerslevHoodHussPedersen(18)} {\tt HExp} model, the $\alpha_{1:P}$ could be taken as free parameters, presetting $\lambda_{1:P}$ to hit various frequencies of interest to the applied researcher (e.g. so they have weights with half lives of, say, 2, 4, 8, 16, 32, 64, etc periods).    Another e.g. holds if each component has uniform weights having  different $\ell$s, which are set a priori, e.g. at 0, 2 and 8, giving blocks of 1, 3 and 9 lags. 
This is inspired by \cite{Corsi(09)}. 
\end{example}

\subsection{The representation}\label{sect:stationary}

Proposition \ref{thm:F} shows that $Q_t$, $Q_{t|t-1}$ \& $Q_{t|T}$ have a common structure.   

\begin{proposition}\label{thm:F} Maintain Assumption \ref{assum:start1}.  Then additionally assume either (i) $\alpha=1$, or (ii) the time series $\{Y_t\}_{t=1}^T$ is first-order strictly stationary. For the prediction case additionally assume $\alpha w_{t,0}<1$. Then the representations 
$$
Q_{t}(\theta) = \int_{\mathcal{Y}} L(\theta,y) {\mathrm d}F_{t}(y),\quad 
Q_{t|t-1}(\theta) = (1-\alpha w_{t,0})  \int_{\mathcal{Y}} L(\theta,y) {\mathrm d}F_{t|t-1}(y),\quad
Q_{t|T}(\theta) = \int_{\mathcal{Y}} L(\theta,y) {\mathrm d}F_{t|T}(y),
$$
hold, writing $\widetilde{n}_t = \sum_{j=1}^{t-1} w_{t,j}=1-w_{t,0}$, where 
\begin{align*}
F_{t}(y) &= (1-\alpha)F_{Y_1}(y) + \alpha \hat{F}_{t}(y),\quad \hat{F}_{t}(y) = \sum_{j=0}^{t-1} w_{t,j} 1(Y_{t-j} \le y), 
\end{align*} and 
\begin{align*}
F_{t|t-1}(y) &= \frac{1}{1-\alpha w_{t,0}}\left\{(1-\alpha)F_{Y_1}(y) + \alpha \widetilde{n}_t \widetilde{F}_{t-1}(y) \right\},\quad \widetilde{F}_{t-1}(y) = \sum_{j=1}^{t-1} \frac{w_{t,j}}{\widetilde{n}_t} 1(Y_{t-j} \le y),   \\
F_{t|T}(y) &= (1-\alpha)F_{Y_1}(y) + \alpha \hat{F}_{t|T}(y),\quad \hat{F}_{t|T}(y) = \sum_{j=t-T}^{t-1} w_{t|T,j} 1(Y_{t-j} \le y).
\end{align*}
If $\widetilde{n}_t=0$ then $F_{t|t-1}=F_{Y_1}$.
\end{proposition}

\begin{proof} Only the prediction case is not immediate.  Under (ii), $\mathbb{E}[L(\theta,Y_{t-j})] = \int_{\mathcal{Y}} L(\theta,y) {\mathrm d}F_{Y_1}(y)$ for every $j$, so, using $w_{t,0}+\widetilde{n}_t=1$,
\begin{align*}
Q_{t|t-1}(\theta) &= (1-\alpha)\int_{\mathcal{Y}} L(\theta,y) {\mathrm d}F_{Y_1}(y) + \alpha \widetilde{n}_t \int_{\mathcal{Y}} L(\theta,y) {\mathrm d}\widetilde{F}_{t-1}(y) = (1-\alpha w_{t,0}) \int_{\mathcal{Y}} L(\theta,y) {\mathrm d}F_{t|t-1}(y).
\end{align*}
Under (i) the same display holds, with the first term vanishing.
\end{proof}

\begin{remark}\label{rem:Fprop}

(a) In all 3 cases the $F_{t}$, $F_{t|t-1}$ and $F_{t|T}$ are convex combinations of $F_{Y_1}$ and a weighted empirical CDF: $\hat{F}_{t}$, $\widetilde{F}_{t-1}$ and $\hat{F}_{t|T}$.

(b) Notice $F_{t}$, $F_{t|t-1}$ \& $F_{t|T}$ do not depend upon the loss function.  

(c) Focus on filtering to illustrate. 
The $\int_{\mathcal{Y}} L(\theta,y) {\mathrm d}F_{Y_1}$ does not involve $t$, unlike 
$$
\int_{\mathcal{Y}} L(\theta,y) {\mathrm d}\hat{F}_{t} = \sum_{j=0}^{t-1} w_{t,j} L(\theta,Y_{t-j}).
$$
More broadly, if the loss depends upon $t$, then 
$
\mathbb{E}_j[L_{t-j}(\theta,Y_{t-j})] = \sum_{j=0}^{t-1} w_{t,j} L_{t-j}(\theta,Y_{t-j}),
$
holds, where the expectation treats $j$ as random, while $Y_{1:t}$ is fixed with  $Pr(j=k)=w_{t,k}$.  All the methods developed later also cover this case, as each sampled lag $j_b$ evaluates its own $L_{t-j_b}$, and carry over the prediction and smoothing.

(d) There are two leading cases: (i) when $\alpha=1$ (no anchoring), (ii) ``first-order strict stationarity'' (when the distribution of $Y_t$ is time invariant).  Assuming first order stationarity, an approach to approximating $F_{Y_1}$ is to use the ECDF of $Y_{1:T}$, i.e. $\hat{F}_{Y_1}(y) = \frac{1}{T} \sum_{j=1}^T 1(Y_j \le y).$  Researchers often use historic data to estimate unknown parameters, in a 2 step estimation procedure \citep[e.g.][]{NeweyMcFadden(94),EngleMezrich(96),FrancqHorvathZakoian(13)}. In our context the variant is to use the entire ECDF. 

(e) When the loss is squared \& $\theta_t=\alpha\sum_{j}w_{t,j}Y_{t-j}+(1-\alpha)\mathbb{E}[Y_1]$, the filter with exponential weights is the steady-state Kalman filter of a Gaussian AR(1) signal with root $\phi$ \& white noise, with $\lambda=\phi(1-K)$, $\alpha=K/(1-\lambda)$ \& Kalman gain $K$.  The anchor is the mean reversion.

(f) Focus on the exponential weight with infinite past, future  and $\lambda \in (0,1).$  Then 
$$
\hat{F}_{t}(y) = (1-\lambda) \sum_{j=0}^{\infty} \lambda^{j} 1(Y_{t-j} \le y),\quad \hat{F}_{t|T}(y) = \frac{1-\lambda}{1+\lambda} \sum_{j=-\infty}^{\infty} \lambda^{|j|} 1(Y_{t-j} \le y),
$$
while 
$
\widetilde{F}_{t-1}(y) = (1-\lambda) \sum_{j=0}^{\infty} \lambda^{j} 1(Y_{t-1-j} \le y) = \hat{F}_{t-1}(y).
$

(f) Prediction weights sum to 1, as $\widetilde{n}_t = 1-w_{t,0}$ so 
$
\{(1-\alpha) + \alpha\widetilde{n}_t\}/(1-\alpha w_{t,0}) =1.
$

(g) The condition $\alpha w_{t,0}<1$ excludes only the combination $\alpha=1$ \& $w_{t,0}=1$, where $Q_{t|t-1}$ is $0$ and $\theta_{t|t-1}$ is undefined. Since $w_{1,0}=1$, at $t=1$ the prediction case requires $\alpha<1$.

\end{remark}

\begin{example}[Weighted median] Think of $\alpha=1$ and the loss being $|y-\theta|-|y|$, then  
$$
\int_{\mathcal{Y}} L(\theta,y) {\mathrm d}\hat{F}_{t} = \sum_{j=0}^{t-1} w_{t,j} (|Y_{t-j}-\theta| - |Y_{t-j}|),
$$
implying $\theta_t$ is a weighted median.
Weighted medians have a long history in statistics (\citealp{Edgeworth(1888)}; \citealp{Tukey(71)}; {\tt spatstat} in {\tt R}) as does median filtering which uses uniform weights \citep{King(1924),Arce(05),FriedEinbeckGather(07),AriasCastroDonoho(09)}. It is simple to state, but for general weight functions $\theta_t$ costs $O(t)$ to compute using {\tt QuickSelect} \citep{RauhArce(12),tibshirani2008fast} or $O(T^2)$ for a full sweep. For large $T$ this is unacceptable.  Various approximations have been proposed, such as averaging medians of subsamples \citep{Tukey(78)} or the remedian of \cite{RousseeuwBassett(90)}.    
\end{example}

That the full sweep of the filter costs $O(T^2)$ is an endemic problem.  

\subsection{Example: rolling median}\label{sect:uniNonstat}

To set up the next section, focus on the uniform weight filtering case, with $\ell$ lags \& $L(\theta,y)=|y-\theta|-|y|$.  Thus $\theta_t$ is the median of $Y_{t-\ell:t}$ --- a rolling median.  

\begin{example}[Signal plus Cauchy noise]\label{ex:running}
Generate a signal, $\mu_{1:T}$, and data $Y_{1:T}$, as
$$
Y_t = \mu_t + \epsilon_t,\quad \epsilon_t \overset{iid}{\sim } Cauchy(0,1),\quad \eta_t \overset{iid}{\sim } N(0,1),\quad \epsilon_{1:T} \ind \eta_{1:T},\quad t=1,...,T,
$$
where the signal is either (a) nonstationary: $\mu_{t+1} = \mu_t + 0.2\, \eta_t$, with $\mu_1=0$, or (b) stationary: $\mu_{t+1} = \phi \mu_t + \sqrt{1-\phi^2}\, \eta_t$, $\phi=0.95$ \& $\mu_1 \sim N(0,1)$, so $\mu_t \sim N(0,1)$ for each $t$.  The $Cauchy(0,1)$ has location $0$ \& scale $1$.  $Y_{1:T}$ has no moments, implying linear filters are useless.  
\end{example}

\subsubsection{Nonstationary case}

Focus on Ex.\ref{ex:running}(a) \& set $\alpha=1$. 
For $t\ge \ell+1$ \& $\ell$ even, then $\theta_t$ is unique.  But if $\ell$ is odd then the minimizer lives in $[\underline{\theta}_t,\overline{\theta}_t]$, where 
$
\underline{\theta}_t = Y_{[t,(\ell+1)/2 ]}$,  
$\overline{\theta}_t = Y_{[t,(\ell+3)/2 ]}
$ 
\& $Y_{[t,1]}\le ...\le Y_{[t,\ell+1]}$ are ranked $Y_{t-\ell:t}$. 

\begin{figure}[htbp]
    \centering
    \begin{tabular}{@{}c@{\hspace{-0.25cm}}c@{\hspace{-0.25cm}}c@{\hspace{-0.25cm}}c@{}}
        & \textbf{$Y_t,\theta_t$} & \textbf{$\mu_t,\theta_t$} & \text{Average loss or error} \\ \vspace{-5mm}
        \raisebox{2.0cm}[0pt][0pt]{{\small $\ell=1$\hspace{0.3cm}}} &
        \includegraphics[width=0.25\linewidth]{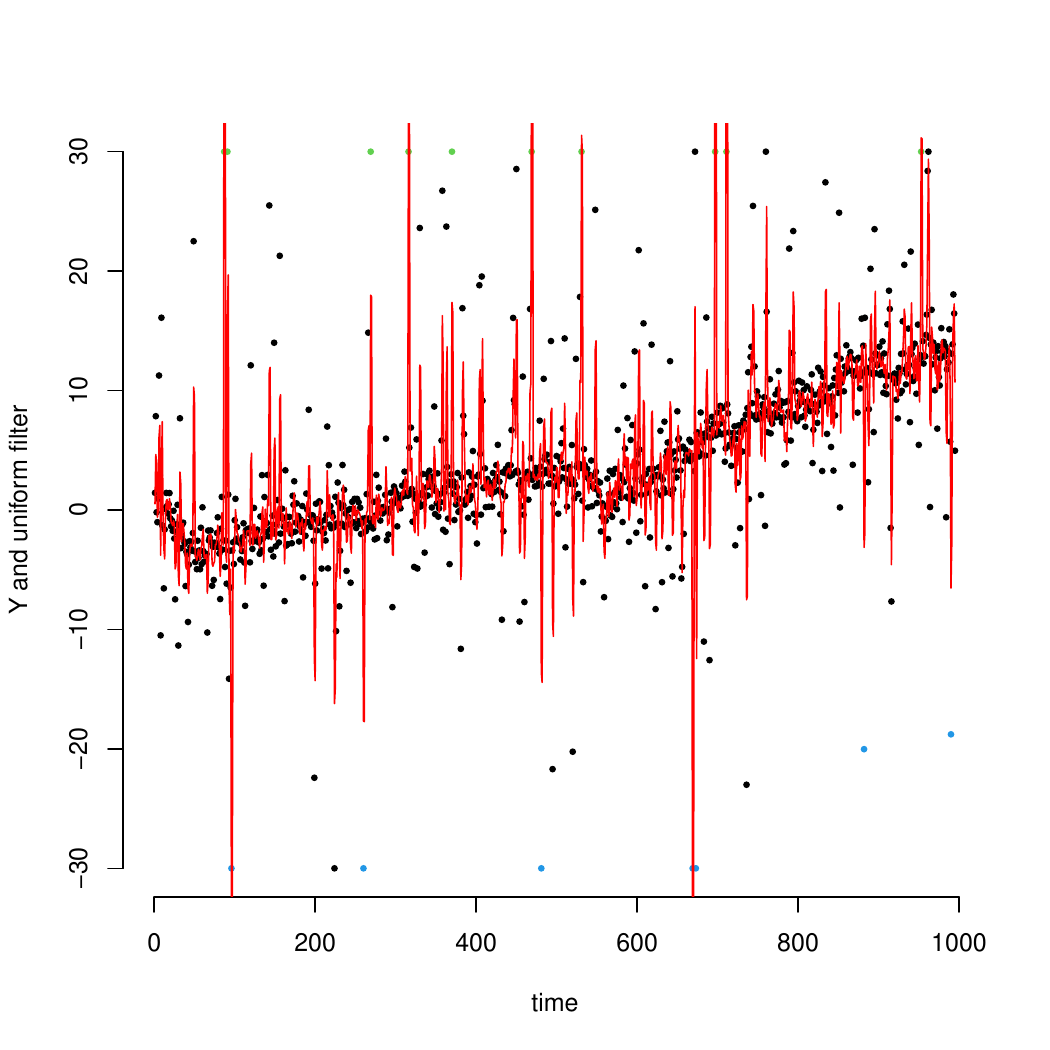} &
        \includegraphics[width=0.25\linewidth]{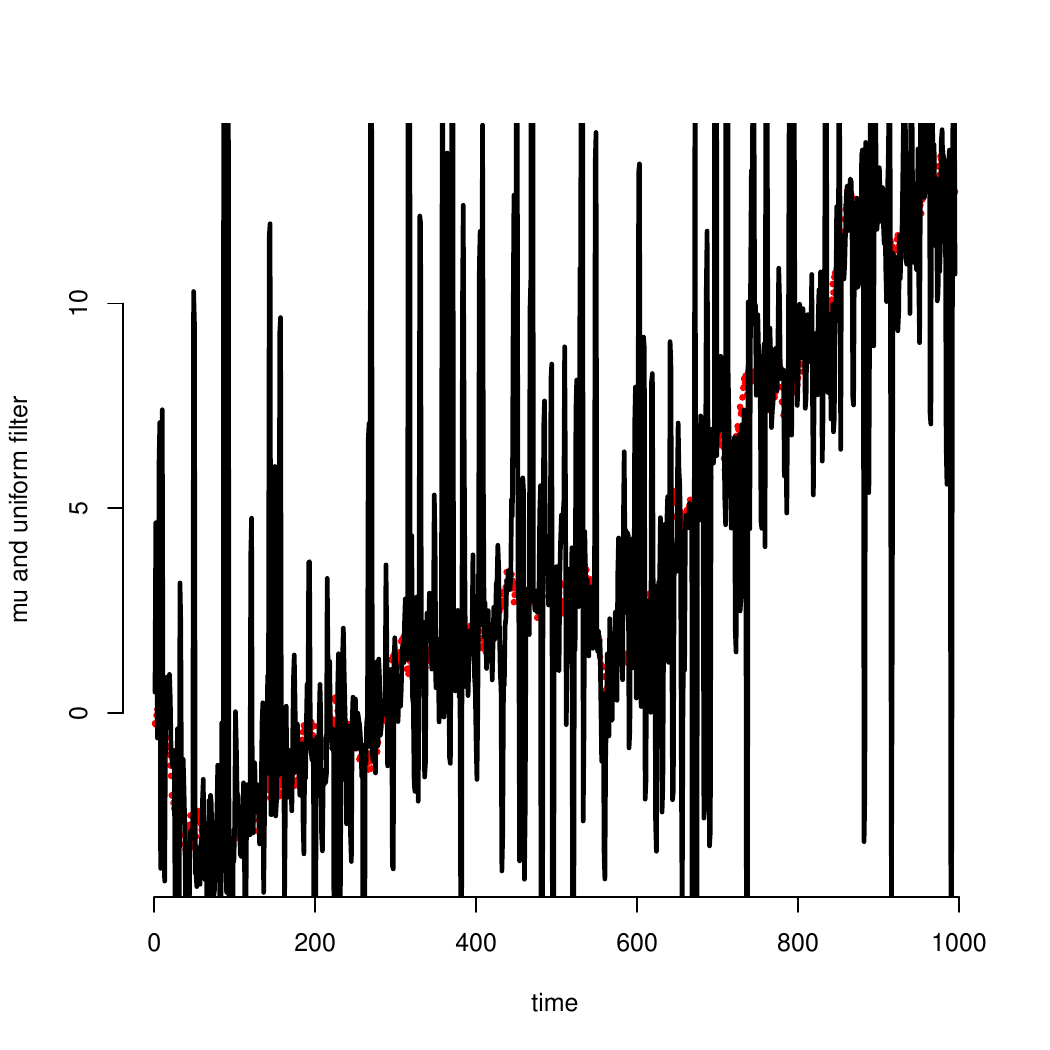} &
        \includegraphics[width=0.25\linewidth]{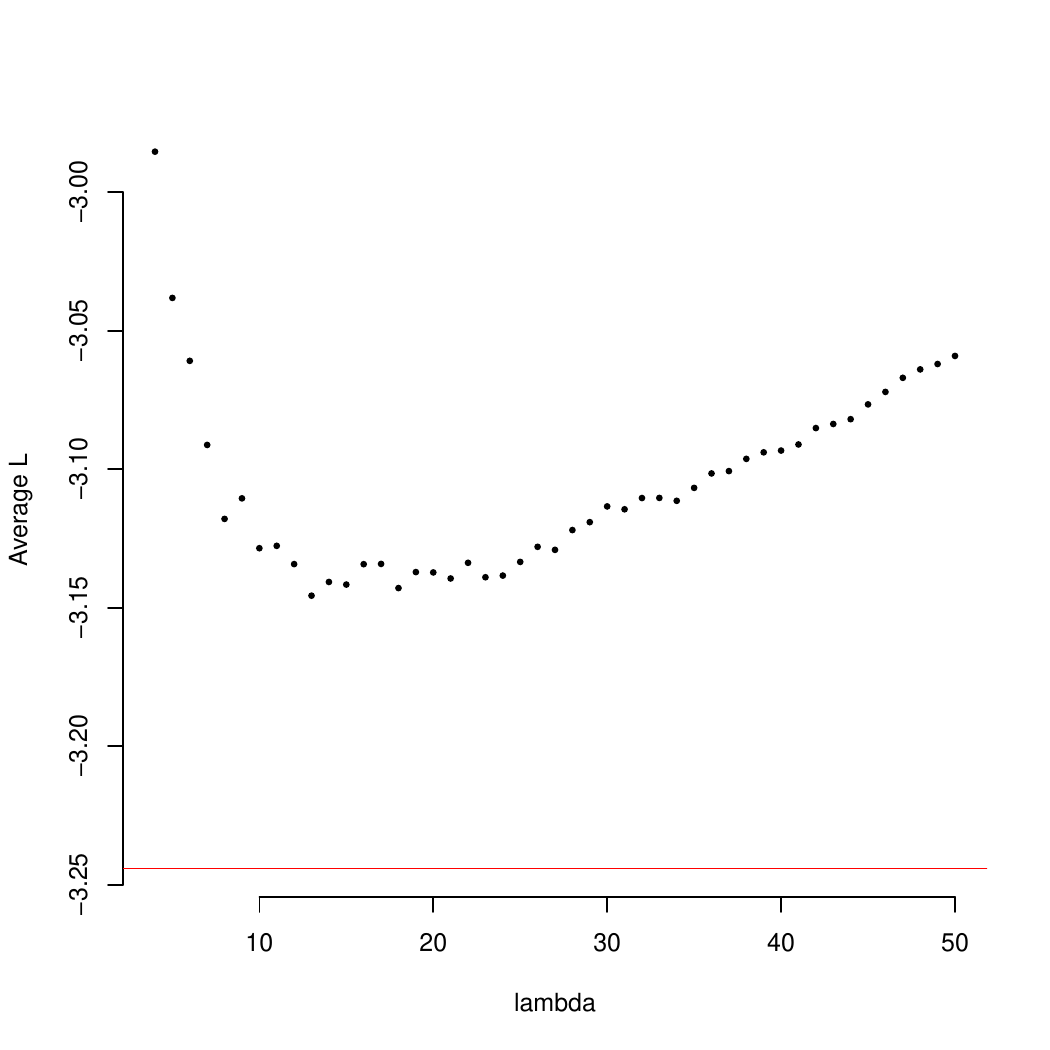} \\ \vspace{-5mm}
        \raisebox{2.0cm}[0pt][0pt]{{\small $\ell=7$\hspace{0.3cm}}} &
        \includegraphics[width=0.25\linewidth]{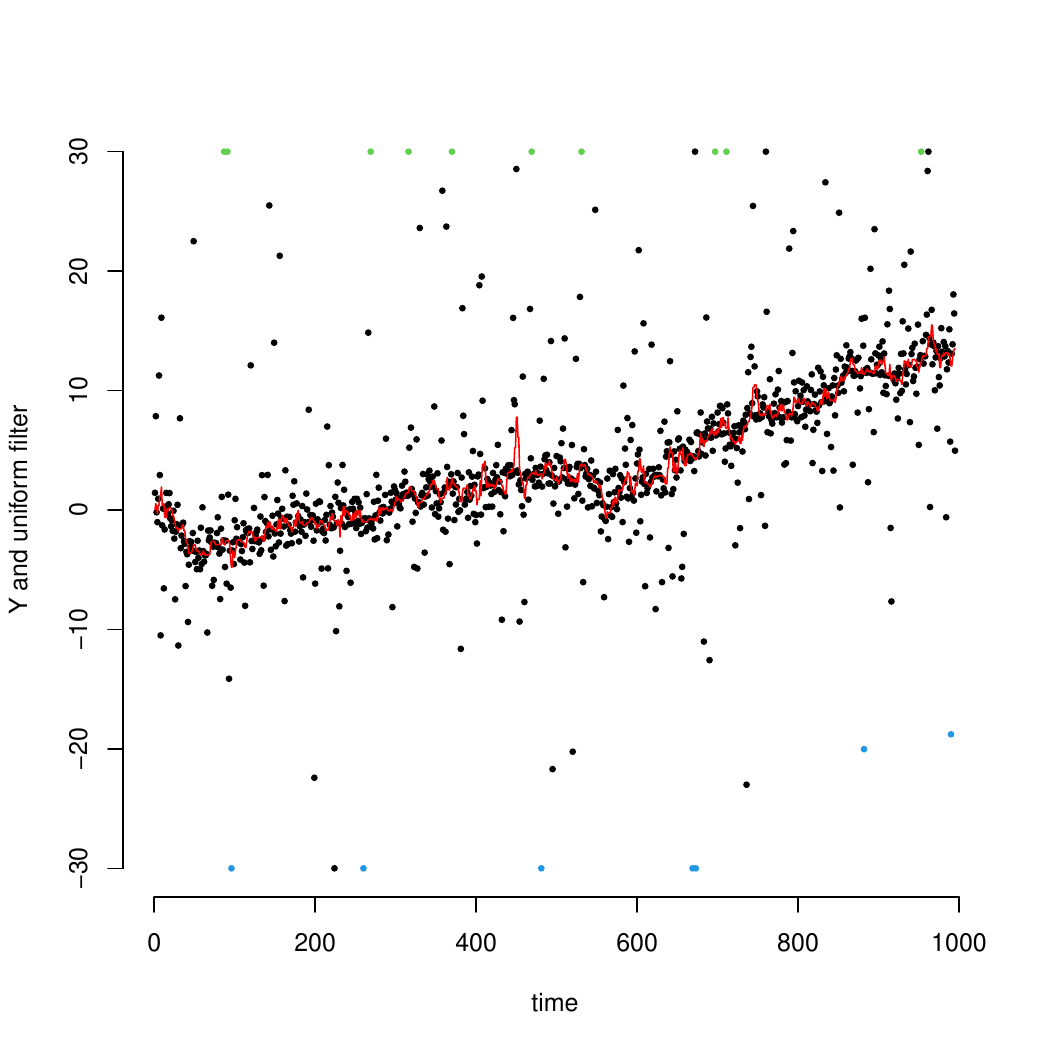} &
        \includegraphics[width=0.25\linewidth]{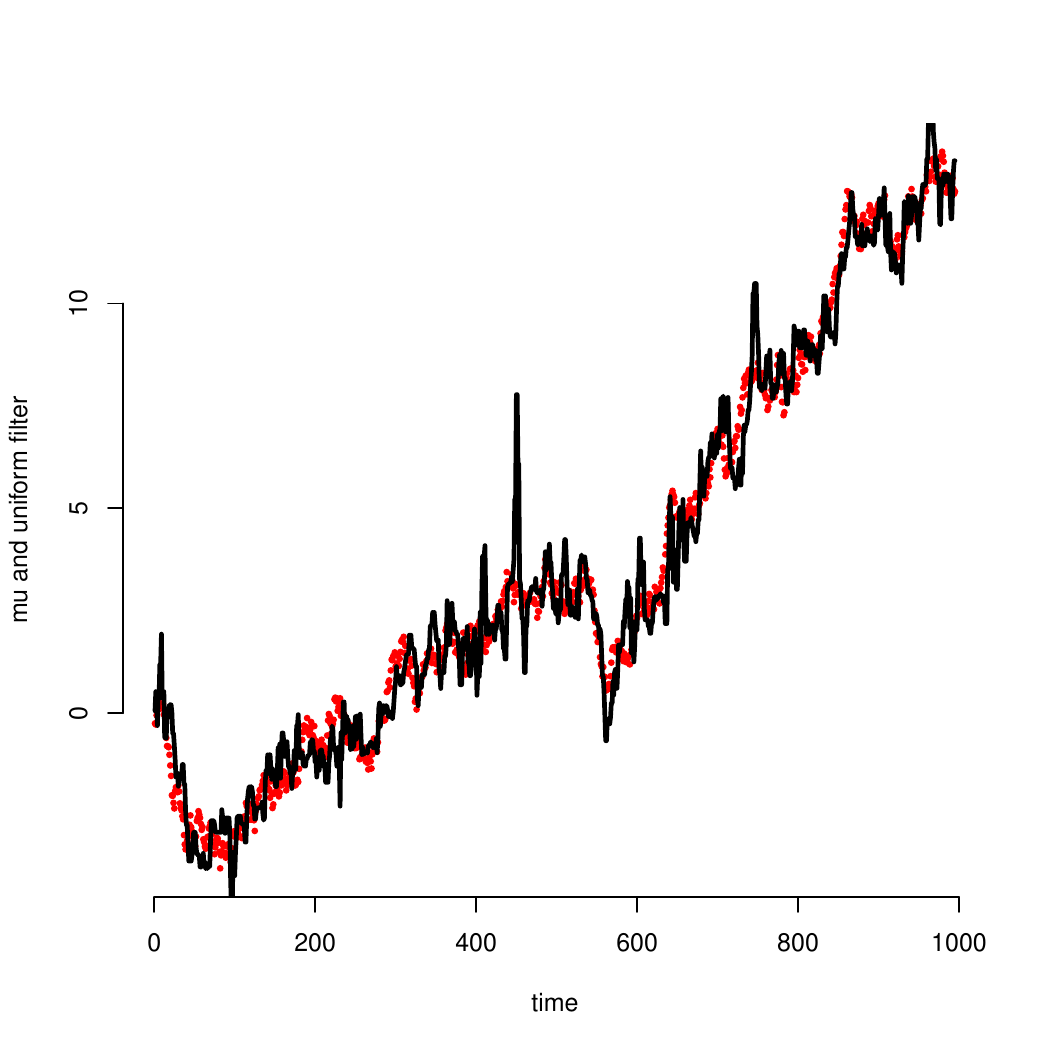} &
        \includegraphics[width=0.25\linewidth]{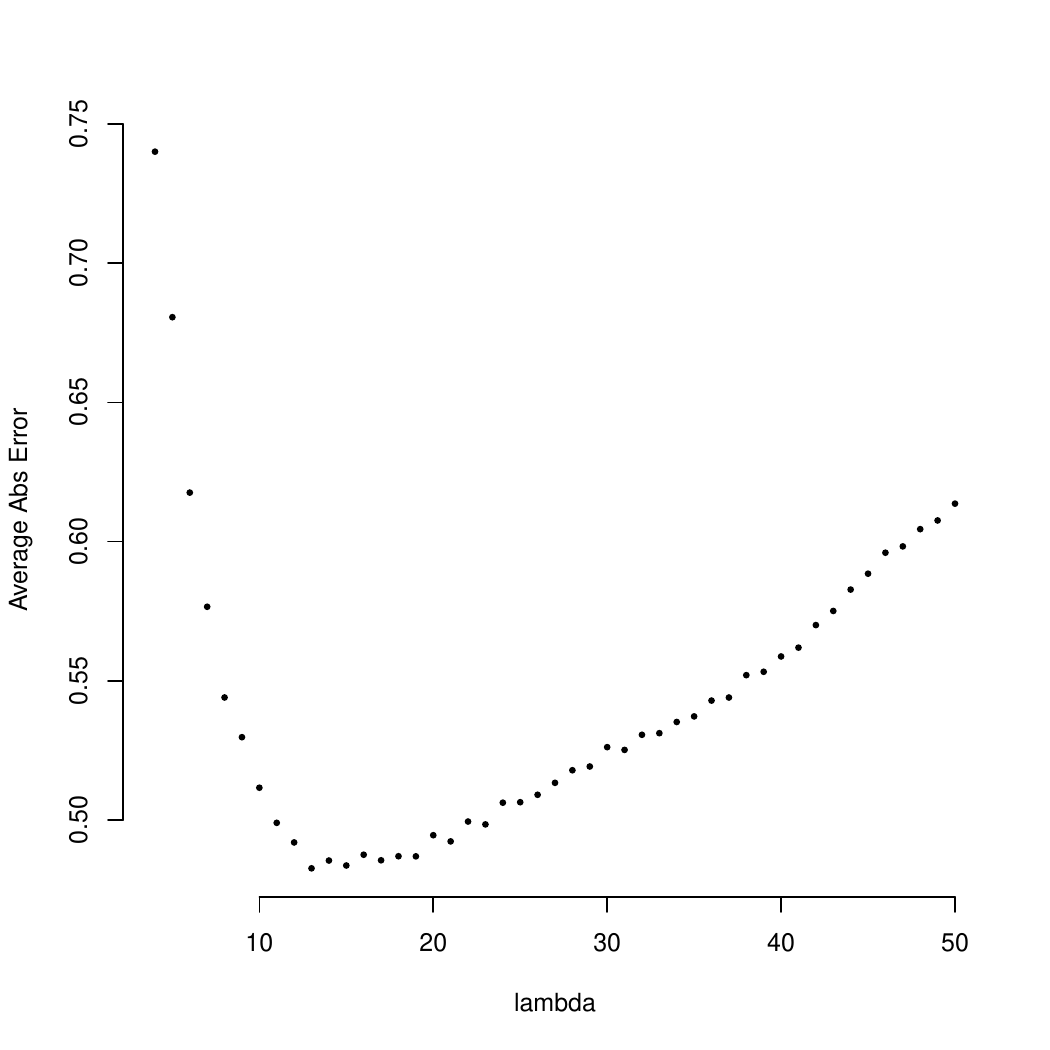}
        \\ 
        \raisebox{2.0cm}[0pt][0pt]{{\small $\ell=15$\hspace{0.25cm}}} &
        \includegraphics[width=0.25\linewidth]{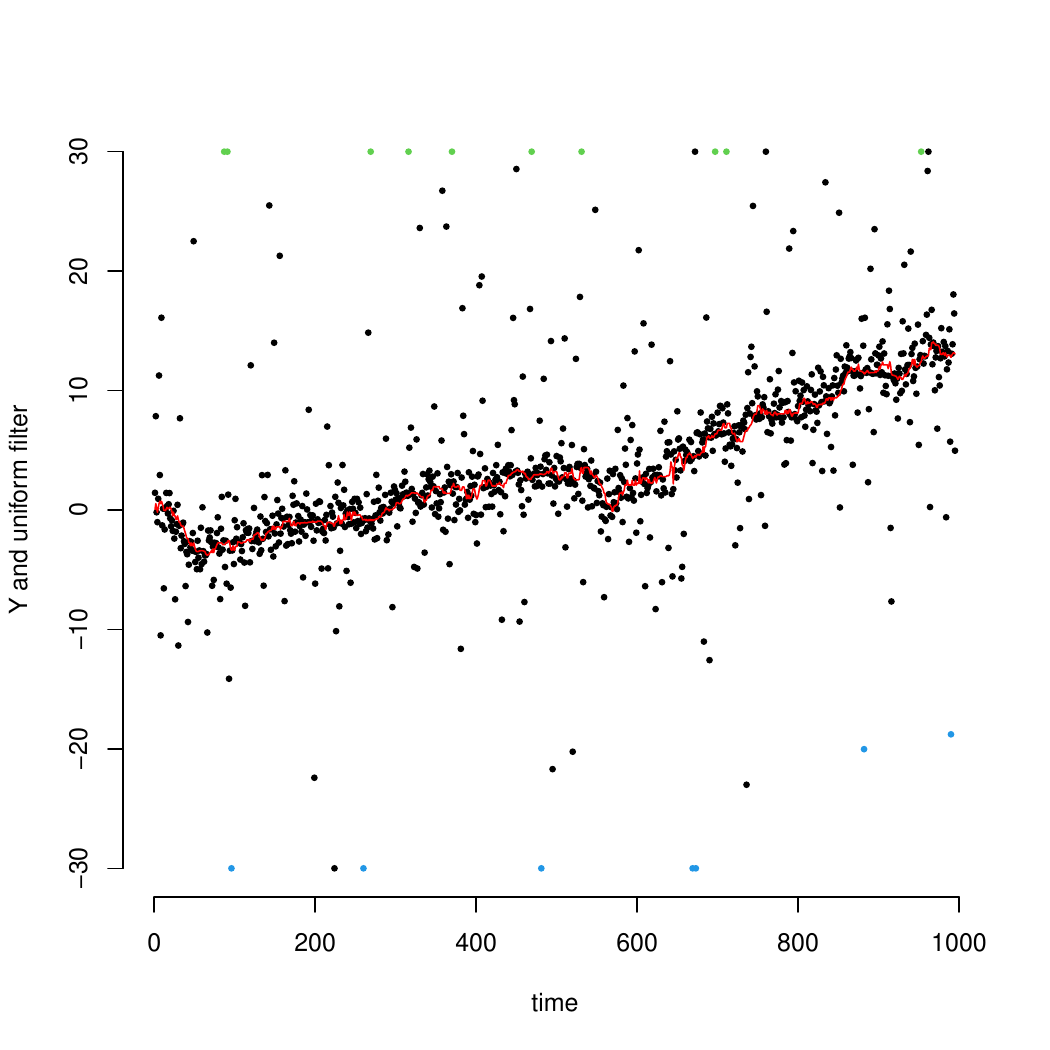} &
        \includegraphics[width=0.25\linewidth]{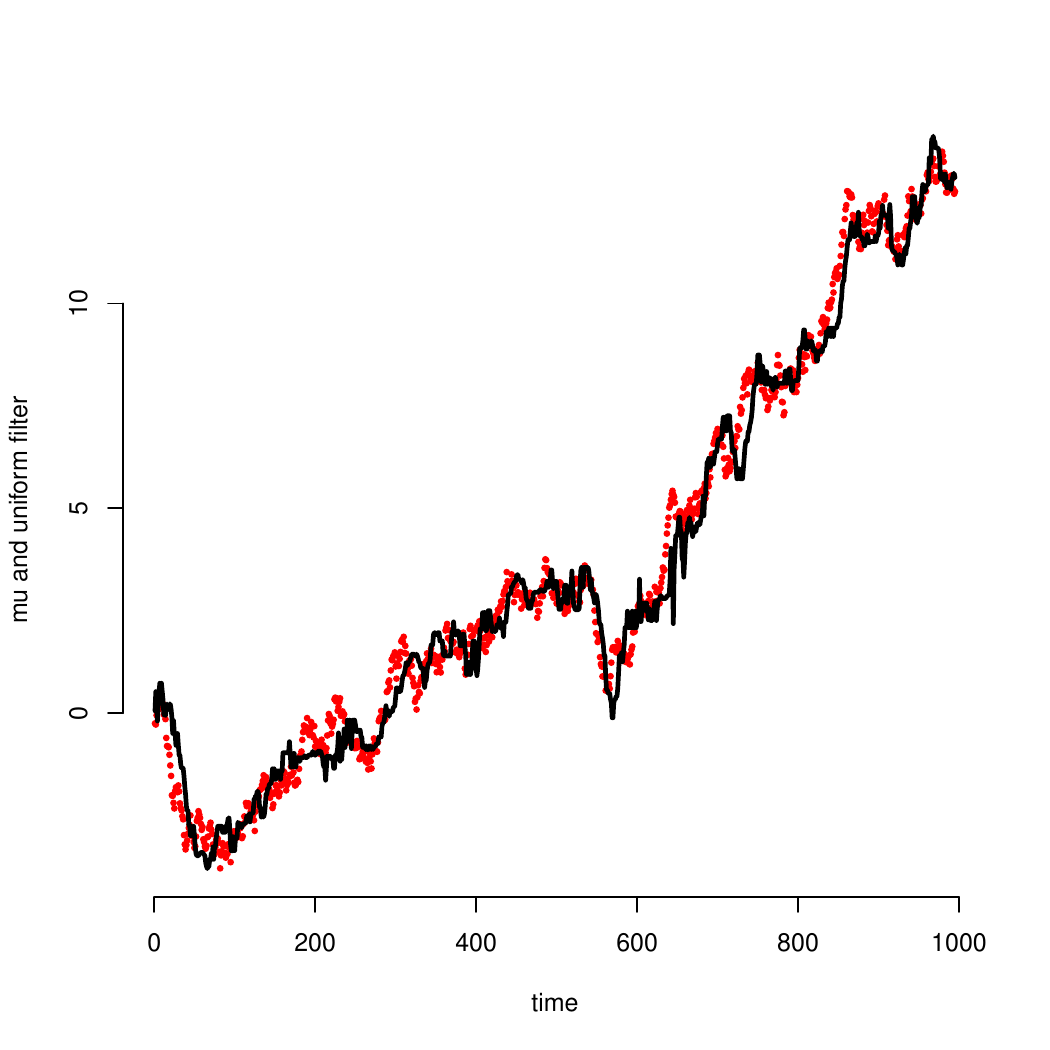} &
    \end{tabular}
    \vspace{-5mm}
    \caption{Nonstationary case.  LHS:  data $Y_t$ (dots) \& rolling filter $\theta_t$ (red line) based on $Y_{t-\ell:t}$ plotted against time.  Middle: signal $\mu_t$ (red dots) \& rolling filter $\theta_t$.  RHS: top shows average loss $|Y_t-\theta_{t-1}|-|Y_t|$ against $\ell\ge 4$ \& middle shows average $|\mu_t-\theta_t|$ error. } 
    \label{fig:medianOld}
\end{figure}  

\begin{remark}\label{remark:median}  \cite{HyndmanFan(96)} document 9 different medians.  We use the \cite{Hazen(1914)} approach, the type 5 median of \cite{HyndmanFan(96)}.  It extends to the weighted version and is stable with respect to small changes in weights.  For the filter it is the interpolated 
$$
\theta_t = \frac{\overline{F}_t-0.5}{\overline{F}_t-\underline{F}_t} \underline{\theta}_t + \frac{0.5-\underline{F}_t}{\overline{F}_t-\underline{F}_t}\overline{\theta}_t,\quad \text{where} \quad \overline{F}_t = \frac{1}{2}\left\{F_t(\overline{\theta}_t)+F_t(\overline{\theta}_t-)\right\},\ 
\underline{F}_t = \frac{1}{2}\left\{F_t(\underline{\theta}_t)+F_t(\underline{\theta}_t-)\right\}.
$$
If $F_t$ is the ECDF of $Y_{t-\ell:t}$ \& $\ell$ is odd then $\theta_t = (\underline{\theta}_t+\overline{\theta}_t)/2.$  We use the type 5 {\tt weighted.median} function from the {\tt R} package {\tt spatstat} of Adrian Baddeley. We knock out any 0 weight datapoints \& implement the ``collapse'' option (for repeat data points we sum the weights \& replace the repeats by a single upweighted data point) as these changes have no impact on $F_t$.     
\end{remark}

\begin{figure}[htbp]
    \centering
    \begin{tabular}{@{}c@{\hspace{-0.25cm}}c@{\hspace{-0.25cm}}c@{\hspace{-0.25cm}}c@{}}
        & \textbf{$Y_t,\theta_t$} & \textbf{$\mu_t,\theta_t$} & \text{Average loss or error} \\ \vspace{-5mm}
        \raisebox{2.0cm}[0pt][0pt]{{\small $\ell=1$\hspace{0.3cm}}} &
        \includegraphics[width=0.25\linewidth]{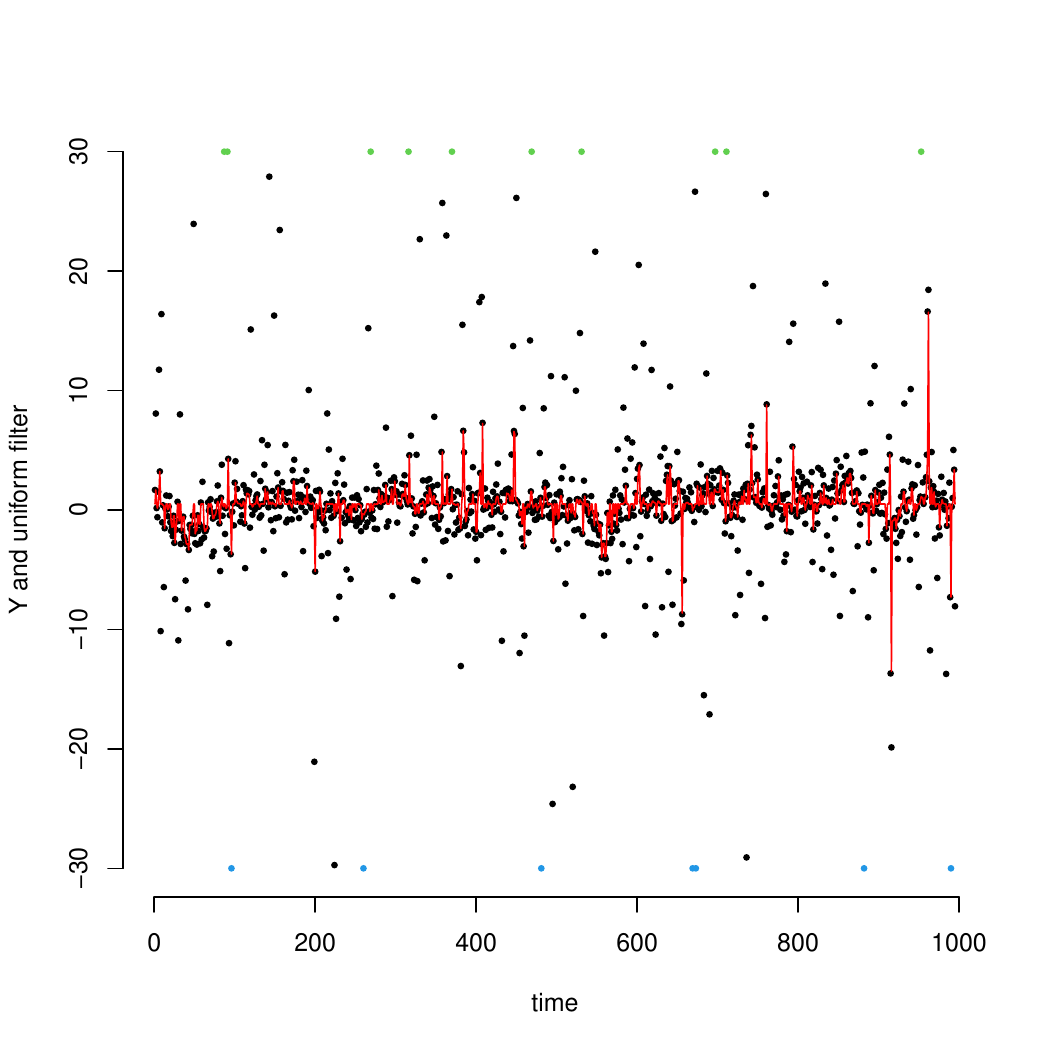} &
        \includegraphics[width=0.25\linewidth]{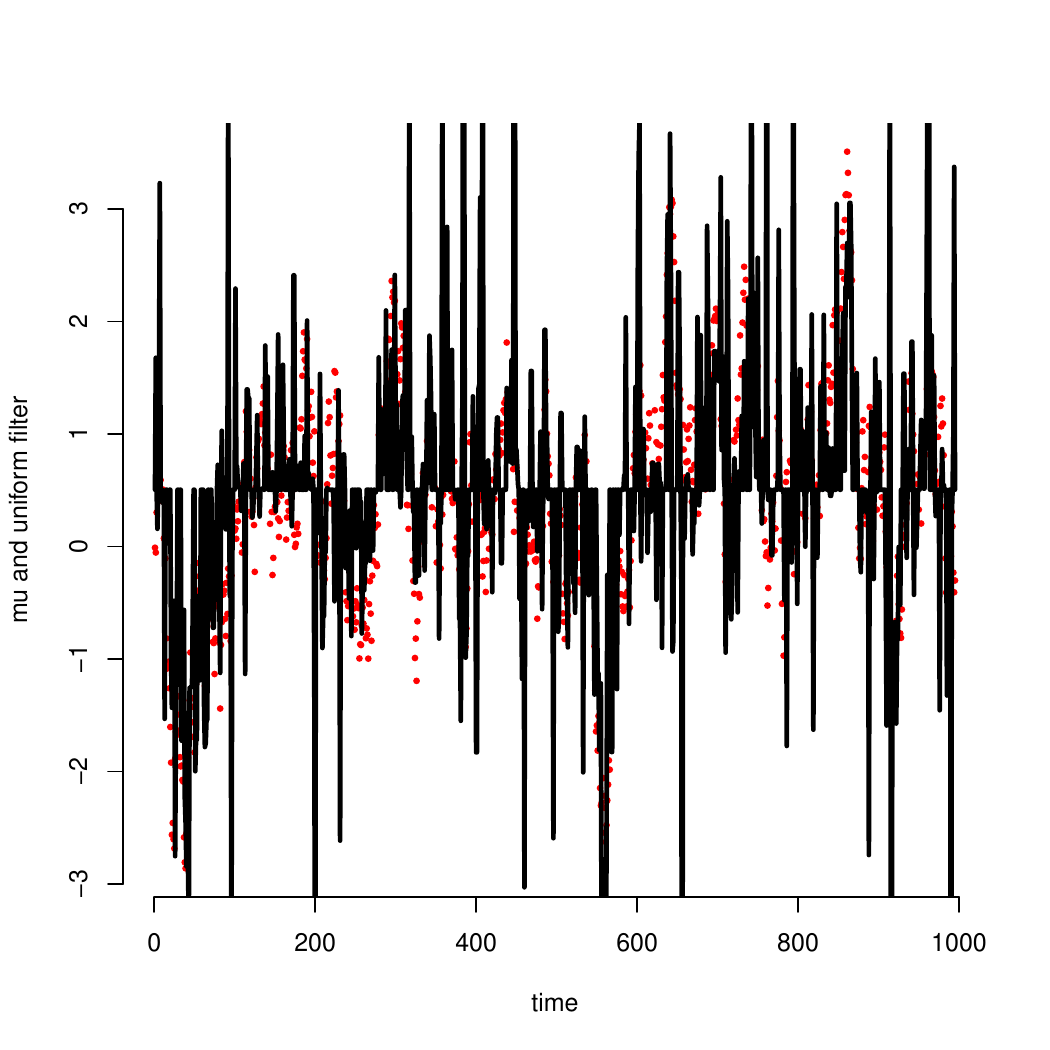} &
        \includegraphics[width=0.25\linewidth]{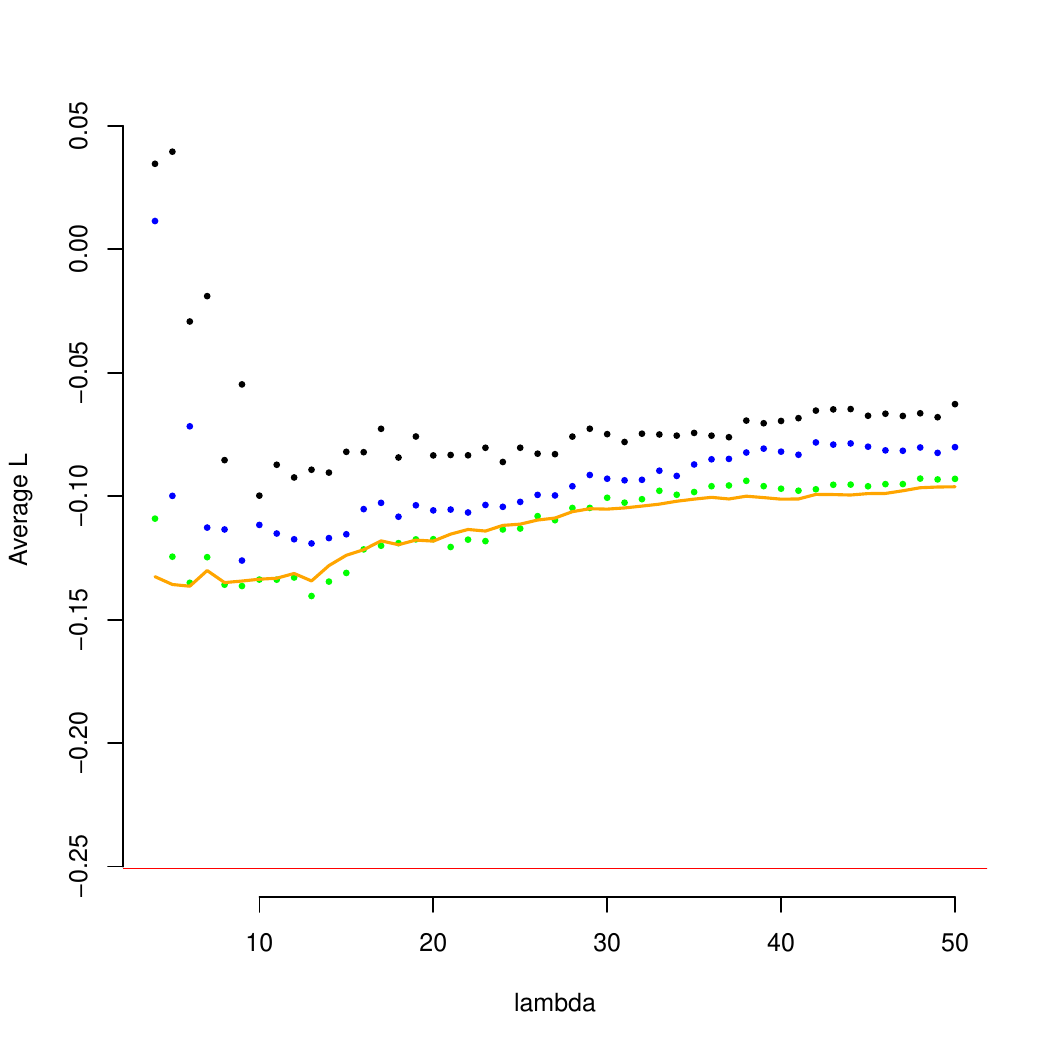} \\ \vspace{-5mm}
        \raisebox{2.0cm}[0pt][0pt]{{\small $\ell=7$\hspace{0.3cm}}} &
        \includegraphics[width=0.25\linewidth]{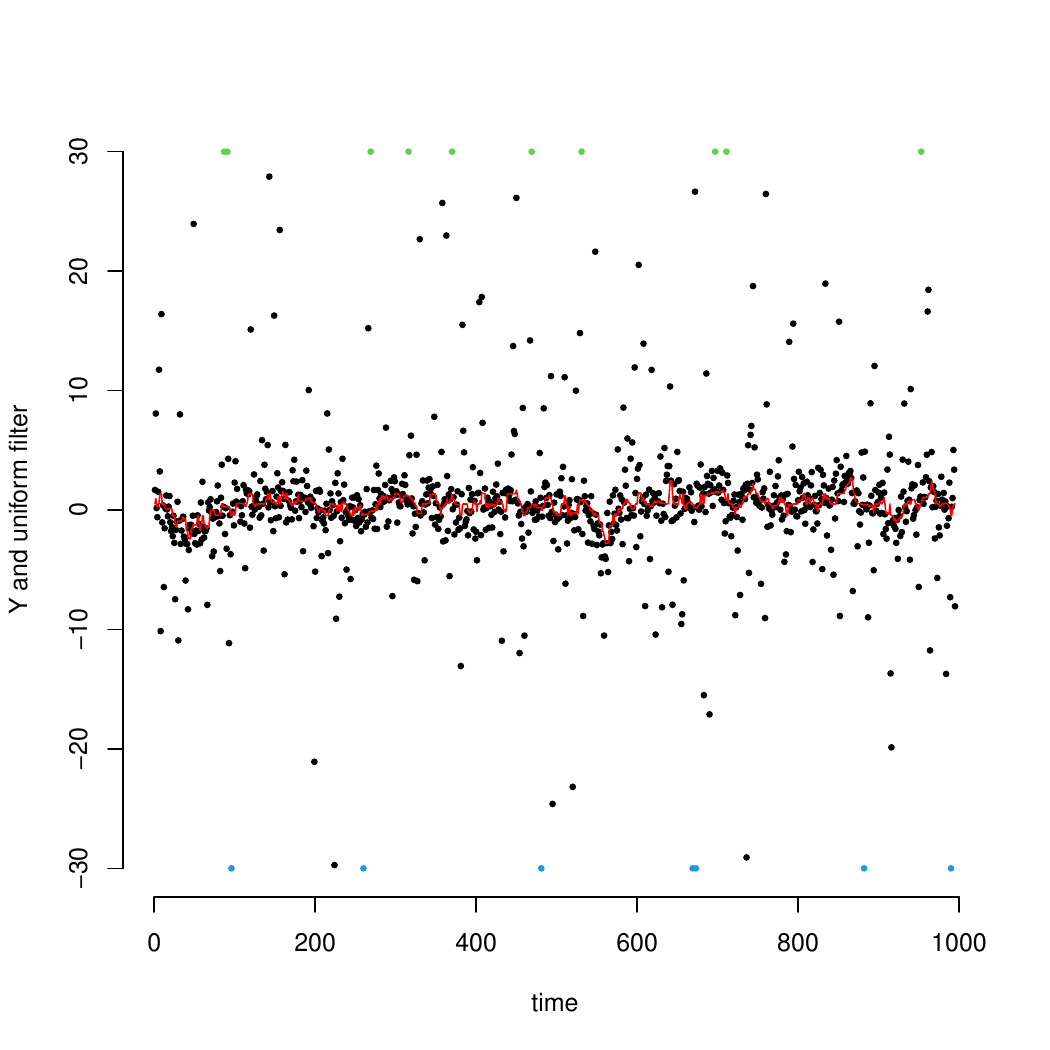} &
        \includegraphics[width=0.25\linewidth]{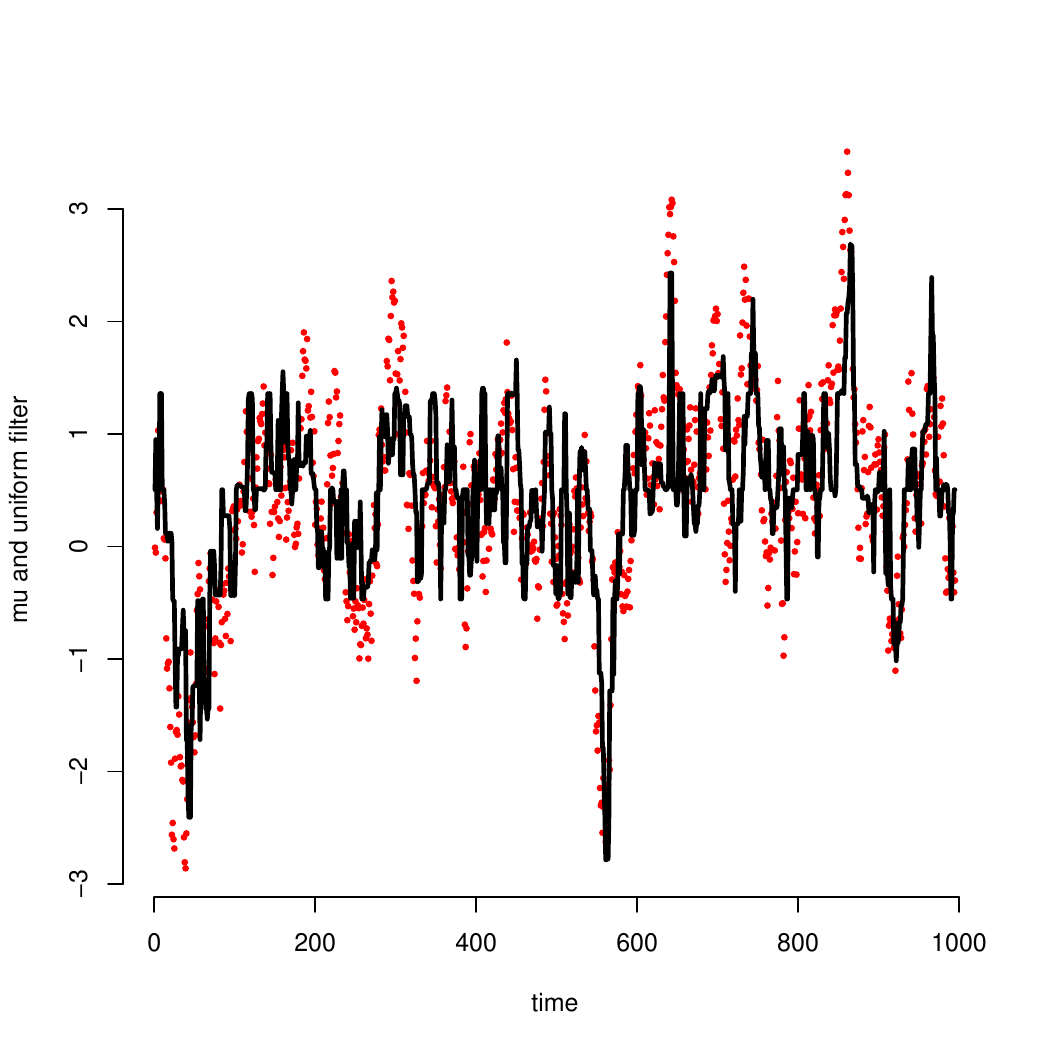} &
        \includegraphics[width=0.25\linewidth]{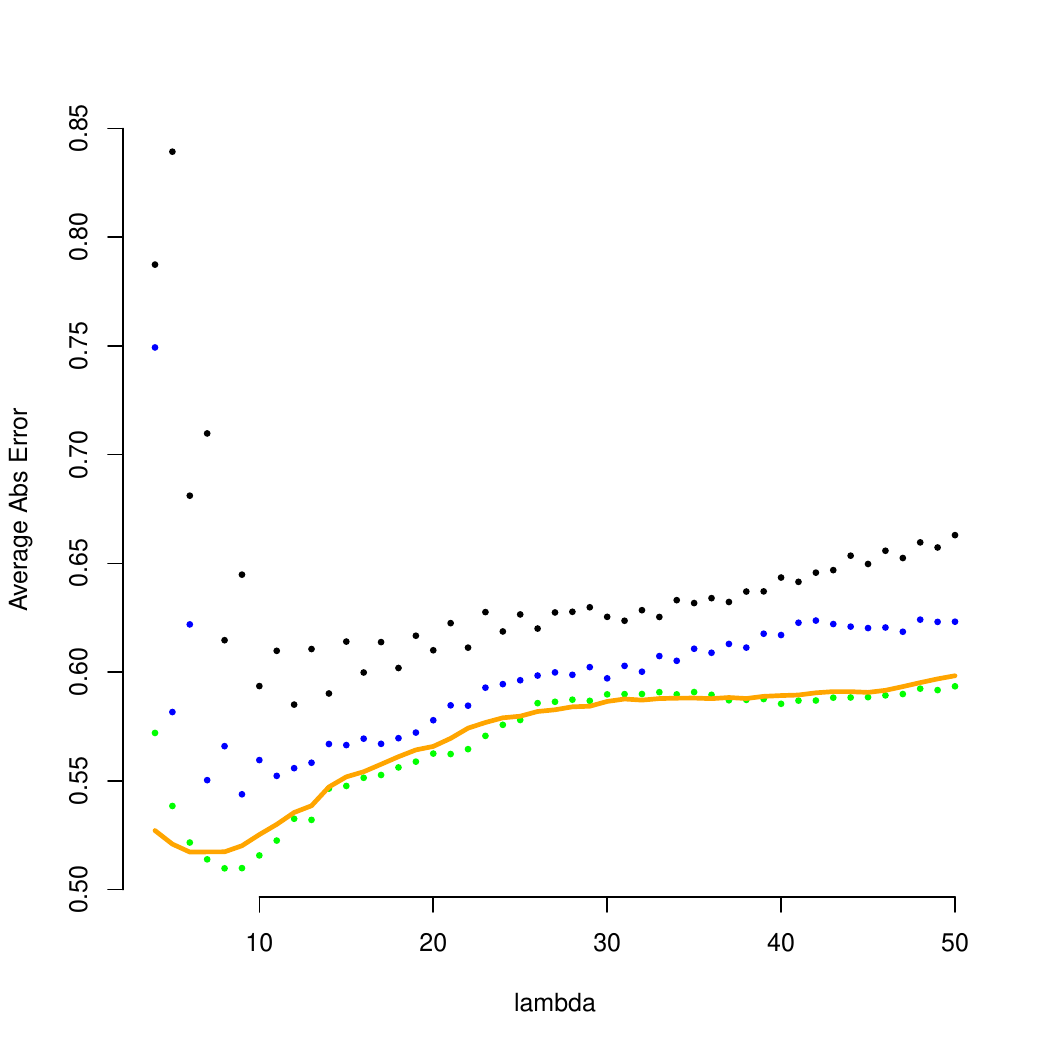}
        \\
        \raisebox{2.0cm}[0pt][0pt]{{\small $\ell=15$\hspace{0.3cm}}} &
        \includegraphics[width=0.25\linewidth]{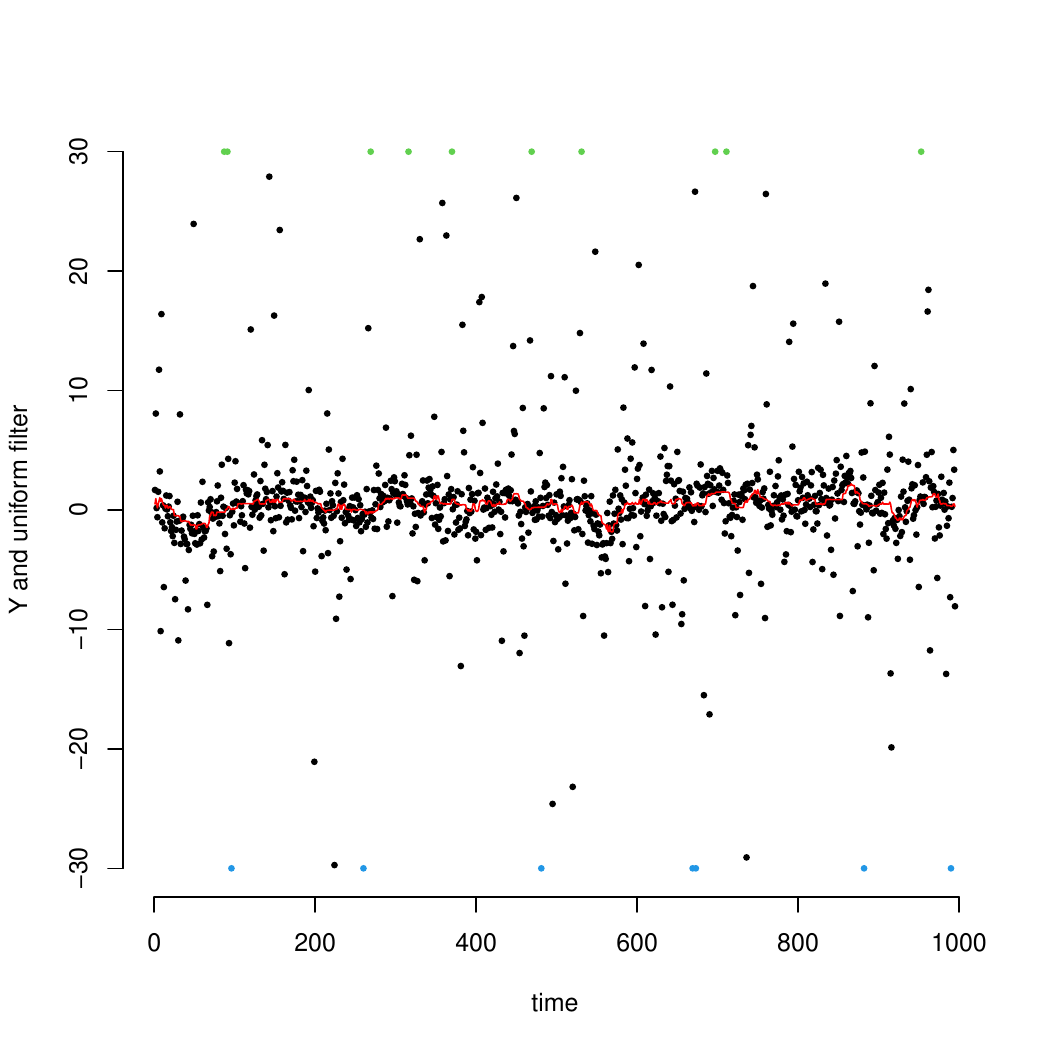} &
        \includegraphics[width=0.25\linewidth]{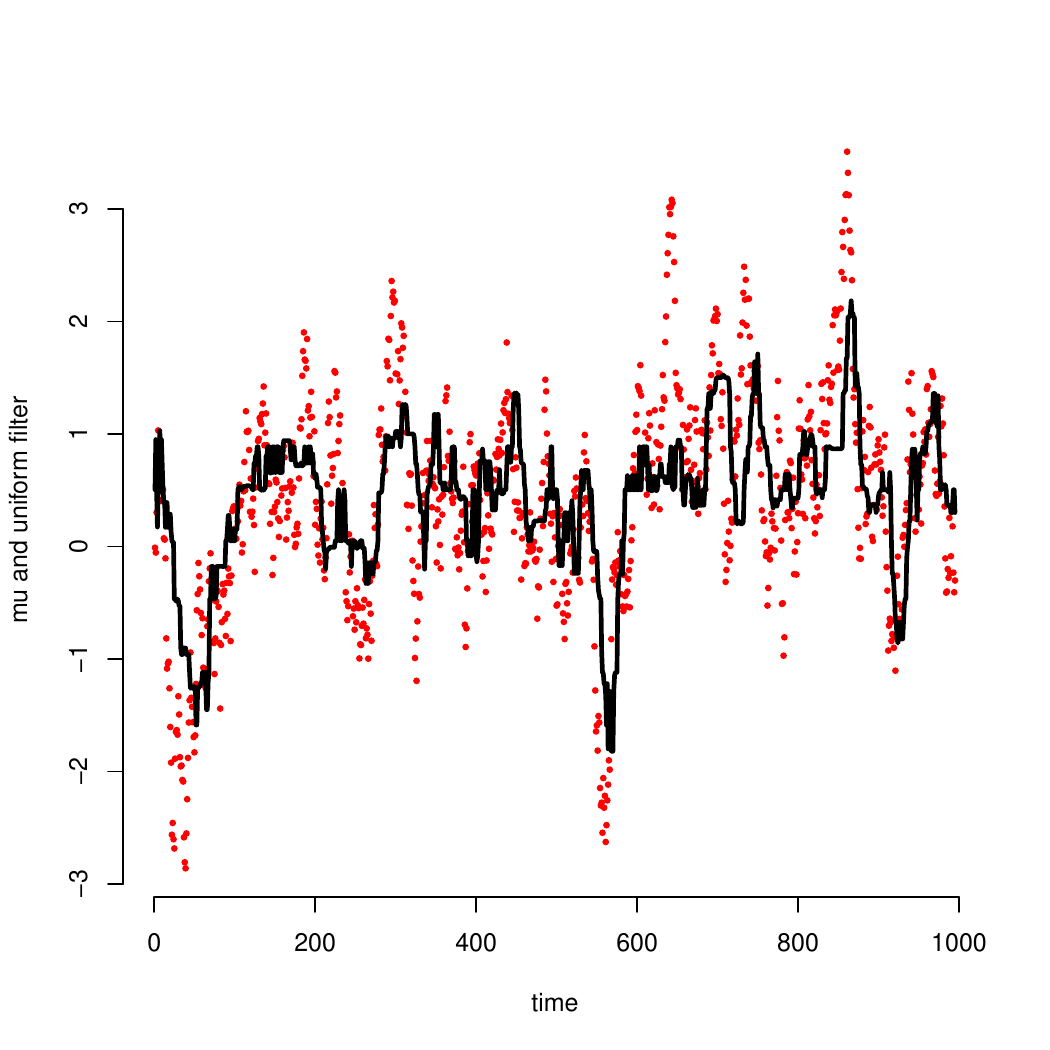} &
    \end{tabular}
    \vspace{-5mm}
    \caption{Stationary case.  LHS: $Y_t$ (dots) \& rolling filter $\theta_t$ (red line) based on $Y_{t-\ell:t}$ plotted against time.  Middle: signal $\mu_t$ (red dots) \& rolling filter $\theta_t$.  RHS: top shows average loss $|Y_t-\theta_{t-1}|-|Y_t|$ against window length $\ell\ge 4$ and middle shows average $|\mu_t-\theta_t|$ estimation error. $\alpha\in \{0.4,0.6,0.8,1.0\}$ shown using black, blue, green (dots) \& orange (line). } 
    \label{fig:medianOldStat}
\end{figure}

To illustrate all of this, the left of Figure \ref{fig:medianOld} shows the sample path of the data and filter $Y_{1:T},\theta_{1:T}$.  Data is truncated at $\pm30$ in each plot (shown in green \& blue, respectively).    The middle graph shows the path of the signal and filter $\mu_{1:T},\theta_{1:T}$.  All of these figures are computed using $T=1{,}000$ and with 3 different values of $\ell \in \{1,7,15\}$. When $\ell=1$, then the $\theta_t = (Y_{t}+Y_{t-1})/2$ and so is entirely non-robust to single data outliers. The performance, consequently, is terrible.  Robustness can only happen when $\ell>1$. The top right of Figure \ref{fig:medianOld} shows $\{1/(T-1)\} \sum_{t=2}^T \{|Y_t-\theta_{t-1}| - |Y_t|\}$ plotted against the hyperparameter $\ell\in \{4,5,...,50\}$ (the results for tiny $\ell$ are terrible and are not plotted). The red line is the oracle version, $\{1/(T-1)\} \sum_{t=2}^T \{|Y_t-\mu_{t}| - |Y_t|\}$, which replaces the feasible predictor $\theta_{t-1}$ by the true signal $\mu_t$. The middle right of Figure \ref{fig:medianOld} shows $\{1/(T-2)\} \sum_{t=3}^T \{|\mu_t-\theta_{t}|\}$ plotted against the hyperparameter $\ell$.  An $\ell \in \{8,...,20\}$ yields reasonable results.

\subsubsection{Stationarity case}\label{sect:uniStat}

Here we use Ex.\ref{ex:running}(b) and set $\alpha=0.6$ in our filter, again with $T=1{,}000$. As the process is stationary the ECDF of $Y_{1:T}$ approximates $F_{Y_1}$.  Hence $\theta_t$ is a weighted median with weights:
$$
w_{t,j} = \frac{1-\alpha}{T} + \alpha \frac{1(0 \le j \le \ell)}{\min(t,1+\ell)},\quad j\in \{t-T,...,t-1\}.
$$
This looks like a smoother, but is a 2-step filter.  The results are in Figure \ref{fig:medianOldStat}.  

The stationary results are very different.  When $\ell=1$ the filter is somewhat robust to noise as the weighted median will involve $Y_{t-1:t}$ with individual weights $0.6/2 + 0.4/T$ while all other data points in $Y_{1:T}$ have weights $0.4/T$. To get an outlier to be the median of this group of data we need multiple outliers close to one another, with one of them occurring at time $t$ or time $t-1$.  Of course, the results for $\ell$ being 7 and 15 are more robust.

The middle pictures suggest having $\ell=7$ seems to yield a filter which best tracks the signal out of these different choices.  The pictures on the right of Figure \ref{fig:medianOldStat} address this issue more systematically, plotting the average loss and the average absolute estimation error against $\ell$, here for $\alpha \in \{0.4,0.6,0.8,1\}$.  The results favor the smaller values of $\alpha$ in both cases.  The results are less sensitive to $\ell\in \{4,...,50\}$ when $\alpha$ is smaller.

\section{Estimator of the filter via simulation}\label{sect:simFilter}    

We use a simulation approach to estimate $\theta_t$ up to user controlled error in $O(1)$ flops, for each $t$.  The entire sweep of the filter can be calculated to arbitrary accuracy in $O(T)$ flops (or even $O(1)$ time on $T$ parallel processors).  The same idea applies to prediction \& smoothing.  

\subsection{Estimator}
     
Recall the estimand $\theta_t$ minimizes weighted lagged losses.   

\begin{definition} Resample $Y_1^*,...,Y_B^*$, where each $Y_b^* \sim F_{t}$. Define the estimator of $\theta_t$ as 
$$
\theta^*_t = \underset{\theta \in \Theta}{\arg }\min \ \widetilde{L}^*(\theta),\quad \text{where} \quad \widetilde{L}^*(\theta) = \frac{1}{B} \sum_{b=1}^B L(\theta,Y_b^*).
$$ 
\end{definition}

Both $\widetilde{L}^*$ and $\theta^*_t$ depend on $B$, yet we suppress this for notational convenience. The $\widetilde{L}^*$ is a simulation based unbiased estimator of $\int_{\mathcal{Y}} L(\theta,y) \mathrm{d}F_{t}(y)$. In the case of simulation samples from $F_{t}$ (that are i.i.d.\ across $b$ conditional on $Y_{1:t}$; Remark \ref{rem:strat}(b)), if $\theta_t$ uniquely minimizes $\int_{\mathcal{Y}} L(\theta,y) \mathrm{d}F_{t}(y)$, which strict convexity delivers for $\alpha \in [0,1)$, then as $B \rightarrow \infty$ the $\theta^*_t -\theta_t \rightarrow 0$ almost surely using standard i.i.d. theory for minimizing average losses based on data from $F_{t}$ \citep[e.g.][]{vanderVaart(98),ChafaiConcordet(07)}. If instead the minimizers form an interval, as can happen for the weighted median when $\alpha=1$, then $\theta^*_t$ converges almost surely to that interval. In practice values of $B$ might be 1,000, or 10,000.  Finally, for a strictly stationary time series the setup of Remark \ref{rem:Fprop}(e) could be the basis of the study of uniform convergence as $B \rightarrow \infty$ for an infinite $T$, but that setup is beyond the scope of this paper.   

\begin{example}
If $L$ is the check function then  $\theta^*_t$ is the empirical quantile of $Y_{1:B}^*$ \citep[e.g.][]{Koenker(05)}.  
If $L$ is $\rho$-type, then $\theta^*_t$ is the M-estimate \citep{Huber(64)}.
\end{example}

\subsection{Sampling from $F_t$} 

Algorithm \ref{alg:F all} is a generic sampler of $Y_1^*,...,Y_B^*$ from $F_t$, regarding $Y_{1:t}$ as fixed, recalling  $F_{t} = (1-\alpha) F_{Y_1} + \alpha \hat{F}_{t},$ and $\hat{F}_{t}(y) = \sum_{j=0}^{t-1} w_{t,j} 1(Y_{t-j} \le y).
$ The proof is by the method of composition.
\begin{algorithm}[Generate $B$ samples from $F_{t}$]\label{alg:F all}  Draw $B_1 \sim {\tt Binomial}(B,\alpha)$ and $B_2=B-B_1$. Draw $Y_{B_1+1}^*,...,Y_{B}^*$ from $F_{Y_1}$.  Then draw $(Y_1^*=Y_{t-j_1}),...,(Y_{B_1}^*=Y_{t-j_{B_1}})$, a weighted resample of $Y_{1:t}$ of size $B_1$ from $\hat{F}_t$, where the $b$-th lag length $j_b$ has 
$$
P(j_b=k)=w_{t,k},\quad b=1,...,B_1,\quad k\in \{0,1,...,t-1\}.
$$
\end{algorithm}

\begin{remark}\label{rem:strat}
    
(a) So long as $j_b$ can be sampled in $O(1)$, that is not depending upon $t$, then Algorithm \ref{alg:F all} costs $O(B)$ and so its cost is not impacted by $t$.  

(b) To force the draws from Algorithm \ref{alg:F all} to be i.i.d. the $Y_1^*,...,Y_B^*$ can be shuffled.  Instead, Algorithm \ref{alg:F all} can be altered, replacing the single binomial draw by a series of i.i.d. Bernoulli draws, each deciding if we draw from $F_{Y_1}$ or $\hat{F}_t$.  

(c) Some researchers will model or know $F_{Y_1}$. The $F_{Y_1}$ will often be estimated by the ECDF for $Y_{1:T}$ --- which can be sampled with replacement from $Y_{1:T}$, satisfying Algorithm \ref{alg:F all}.   

(d) With access to weighted M-estimator software, it may be more efficient to stratify, replacing many tiny weights with a few samples instead. For simplicity take $\alpha=1$, then define   
$$
\grave{L}^*(\theta) = \sum_{j=0}^h w_{t,j} L(\theta,Y_{t-j}) 
+ \left(1-\sum_{j=0}^h w_{t,j} \right)\frac{1}{B} \sum_{b=1}^B L(\theta,Y_b^*),\quad \theta \in \Theta,
$$
where $(Y_1^*=Y_{t-j_1}),...,(Y_{B}^*=Y_{t-j_{B}}),$ where $P(j_b=k\,|\,j_b>h)\propto w_{t,k}$ for $k=h+1,...,t-1$. This structure is used in our applied work, each time taking $h=W_t^{-1}(0.999)$ \& $B=100$.

(e) Sampling from $\hat{F}_t$ is identical to sampling the lag length $j_b^*$ with probability $P(j_b^*=k) = w_{t,k}$ and returning $Y_b^* = Y_{t-j_b^*}$. It was this version which was presented in the introduction.  
\end{remark}

\subsection{Properties of $\theta_t^*$}

The $\theta_t^*$ can be thought of as minimizing empirical loss from a sample $Y_1^*,...,Y_B^*$ whose marginal distribution is from $F_t$.  The unusual aspect of the setup is that there are jumps in $F_t$, corresponding to the distinct elements in $Y_{1:t}$, so long as $\alpha>0$.  

Focus on the special cases: the mean \& median. The mean has a differentiable loss function \& so is smooth in the data \& the weights. The median has a continuous loss function, but with points that are not differentiable.  Hence the median is not smooth in the data or weights.    

\subsubsection{Mean case}

Theorem \ref{thm:mean} characterizes the difference between $\theta^*_{t}$ and $\theta_{t}$ in the mean case with $\alpha=1$.

\begin{theorem}\label{thm:mean} In the quadratic loss case with $\alpha=1$, then 
$
\theta_{t} = \sum_{j=0}^{t-1} w_{t,j} Y_{t-j}.
$
The $\theta^*_t$ is conditionally unbiased for $\theta_t$, given $Y_{1:t}$.  Further, as $B \rightarrow \infty$, then 
\begin{align*}
\sqrt{B}(\theta^*_t-\theta_{t})|Y_{1:t} & \xrightarrow{D} N(0,\sigma_{t}^2),\quad
\text{where}\quad \sigma_{t}^2 = \sum_{j=0}^{t-1} w_{t,j} (Y_{t-j} - \theta_{t})^2.
\end{align*}
\end{theorem}

\begin{proof} The 
\begin{align*}
\theta^*_t &= \frac{1}{B} \sum_{b=1}^B Y_b^*  = \sum_{j=0}^{t-1} w^*_{t,j} Y_{t-j},\quad w^*_{t,j}=  \frac{1}{B} \sum_{b=1}^B 1(j_b = j),
\end{align*}
which has randomized weights.  The counts $(Bw^*_{t,0},...,Bw^*_{t,t-1})|Y_{1:t}$ are multinomial with $B$ trials \& probabilities $(w_{t,0},...,w_{t,t-1})$.  Hence the weights are conditionally unbiased \&, by the multivariate central limit theorem for multinomial proportions, jointly in $j$ as $B\rightarrow \infty$,
$$
\sqrt{B}(w^*_{t,j} - w_{t,j}) | Y_{1:t}  \xrightarrow{D} N(0, w_{t,j}(1- w_{t,j})),\quad B\,{\mathrm{Cov}}(w^*_{t,j},w^*_{t,i}|Y_{1:t}) = -w_{t,j}w_{t,i},\quad i\ne j.
$$
Further, $\theta^*_t$ is conditionally unbiased for $\theta_{t}$ with 
\begin{align*}
\sqrt{B}(\theta^*_t-\theta_{t}) & = \sum_{j=0}^{t-1} \sqrt{B}(w^*_{t,j} - w_{t,j}) Y_{t-j} \xrightarrow{D} N(0,\sigma_{t}^2),
\end{align*}
where, conditioning throughout on $Y_{1:t}$, 
\begin{align*}
\sigma_{t}^2 &= \sum_{j=0}^{t-1} w_{t,j}Y_{t-j}^2 - \sum_{j=0}^{t-1} \sum_{i=0}^{t-1} w_{t,j} w_{t,i} Y_{t-j} Y_{t-i} 
= \sum_{j=0}^{t-1} w_{t,j} (Y_{t-j} - \theta_{t})^2.
\end{align*} \end{proof}

Under the stratification of Remark \ref{rem:strat}(d) only the tail is sampled, so the same limit holds with $\sigma_t^2$ replaced by $\{1-W_t(h)\}\sum_{j>h} w_{t,j}(Y_{t-j}-m_{t,h})^2$, $m_{t,h}=\sum_{j>h} w_{t,j}Y_{t-j}/\{1-W_t(h)\}$.

\subsubsection{Median case}\label{sect:medianCase}

Focus on the $\alpha=1$ with odd $B$ and $Y_{1:t}$ being distinct.  Take a median to be the smallest value at which the cumulative weight reaches 0.5.  Then $\theta^*_{t}\in Y_{1:t}$ is the median of $Y_{1:B}^*$ \& $\theta_{t}\in Y_{1:t}$ is the weighted median of $Y_{1:t}$. 
Theorem \ref{thm:dist of theta} gives the CDF of $\theta_{t}^{\ast }|Y_{1:t}$.  It implies: 
\[
P(\theta_{t}^{\ast }=\theta_{t}|Y_{1:t})=F_{\mathtt{Bin%
}(B,\widehat{F}_{t}(\theta_{t}-))}\left( \lceil B/2\rceil-1\right) -F_{%
\mathtt{Bin}(B,\widehat{F}_{t}(\theta_{t}))}\left( \lceil B/2\rceil-1\right),
\]%
where {\tt Bin} denotes a binomial.  
Now 
$
\widehat{F}_{t}(\theta_{t})-\widehat{F}_{t}(\theta%
_{t}-)=w_{t,\widetilde{k}},$
where
$\widetilde{k}$ is such that 
$Y_{t-\widetilde{k}}=\theta_{t}.$ (With ties the jump is the summed weight of the tied values.)
This could be large if $\widetilde{k}$ is small. 

\begin{figure}[htbp]
    \centering
\includegraphics[width=0.32\linewidth]{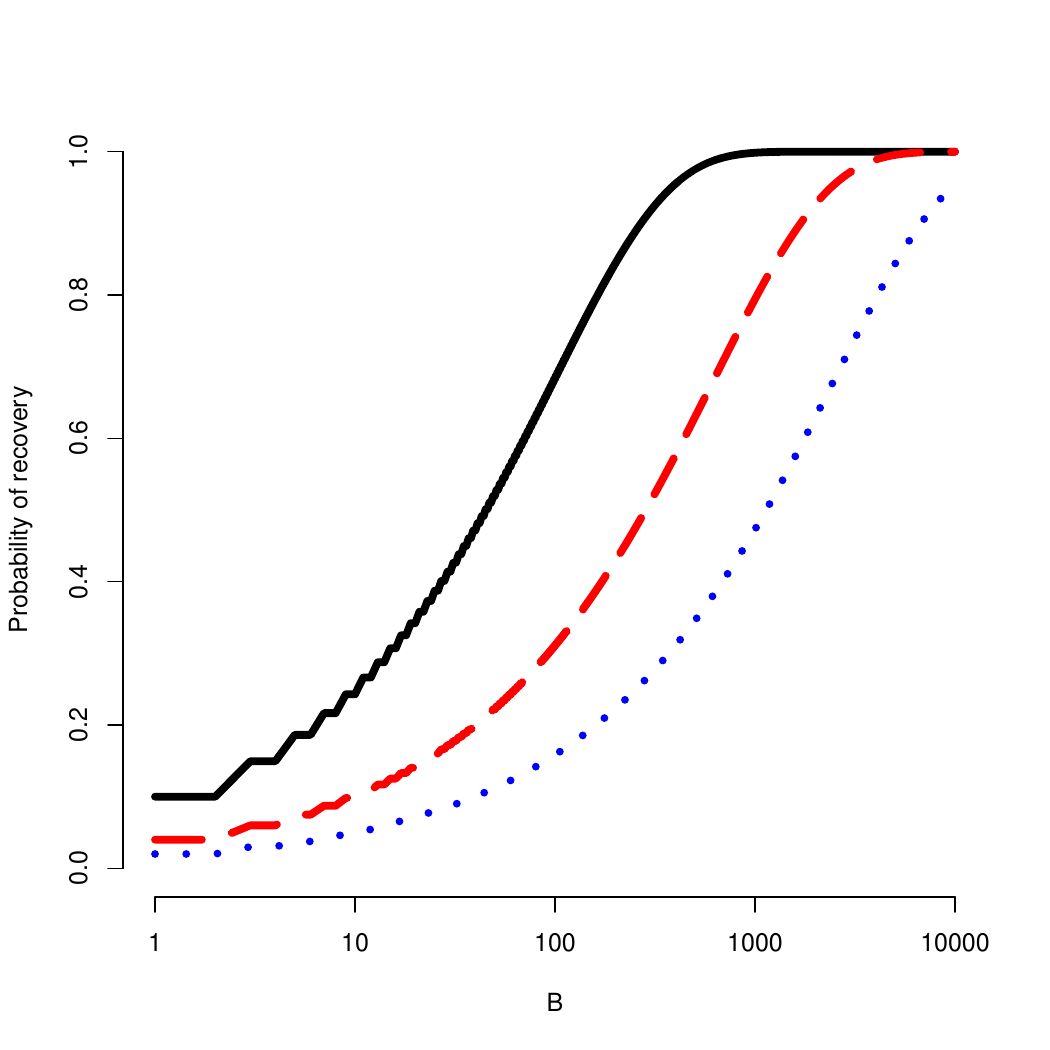} 
    \caption{$P(\theta_{t}^{\ast }=\theta_{t}|Y_{1:t})$, as a function of $B$, on the log scale.  Black full line:  $\epsilon =0.05$, taking $\widehat{F}_{t}(\theta_{t}-)=0.5-\epsilon$ \& $\widehat{F}_{t}(\theta_{t})=0.5 + \epsilon$.  Red dashed line: $\epsilon=0.02$.  Blue dotted line: $\epsilon= 0.01$.  Roughly $2\epsilon$ corresponds to the size of the weight at $\theta_{t}$.} 
    \label{fig:recover}
\end{figure}

Figure \ref{fig:recover} shows a stylized version of this taking $\widehat{F}_{t}(\theta_{t}-)=0.5 -\epsilon$ and $\widehat{F}_{t}(\theta_{t})=0.5+\epsilon$, where $\epsilon \in \{0.05,0.02,0.01\}$.  The resulting $P(\theta_{t}^{\ast }=\theta_{t}|Y_{1:t})$ is plotted against $B$.  This shows the probability drives upwards towards 1 quickly as $B$ increases, even when $\epsilon$ is quite small. 
\begin{theorem}
\label{thm:dist of theta}For the median filter with $\alpha=1$, writing $\lceil x \rceil$ as the ceiling of $x$, then  
\[
P(\theta_{t}^{\ast }\leq c|Y_{1:t})=1-F_{\mathtt{Bin}(B,\widehat{F}_{t}(c))}\left( \lceil B/2\rceil-1\right) ,\quad \widehat{F}_{t}(c)=\sum_{j=0}^{t-1}w_{t,j}1(Y_{t-j}\leq c),\quad c\in \mathbb{R}.
\]
\end{theorem}

\begin{proof}
Under this median convention $\theta_{t}^{\ast },\theta_{t}\in Y_{1:t}$ are unique. Write  
\[
B_{t}^{\ast }(c)=B\sum_{j=0}^{t-1}w^{\ast}_{t,j}1(Y_{t-j}\leq c),\quad
B_{t}^{\ast }(c)|Y_{1:t}\sim \mathtt{Bin}(B,\widehat{F}_{t}(c)).
\]%
We will work with $c$'s non-decreasing subgradient of $\frac{1}{2}\sum_{j=0}^{t-1} w^{\ast}_{t,j}|Y_{t-j}-c|$: 
\begin{eqnarray*}
\widehat{s}_{t}(c) &=&\sum_{j=0}^{t-1}w^{\ast}_{t,j}g(c,Y_{t-j})=\frac{1}{B}B_{t}^{\ast }(c)-\frac{1}{2},\quad
g(c,y)=1(y\leq c)-\frac{1}{2}\in \left\{ 1/2,-1/2\right\}.  
\end{eqnarray*}%
Then $\theta_{t}^{\ast }>c$ is the same event as $\widehat{s}_{t}(c)<0$. Thus 
\begin{eqnarray*}
P(\theta_{t}^{\ast } &\leq &c|Y_{1:t})=1-P(\theta_{t}^{\ast }>c|Y_{1:t})=1-P(\widehat{s}_{t}(c)<0|Y_{1:t})=1-P\left(
B_{t}^{\ast }(c)<\frac{B}{2}|Y_{1:t}\right)  \\
&=&1-P\left( B_{t}^{\ast }(c)\leq \lceil B/2\rceil-1|Y_{1:t}\right) .
\end{eqnarray*}%
Rewriting, this is the main stated result in the Theorem. \end{proof}

\begin{example} Use the nonstationary case of Ex.\ref{ex:running}(a) \& set $\alpha=1$.  Focus on the uniform weight case, with $\ell$ lags, \& on the rolling median.  In this case there is no need to use simulation.  We do it anyway: it teaches some important practical lessons.  The $Y_{1:B}^*$ samples with replacement from $Y_{t-\ell:t}$, after which we compute $\theta^*_t$ based on $Y_{1:B}^*$. The $Y_{1:B}^*$ have at most $\ell+1$ points of support.  Whether $B$ is even or odd is not so important now as $B$ will be very large.  What matters is $\ell$. If $\theta_t^*$ is taken as the median convention of Section \ref{sect:medianCase}, then as $B\rightarrow \infty$, 
$
P(\theta_t^* = \underline{\theta}_t|Y_{1:t}) \rightarrow 1/2$ and 
$P(\theta_t^* = \overline{\theta}_t|Y_{1:t}) \rightarrow 1/2,
$
if $\ell$ is odd, while when $\ell$ is even then $P(\theta_t^* = \theta_t|Y_{1:t})\rightarrow 1.$ The odd case is not ideal.         

Instead, inspired by the Hazen interpolative version of a median, first compute 
$
w^*_{t,j} = \frac{1}{B}\sum_{b=1}^B 1(Y_b^* = Y_{t-j}),$ where $j=0,...,\ell$ 
and then compute the weighted median of $(Y_t,w^*_{t,0}),...,$ $(Y_{t-\ell},w^*_{t,\ell}),$ using interpolation by the method detailed in Remark \ref{remark:median}.        
\end{example}

\subsection{Fast sampling of the lag length}\label{sect:lag}

We need to sample the lag $j$ with probability 
$
P(j=k) = w_{t,k}$ for $k=0,1,2,...,t-1,$
in a cost which does not depend upon $t$. Notice this task is free of data and loss function.        

Sampling from discrete random variables is classical in Monte Carlo, e.g. \cite{Ripley(87)} and Ch. 3 of \cite{Devroye(86)}.  In the uniform, exponential \& hyperbolic cases the quantile function of the weights is available and so by inverse transform sampling $j \sim W_t^{-1}(U)$, where $U\sim U(0,1)$.  The superposition case can be carried out using the composition method.

\subsubsection{Example: exponential weights and the median filter}\label{ex:expweights}

Let  
$
L(\theta,y) = |y - \theta|-|y|.
$
Then $\theta^*_t$ is the median of $Y_{1:B}^*$. 
We use the nonstationary case Ex.\ref{ex:running}(a), with $\alpha=1$, and exponential weights. 

\begin{figure}[htbp]
    \centering
    \begin{tabular}{@{}c@{\hspace{-0.25cm}}c@{\hspace{-0.25cm}}c@{\hspace{-0.25cm}}c@{}}
        & \textbf{$Y_t,\theta_t^*$} & \textbf{$\mu_t,\theta_t^*$} & \textbf{Average loss} \\
        \raisebox{2.0cm}[0pt][0pt]{{\small $\lambda=0.9$\hspace{0.3cm}}} &
        \includegraphics[width=0.25\linewidth]{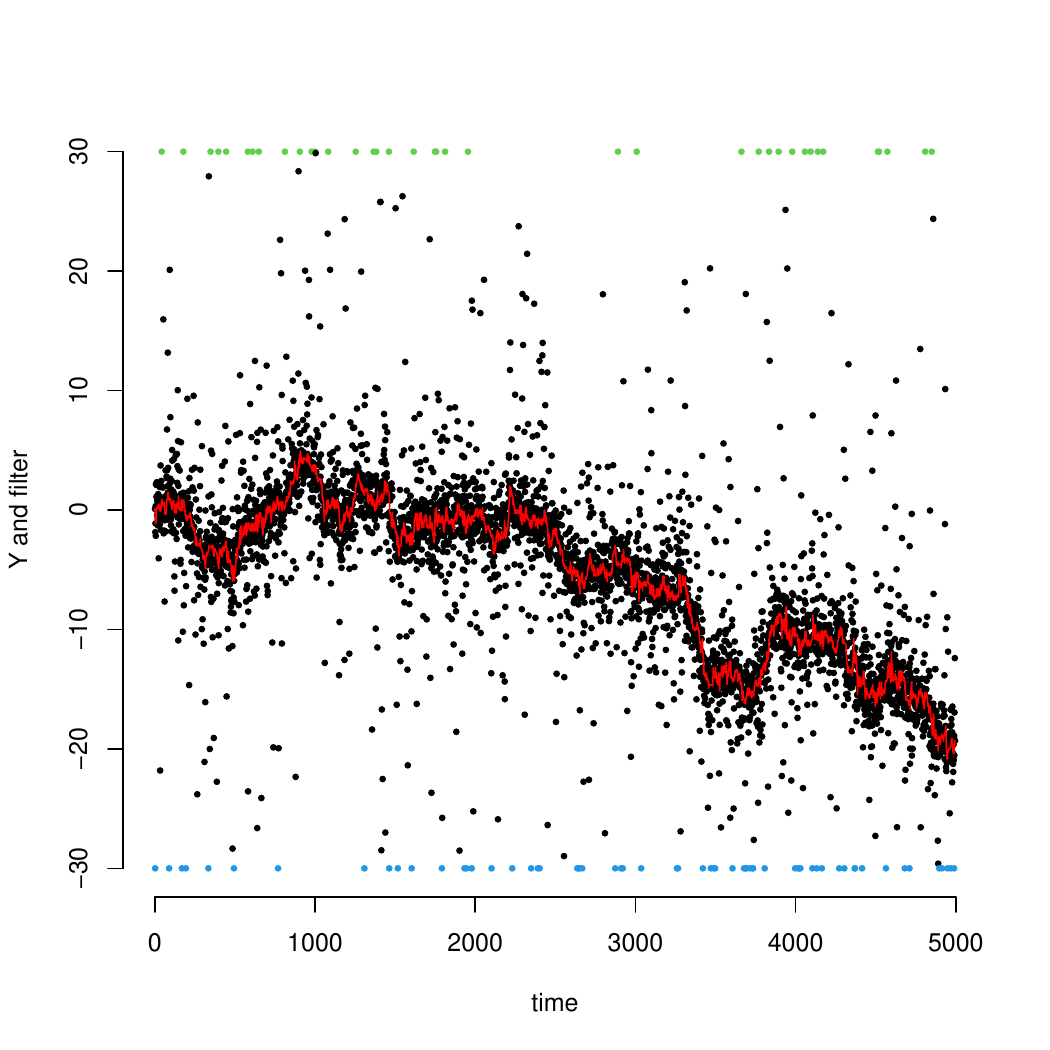} &
        \includegraphics[width=0.25\linewidth]{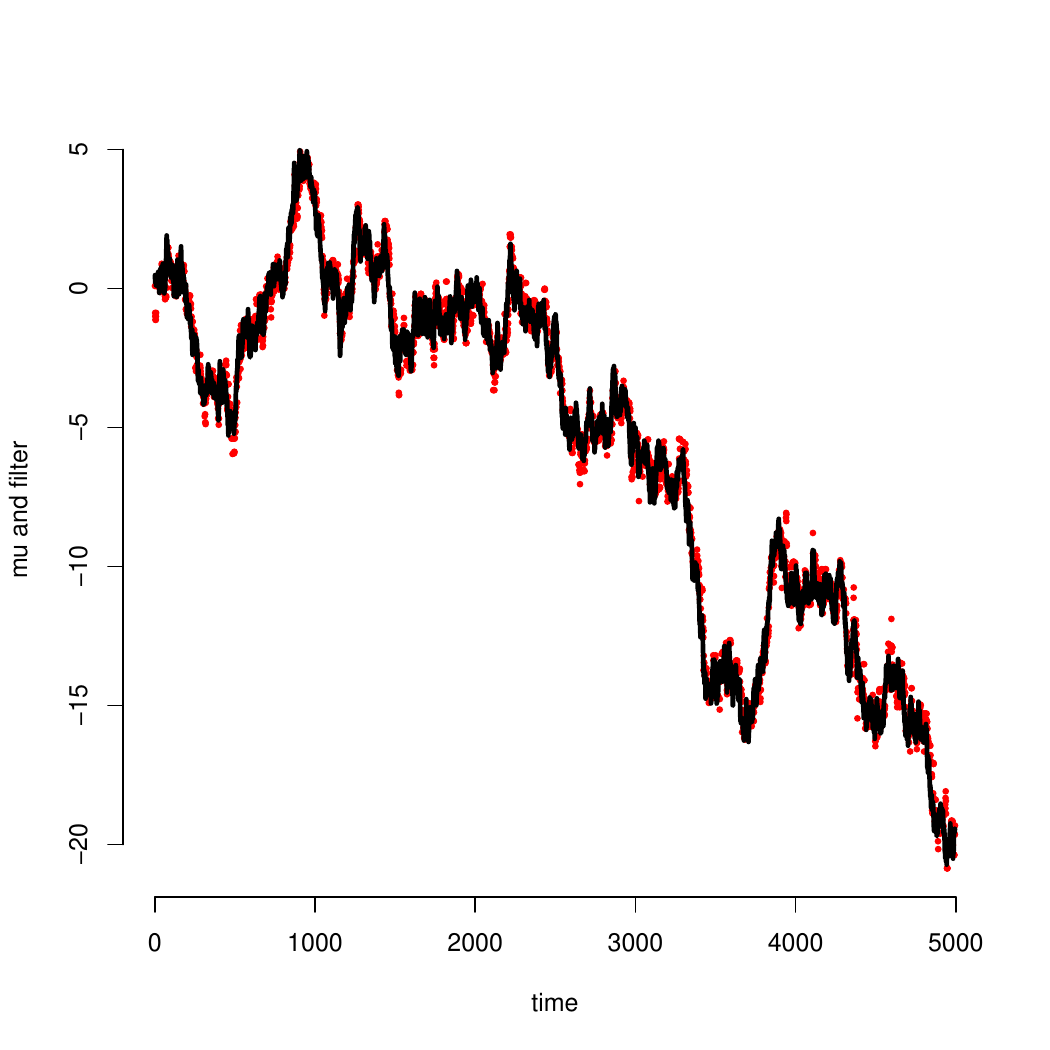} &
        \includegraphics[width=0.25\linewidth]{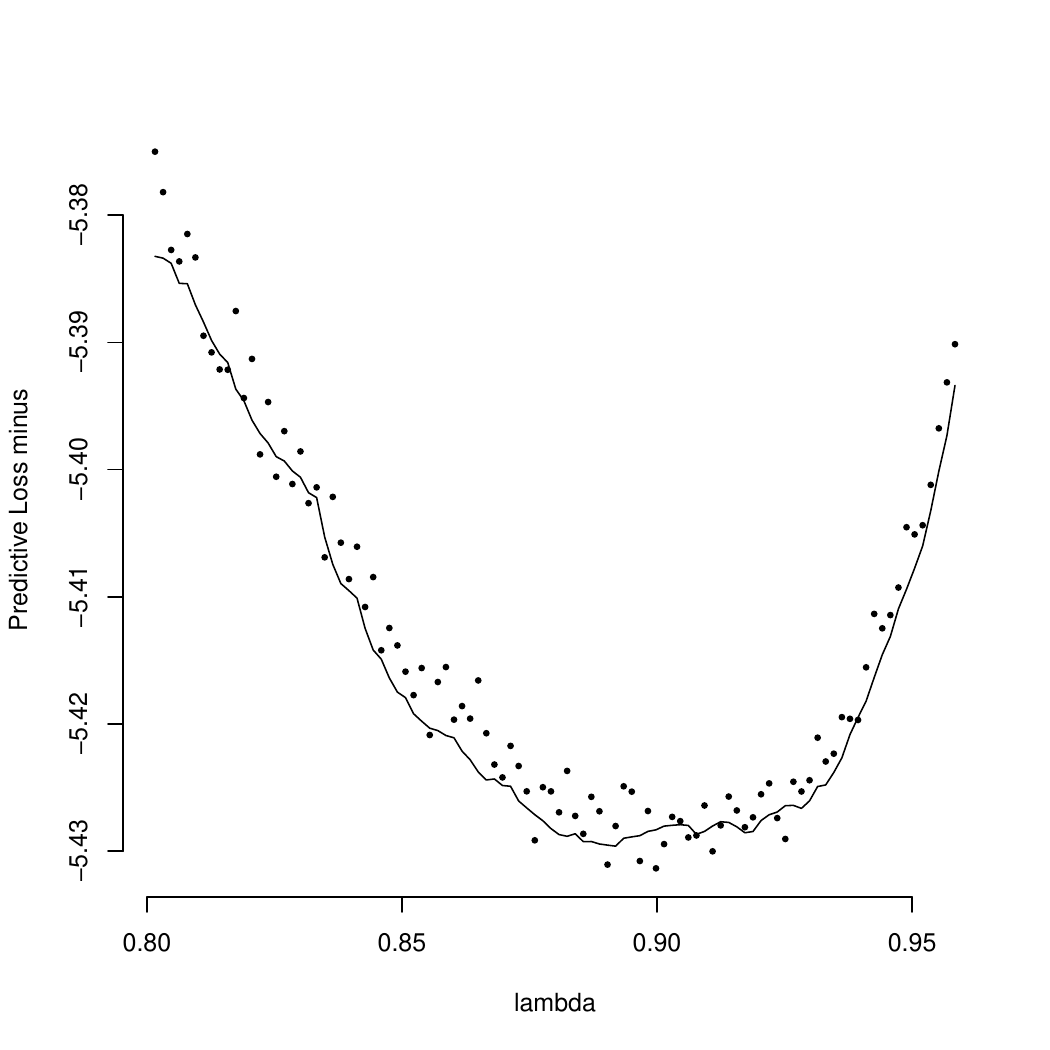} \\
    \end{tabular}
    \vspace{-5mm}
    \caption{Median filter using $B=500$.  Left hand side: the data $Y_t$ \& the median filter $\theta_t^*$ plotted against time, run with $\lambda=0.9$.  In the plot the data is truncated at $\pm30$ (shown in green and blue, respectively).    Middle: the signal $\mu_t$ and the median filter $\theta_t^*$ run with $\lambda=0.9$.  RHS: average loss plotted against $\lambda$ using $\theta_t^*$ (dotted line) \& $\theta_t$ (solid line). } 
    \label{fig:median}
\end{figure}

The left of Figure \ref{fig:median} shows the data and $\theta_t^*$ based on $B=500$ and $\lambda=0.9$.  The heavy tailed data is truncated in the picture (only) at $\pm30$ and represented by green and blue dots. The result is as expected, the filter is entirely robust to the heavily tailed noise.  The middle graph shows the signal $\mu_t$ and $\theta_t^*$, indicating close alignment.

The choice of $\lambda=0.9$ is ad hoc.  A reasonable way to select a value of $\lambda$ from data is to minimize an average predictive loss, compared to the loss with a zero predictor.    The right of Figure \ref{fig:median} shows $\lambda=0.9$ is not poor by plotting$
\frac{1}{T-1} \sum_{t=2}^T \{|Y_t-\theta_{t-1}^*| - |Y_t|\}$ as well as $
\frac{1}{T-1} \sum_{t=2}^T \{|Y_t-\theta_{t-1}| - |Y_t|\}.$
Neither is continuous in $\lambda$ as $\theta_{t-1}$ is not continuous in $\lambda$.  As $\lambda$ moves, $\theta_{t-1}$ jumps through $Y_{1:t}$. 

\begin{figure}[htbp]
    \centering
    \begin{tabular}{@{}c@{\hspace{-0.25cm}}c@{\hspace{-0.25cm}}c@{\hspace{-0.25cm}}c@{}}
        & \textbf{$B=500$} & \textbf{$B=5{,}000$} & \textbf{$B=50{,}000$} \\
        \raisebox{2.0cm}[0pt][0pt]{{\small $\lambda=0.9$\hspace{0.3cm}}} &
        \includegraphics[width=0.25\linewidth]{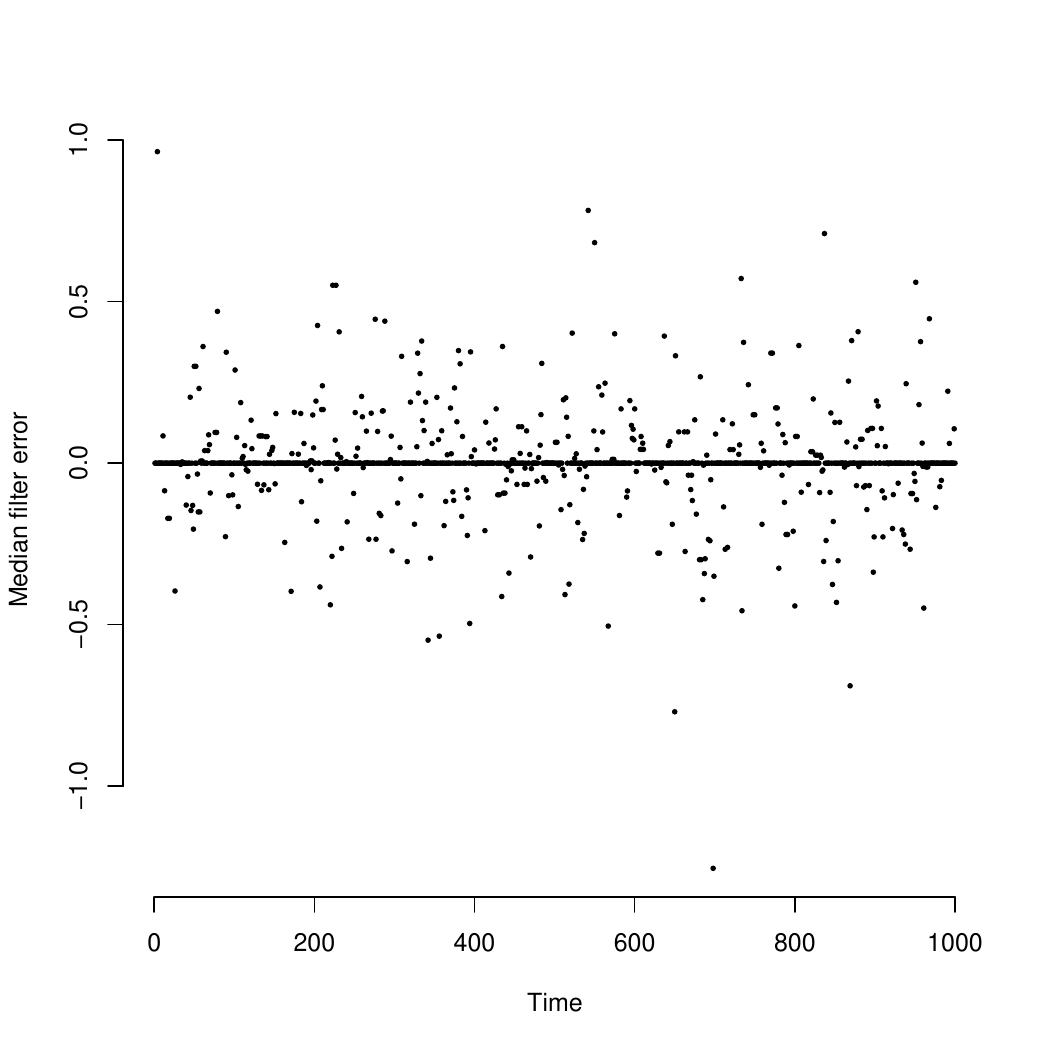} &
        \includegraphics[width=0.25\linewidth]{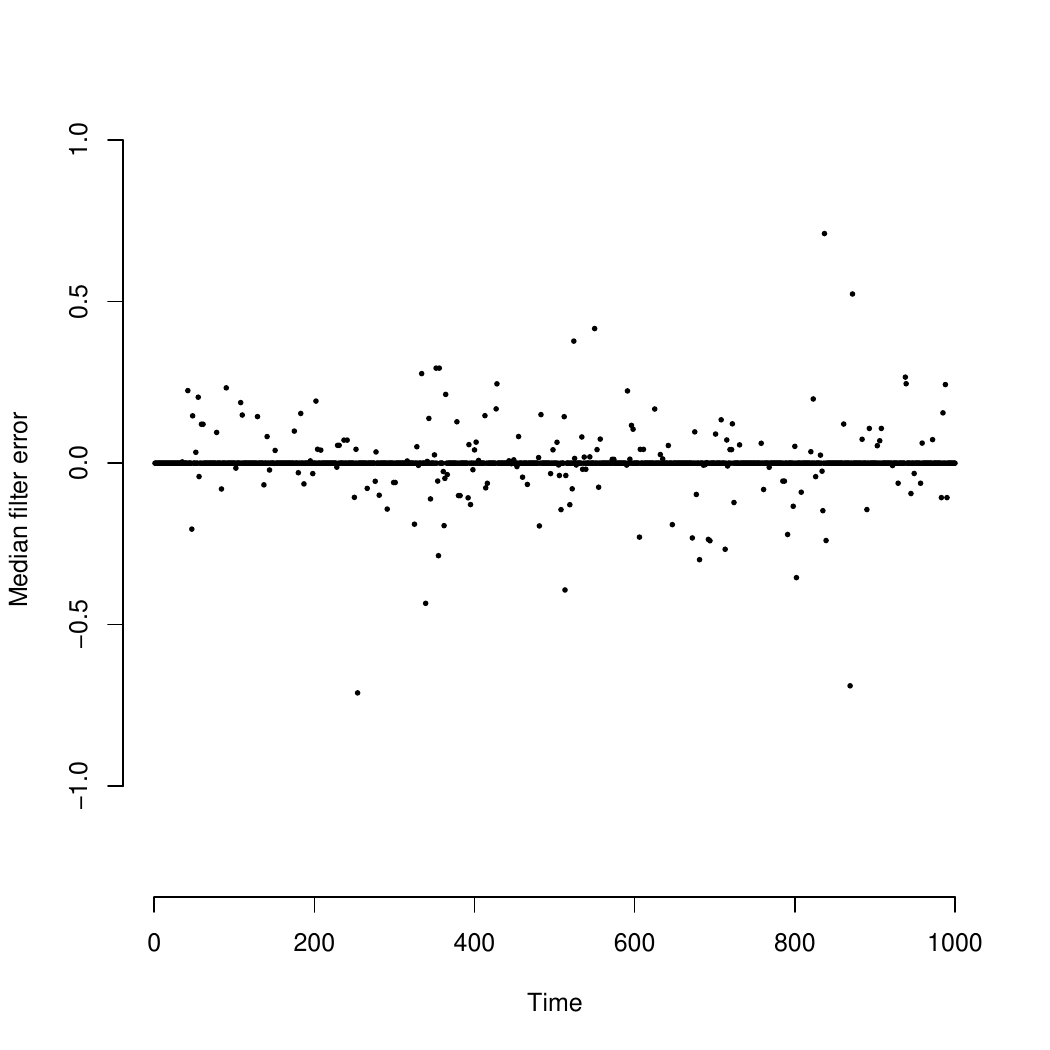} &
        \includegraphics[width=0.25\linewidth]{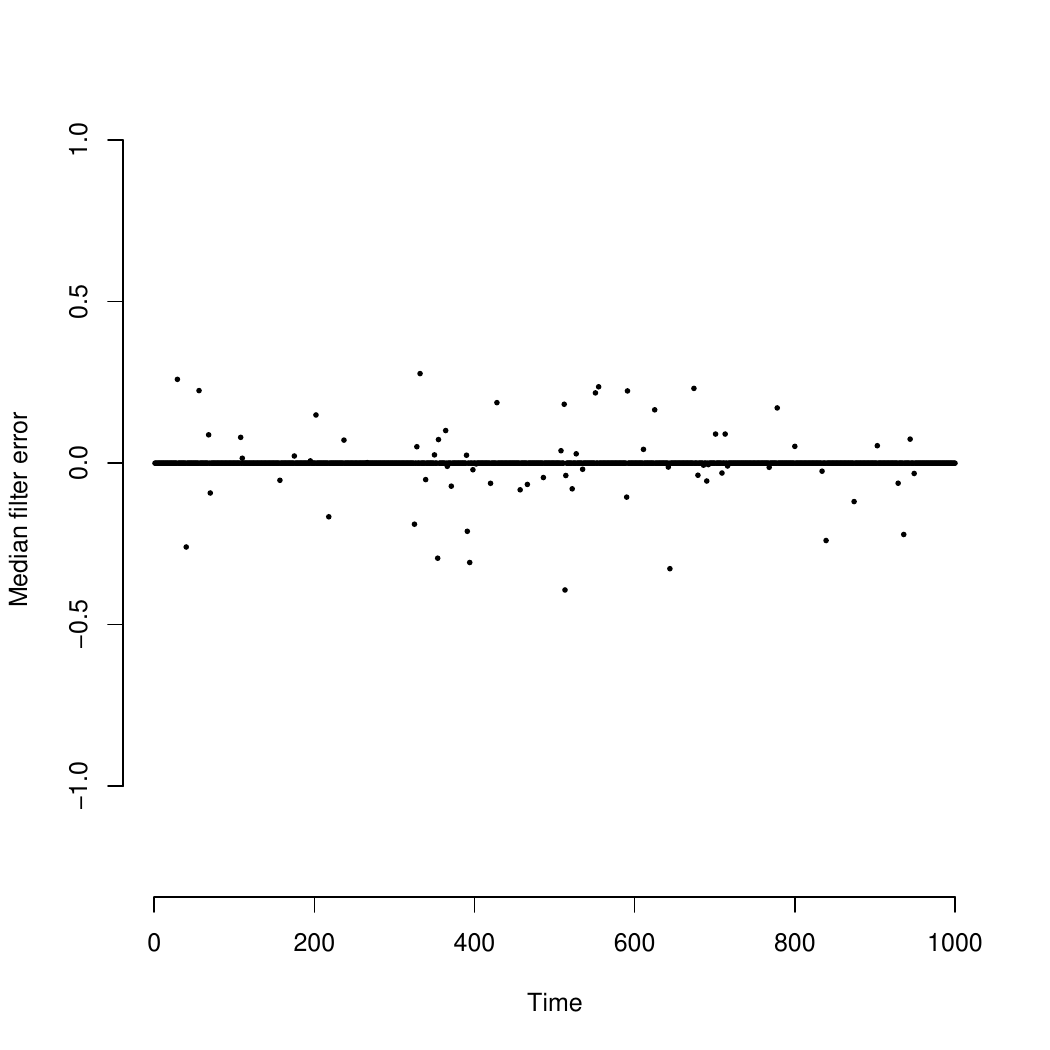} \\
    \end{tabular}
    \vspace{-5mm}
    \caption{Median filter error  $\theta^*_t - \theta_t$ implemented using the lag sampler with $B=500$, $B=5{,}000$ and $B=50{,}000$, compared to the slow analytic median filter $\theta_t$. All with $T=1{,}000$. } 
    \label{fig:medianErr}
\end{figure}

Figure \ref{fig:medianErr} plots $\theta^*_t - \theta_t$ for different values of $B$, using data up to time $1{,}000$ (the scatter is harder to see when $T=5{,}000$).  As $B\rightarrow \infty$ the errors get smaller in magnitude (as expected), with the probability of the error being 0 going to 1, as was formalized in Theorem \ref{thm:dist of theta}.     

\begin{remark} 

(a) The important innovation is simulating the lag $j$, with a truncated geometric distribution.  This algorithm resamples past data, with the weights determined by $\lambda$.  

(b) If $\lambda \in [0,1)$, then as $t\rightarrow \infty$ the $j \xrightarrow{D} {\tt Geometric}(\lambda),$ a geometric random variable.

(c) The $\theta^*_{1:T}$ can be implemented in parallel.  It  has no recursive feature.  This contrasts with Kalman filtering \citep[e.g.][]{DurbinKoopman(12)}, sequential Monte Carlo \citep[e.g.][]{ChopinPapasphiliopoulos(20)} as well as the EWMAs \citep{Brown(56)}.

(d) The cost of an entire sweep of the filter for $t=1,...,T$ ignores the cost of going from knowing $j$ to retrieving $Y_{t-j}$ out of computer memory, taking it from the stored $Y_{1:T}$. This address lookup cost is $O(1)$ with respect to $T$.  
\end{remark}

\section{Some uses of robust filtering in financial economics}\label{sect:preav}

\subsection{Scientific background}

We illustrate the filtering methods in the context of nonparametrically estimating the volatility of a heavily traded financial asset using the sequence of transactions.  The background for this type of problem starts with the foundational work on realized volatility \citep{AndersenBollerslevDieboldLabys(01),BarndorffNielsenShephard(02realised)} and connects to the stochastic analysis of Brownian motion, quadratic variation, stochastic differential equations and discretization \citep[e.g.][]{Levy(40),JacodProtter(12)}.  Here we are interested in the branch of the literature which nonparametrically deals with so-called market microstructure noise.  The main 3 approaches in this literature are the two-scale estimator \citep{ZhangMyklandAitSahalia(05)}, the realized kernel \citep{BarndorffNielsenHansenLundeShephard(08realised)} \& the preaverage estimator \citep{JacodLiMyklandPodolskijVetter(07)}.  A recent paper is \cite{LiLinton(26)}.  Alternative rich data sources for volatility estimation include  candlesticks \citep{BollerslevLiLiLi(26)} and implied volatilities \citep{AndersenArchakov(21)}.

Our starting point is preaveraging.  It has 3 stages.  (1), it sample averages recent trades to produce a continuous time estimate of the most recent underlying price --- we think of this as a linear filtering (the averaging could be over a moving window or use a decay function such as exponential).  (2), it computes the squared first difference of the filtered price over a fixed time interval, e.g. 1 minute.  (3), it square roots weighted sums of these squared differences to form a spot volatility filter. If the sum is over longer periods, such as a day, and the weights are uniform, it is an ``integrated volatility estimator.''   

Inspired by the theory of \cite{MyklandZhang(16)}, our approach replaces the first stage linear filter of recent trades with a robust filter.  The second and third stages are largely unchanged, although we employ hyperbolic weights in our volatility filter.  The robust preaveraging approach is compelling. We will show the trade-by-trade data is full of extremely heavy tailed outliers --- in our empirical work the variance of the noise is likely infinite, although this is not necessarily true for other assets.  In theory the existence of this variance is a necessary condition for the validity of two-scale, realized kernels \& preaverage estimators.  We will show that using robust methods reduces the fragility of the resulting volatility sequence.  

Our methods relate to a step of the high frequency data cleaning procedure pioneered by \cite{BrownlessGallo(06)} and \cite{BarndorffNielsenHansenLundeShephard(09)}.  Both use smoothers (not filters) with uniform weights over a moving window to identify ``outliers'' in the trade data which are then removed.  \cite{BrownlessGallo(06)} and \cite{cipollini2026volare} (i.e. the VOLARE archive) employ a trimmed mean smoother.  \cite{BarndorffNielsenHansenLundeShephard(09)} use a median smoother.  The use of smoothers makes their approach hard to use in real time, but they show it is helpful for academic research.  

Carrying out robust filtering on massive trade-by-trade datasets (which can have more than 1 million datapoints in a day) is computationally non-trivial except in the case where a moving block of data is used in the filtering.  But using moving blocks of data is very constraining.  It typically generates inefficient filters and automatically leads to (highly predictable) jumps in the filtered estimate as new datapoints enter and leave the block.  

Relatedly, \cite{hansen2026exact} look at a variety of robust score driven filters for a time series of realized volatilities, while \cite{FanKim(18)} use Huber loss on the square of preaveraged returns. 

\subsection{Framework}

We model a high-frequency stock price as the efficient price observed with microstructure noise,
\begin{equation*}
    Y_t = X_t + \epsilon_t,\quad t \in \{1,...,T\}, 
\end{equation*}
where $Y_t$ is the logarithm of the traded price (we will sometimes casually refer to it as prices), $X_t$ is the latent efficient log-price at trade $t$ and $\epsilon_t$ is the ``noise.'' We will generically write $\hat{X}_t$ as an estimator of $X_t$ \& $\hat{\epsilon}_t = Y_t - \hat{X}_t$.  Here $T$ is the number of trades on a specific day, $T_s$ is the number of trades in the first $s$ seconds after that day's market opens, $\tau_t$ is the seconds after the open of the $t$-th trade and $[X,X](u)$ is the quadratic variation of the $X$ process up to time $u$-seconds after the open. We will analyze each day's data separately, although sometimes we will select parameters indexing weights \& diurnal features using data across days.        

Pre-averaging estimates $X$ by linear filtering: averaging recent observed prices to remove the microstructure noise. The linear filter is sadly not robust to outlying trades.  This prompted \citet{MyklandZhang(16)} to replace it by the block $M$-estimator of their equation~(9). The day is divided into non-overlapping blocks of trades and the efficient price is estimated once per block. Their estimator is the $\alpha=1$ case of the weighted-loss smoother of Definition~\ref{def:ewm}, with uniform weights over the block's trades. For us, at trade $t$ the filter minimizes a weighted loss over the recent trades, indexed by their lag $j$,
\begin{equation*}
    \theta_t = \underset{\theta \in \Theta}{\arg }\min \  \sum_{j=0}^{t-1} w_{t,j}\, L(\theta, Y_{t-j}),\quad w_{t,j} \ge 0, \quad \sum_{j=0}^{t-1} w_{t,j}=1,\quad t=1,...,T,
\end{equation*}
We build on their design in 3 directions, allowing general weights $w_{t,j}$ and losses, computing the filter at \emph{every trade} $t$ from the past trades alone, and letting the robustness threshold move with the noise scale through the day (Section~\ref{sect:preav-level}). The researcher chooses $L$.  

Under squared loss $L(\theta,y)=\tfrac12(y-\theta)^2$, $\theta_t$ is the local average of the prices, driving pre-averaging \citep{JacodLiMyklandPodolskijVetter(07)}. The Huber loss $L(\theta,y)=H_\delta(y-\theta)$, is, for a deviation $u$,
\begin{equation*}
    H_\delta(u) = \begin{cases} \tfrac{1}{2}u^2, & |u|\le \delta,\\[2pt] \delta\,|u| - \tfrac{1}{2}\delta^2, & |u| > \delta, \end{cases}
\end{equation*}
bounds the influence of outlying trades \&, unlike the absolute loss, is differentiable everywhere \& offers a potentially efficient robust estimator. Efficiency is vital here.  Although our datasets are massive, our estimation is local and so the half life of the weights at trade $t$ can be modest.  

\citet{MyklandZhang(16)} use the Huber loss. The asymptotic theory of their block case, its consistency, rate and treatment of jumps, is given in the same paper, building on the linear pre-averaging theory of \citet{JacodLiMyklandPodolskijVetter(07)}.

Our empirical work will be based on trades on Nvidia (ticker \texttt{NVDA}), from 1, 2, 3, 4, 7, 8, 9, 10, 11, 14, 15 \& 16 October 2024, with the 10 trading days before them as a presample for the start-up and diurnal rules of Section~\ref{sect:preav-ends}.  The data is downloaded from the NYSE Trade and Quote (TAQ) database, accessed through Wharton Research Data Services.  A typical day is open from 9.30am to 4.00pm, Eastern Time, \& has around 1.5 million trades over these 23,400 seconds, averaging roughly 70 trades a second. We use trades only. Quotes move on a 1-cent grid but trades need not. About 2/3 of the trades are at whole cents, 1 in 10 at the 1/2 cent, and the rest at other sub-penny prices inside the spread. No trade shares the same nanosecond time stamp.  Appendix~\ref{app:xsec} repeats the calculation on 53 other major U.S. stocks \& 1 ETF for those 12 days; our Web Appendix gives the full results for each.     

For the \texttt{NVDA} trades we run five filters of the type we developed above. Later filters are built from the output of earlier ones. The discount rate $\lambda$ differs across the filters.  Each filter treats the trade times \& the output of the earlier filters as given, as Definition~\ref{def:ewm} requires.
\begin{enumerate}
    \item A median filter, $m_t$, gives a rough estimate of the price level (Section~\ref{sect:preav-level}).  
    \item The noise scale, $b_t$, is a slowly-varying weighted median filter of the nonzero absolute residuals $|Y_t-m_t|$.  It sets the Huber threshold $\delta_t$ (Section~\ref{sect:preav-level}).
    \item A Huber filter, $\theta_t$, gives the estimated efficient price at trade $t$ (Section~\ref{sect:preav-path}). 
    \item A filtered volatility $\hat{\sigma}_t$ calculated as the square root of the sums of weighted squared filtered returns  (Section~\ref{sect:preav-vol}), as well as the corresponding estimator of the integrated volatility.  The results are compared across 4 different types of filtered returns: 
    \begin{enumerate}
    \item the raw returns \citep{AndersenBollerslevDieboldLabys(01),BarndorffNielsenShephard(02realised)}, 
    \item linear filter \citep{JacodLiMyklandPodolskijVetter(07)}, 
    \item median filter,
    \item Huber filter \citep{MyklandZhang(16)}. 
    \end{enumerate}
    \item An outlier direction, $p^{+}_t$, gives the filtered probability that $Y_t-\theta_t > \delta_{t-1}$, that is, an outlier lies above the estimated efficient price (Section~\ref{sect:preav-direction}).  The same filter yields $p^{-}_t$, the probability that $Y_t-\theta_t < -\delta_{t-1}$, that an outlier lies below the estimated efficient price.  $p^{+}_t/(p^{+}_t+p^{-}_t)$, the conditional probability of a positive outlier, is time-varying.   
\end{enumerate}

The dataset refines the rules of \citet{BarndorffNielsenHansenLundeShephard(09)} so that it would be available in real time for market participants.  In particular, we keep trades during regular trading hours (9.30am to 4.00pm, Eastern Time) at strictly positive prices, dropping those the exchange records as ``corrected''. We also include the price from the closing auction and record its time exactly at 4pm in our database.\footnote{Trades whose sale condition marks them as extended-hours reports are dropped. On 16 Oct this keeps $1{,}521{,}749$ of $1{,}616{,}806$ records. $94{,}891$ fall outside regular hours, $152$ carry extended-hours sale conditions, $14$ are recorded as corrected, \& none have non-positive prices. Taking the opening auction as trade $1$ then drops the $411$ trades recorded before it, leaving $1{,}521{,}338$; appending the closing auction gives $T=1{,}521{,}339$.} Beyond this ``cleaning'' we keep every trade, including the outliers, since resisting them in real time is the point. In particular, we keep ``off-exchange trades'', which come from dark pools, negotiated blocks, or through the filling of retail orders by wholesalers. These trades appear in the database through a FINRA Trade Reporting Facility \& represent around 51\% of the trades on 16 October. There is a reporting requirement that such trades are reported within 10 seconds, \& we use the time of the report. Around 91\% of what we classify as outliers to our Huber filter come from these off-exchange trades.  

We highlight the five filters on 16 October 2024, a day with $T=1{,}521{,}339$ trades (the median size of each trade is around 50 shares). The filters are computed on the log-prices $Y_t$, so residuals, scales and thresholds are log increments inside every objective. We report them at the prevailing price level, in dollars or cents, for readability.

Finally, we take 
$
Y_1 = X_1,
$
as the single price from the 9.30am opening auction (viewing the auction as a form of preaveraging).  The trade in the opening auction on NASDAQ is indicated in the TAQ database by the trade condition code ``O''. For our 12 days of data, the result of the auction is recorded between 9.30:00.1am and 9.30:01.4am (so $\tau_1$ is between 0.1 and 1.4 seconds over these days), for 1.1 to 2.7 million shares. On 16 October the auction settled 1,857,756 shares at a price of \$134.01 per share.  We take $Y_{T} = X_{T},$ as the price of the closing auction on NASDAQ, indicated in the TAQ database by the trade condition code ``6''. This auction happens algorithmically exactly at 4pm and the result is broadcast immediately to market participants.  We include this datapoint in the database and label the associated time exactly at 4pm whatever the time print on the TAQ tape so all our data falls in the interval 9.30am to 4.00pm.  On 16 October the closing auction settled 15,146,434 shares at a price of \$135.72 per share.  That the closing auction has much bigger volume than the opening auction is typical across time \& assets \& is important in terms of market impact \citep{GoyalJegadeeshWu(26)}.

\subsection{Initial estimators: median filter and noise scale}\label{sect:preav-level}

The Huber filter needs help: a reasonable choice for the threshold $\delta_t$. To implement this we will first run a median filter on the prices and compute a time-varying noise scale.        

The first filter, $m_t$, is an exponentially weighted median of the recent trades (Ex.\ref{ex:exponential}) having a half-life of 
3 trades (the half-life determines $\lambda$; Section~\ref{sect:preav-path} explains the choice), giving
$$m_t = \underset{\theta \in \mathbb{R}_{\ge 0}}{\arg }\min \  \sum_{j=0}^{t-1} w_{t,j}\, |Y_{t-j}-\theta|,\quad w_{t,j}\propto \lambda^j,\quad t \in \{2,...,T-1\},\quad m_1=Y_1,\quad m_T=Y_T.
$$ Most trade prices lie on the 1-cent quote grid and bounce between the bid and the ask, so the median lands on a grid value, while the efficient price $X_t$ presumably lies between the bid and the ask. The median $m_t$ is a robust local level, not a great estimator of the efficient price, and serves only as input to the noise scale $b_t$ (\cite{BarndorffNielsenHansenLundeShephard(09)} use a uniform weighted median smoother on quote data with a window length of 50 quotes.)  We compute $m_t$ by the stratified estimator of Remark~\ref{rem:strat}(d) (with $h = W_t^{-1}(0.999)$ and $B=100$), at $O(1)$ flops a trade, taking the weighted median as the smallest value reaching half the weight.

The 2nd filter, a noise scale $b_t$, is also a weighted median, but now of the nonzero absolute versions of the median filter residuals 
$
\tilde{\epsilon}_t = Y_t - m_t,$ with $t\in \{1,...,T\},$ with $\lambda$ selected (the details of this will be given in a moment) so the weights have a 4-minute clock-time half-life (roughly $16{,}800$ trades at the average trading rate). Recalling $\tau_t$, the seconds after the open of $t$-th trade, the weights decay in clock time, yielding the definition of the noise scale at trade $t$:  
$$b_t = \underset{\theta \ge 0}{\arg }\min \  \sum_{j=0}^{t-1} w_{t,j}\,\mathbf{1}\{|\tilde{\epsilon}_{t-j}|>0\}\, \bigl||\tilde{\epsilon}_{t-j}|-\theta\bigr|,\quad w_{t,j}\propto \lambda ^{\tau_t-\tau_{t-j}},\quad t \in \{T_{60}+1,...,T\}.$$
More than half of the residuals, $\tilde{\epsilon}_t$, are exactly zero, since prices concentrate on the 1-cent quote grid and most trades fall on the level $m_t$; these zeros are uninformative about the scale, so we take $b_t$ over the nonzero residuals. The median-absolute-deviation scaling is the standard robust choice \citep{Huber(81)}.
For the end effects $t\le T_{60}$, the trades before 9.31am, \& at $t=0$, we set $b_t$ as in Section~\ref{sect:preav-ends}. For our 12 different days, $T_{60}$ ranges from roughly 16,200 to 34,900 trades. The scale filter is computed using a weighted median through a running order-statistic \citep{Fenwick(94)} binary indexed tree.  This is computationally convenient as the residuals from the median filter have a relatively small number of points of support. 
\begin{figure}[h!]
    \centering
    \includegraphics[width=0.85\linewidth]{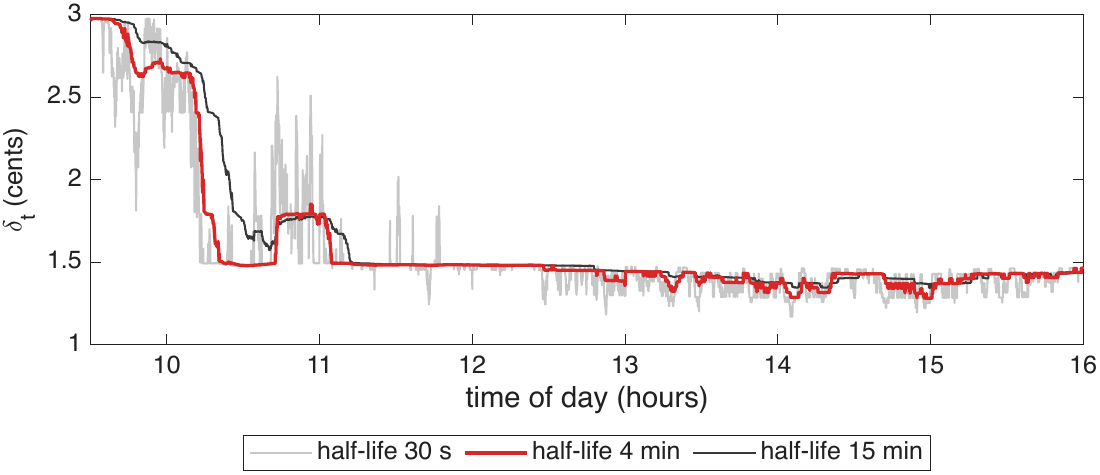}
    \caption{The Huber threshold $\delta_t$ in cents for \texttt{NVDA} on 16 Oct 2024, with half-lives of 30 seconds, 4 \& 15 minutes, each starting $b_t$ from the previous 5 openings \& swapping in the day's own trades for first minute, so all 3 start at $2.97$ cents. The 4-minute half-life, in red, is used.}
    \label{fig:preav_delta_hl}
\end{figure}

A slightly simpler alternative would be the weighted sample mean of the absolute residuals, but the residuals turn out to have such heavy tails (Section~\ref{sect:preav-path}) that the variance of the absolute residuals is infinite, implying such a sample mean would be fragile.   

We set the Huber threshold
$\delta_t = 2 \times 1.4826 \times b_t,$ for $t \in \{0,1,...,T\},$ 
following the scale of the noise through the day. $b_t$ \& so $\delta_t$, are about twice as large in the first minutes of trading as during the rest of the day, while \citet{MyklandZhang(16)} hold the threshold fixed.\footnote{The factor $1.4826=1/\Phi^{-1}(3/4)$ makes the median absolute deviation a consistent estimator of a Gaussian standard deviation (s.d.); with the zeros dropped, $1.4826\,b_t$ is a scale calibration rather than a s.d. estimator. Under Gaussian noise, a threshold of two s.d.s gives the Huber estimator of location an asymptotic variance about $1\%$ above the sample mean's, $99\%$ efficiency relative to the mean \citep{Huber(81)}.}

The 4-minute half-life, which determines $\lambda$, is set by prediction, using the check loss of $b_{t,60}$, the scale at the most recent trade at least a minute before trade $t$, for the nonzero absolute residual $|\tilde{\epsilon}_t|$. Pooled across days, the loss is smallest at about 2 minutes \& is almost flat around that, within 1\% of its smallest value from 30 seconds to 15 minutes. Within that flat range the choice is about how much $\delta_t$ moves through the day. A 30-second scale moves about nine times as much as the 4-minute one, while a 15-minute scale is slow to follow changes in the level of the noise. Figure~\ref{fig:preav_delta_hl} shows $\delta_t$ at 3 half-lives.

Figure~\ref{fig:preav_delta} shows $\delta_t$ across the 12 trading days. On each day, $\delta_t$ is approximately 1.5 cents for most of the day but starts off higher, at around $3$ cents near the open where trading is most intense. $\delta_{1:T}$ is computed before the Huber filter runs, from the median-filter residuals alone. 

\begin{figure}[h!]
    \centering
    \includegraphics[width=0.85\linewidth]{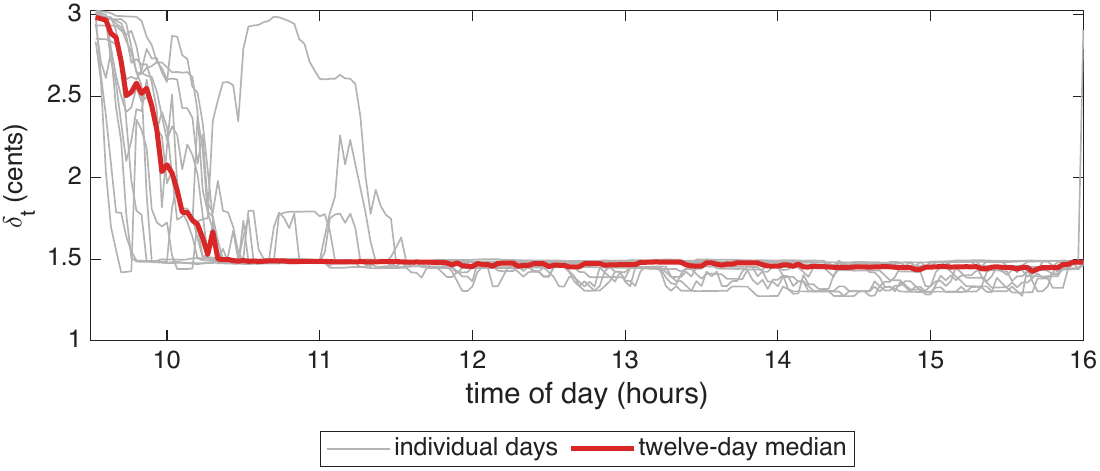}
    \caption{Huber threshold $\delta_t$ in cents through the trading day, for 12 \texttt{NVDA} trading days in Oct 2024 (grey lines) and their median (red line).}
    \label{fig:preav_delta}
\end{figure}

\subsection{Huber filter applied to prices}\label{sect:preav-path}

The third filter, $\theta_t$, is the Huber filter applied to the prices  
$$\theta_t = \underset{\theta \in \Theta}{\arg }\min \  \sum_{j=0}^{t-1} w_{t,j}\, H_{\delta_{t-j-1}}(Y_{t-j}-\theta),\quad w_{t,j} \propto \lambda ^j,\quad t\in \{2,...,T-1\},\quad \theta_1=Y_1, \quad \theta_T=Y_T,
$$ 
using it to estimate $X_t$. We write $\hat{X}_t$ as any filter of $X_t$, so sometimes we will use that instead of $\theta_t$ in the text. We compute $\theta_t$ by iteratively reweighted least squares, which here is just a weighted mean recomputed until it moves by less than $10^{-12}$, starting from the median filter $m_t$. A trade within the threshold $\delta_{t-j-1}$ enters with its weight $w_{t,j}$, and one beyond it has that weight scaled down by $\delta_{t-j-1}/|Y_{t-j}-\theta|$. The Huber filter is computed using the stratified simulation estimator of Remark~\ref{rem:strat}(d) ($h = W_t^{-1}(0.999)$ and $B=100$), where each tail draw at lag $j=j_b$ carries its own trade's threshold $\delta_{t-j-1}$.

$X_t$ is close to a martingale, so we expect a filter with a short half-life. We choose the half-life by minimizing the Huber loss of the $k$-trade-ahead forecast of the price, $H_{\delta_{t-k}}(Y_t-\theta_{t-k})$, pooled across days. The minimizing half-life is slightly more than 2 trades for every $k$ from 2 to 100, and the loss changes by less than $2.5\%$ between half-lives of 2 and 3. We use 3 trades, as this gives the larger effective sample $(1+\lambda)/(1-\lambda)$ of about 9 trades rather than 6.

Figure~\ref{fig:preav_path} plots the price of trades together with the linear \& Huber filters.  The top display shows the result over the full day. There are two extraordinary outlier trades, one just after 11am and one at 12:43 --- there are many other outlier trades, but they are not visible in this picture.  The other 3 plots are three $5$-second snapshots.  

The 2nd picture shows a typical stretch of trading. Away from outliers the paths track each other closely. The linear filter briefly darts downwards, hit by modest outliers. 

The 3rd \& 4th pictures are the most interesting.  The third uses a 4-cent scale \& adds the median filter, showing results at the tick scale, with most traded prices sitting at the cents \& the half-cent values. The linear \& the Huber filter coincide, running between those levels, while the median filter steps across them (which makes it a poor estimator of the efficient price). This is the case where there are no impactful outliers which challenge the linear filter.

The 4th picture shows the impact of the extreme \$118.85 trade, an order of $85$ shares, about $\$10{,}000$ and a typical size that day (recall the median-sized trade is around $50$ shares), $11.5\%$ below the local price level of about \$134, the linear filter is pulled down by about \$3.4 and recovers as the trade's weight in the average decays, within about half a second. The Huber is barely affected by the trade.  This shows that the linear filter is simply unacceptable for use with high frequency financial data of this kind.    

\begin{figure}[h!]
    \centering
    \includegraphics[width=0.90\linewidth]{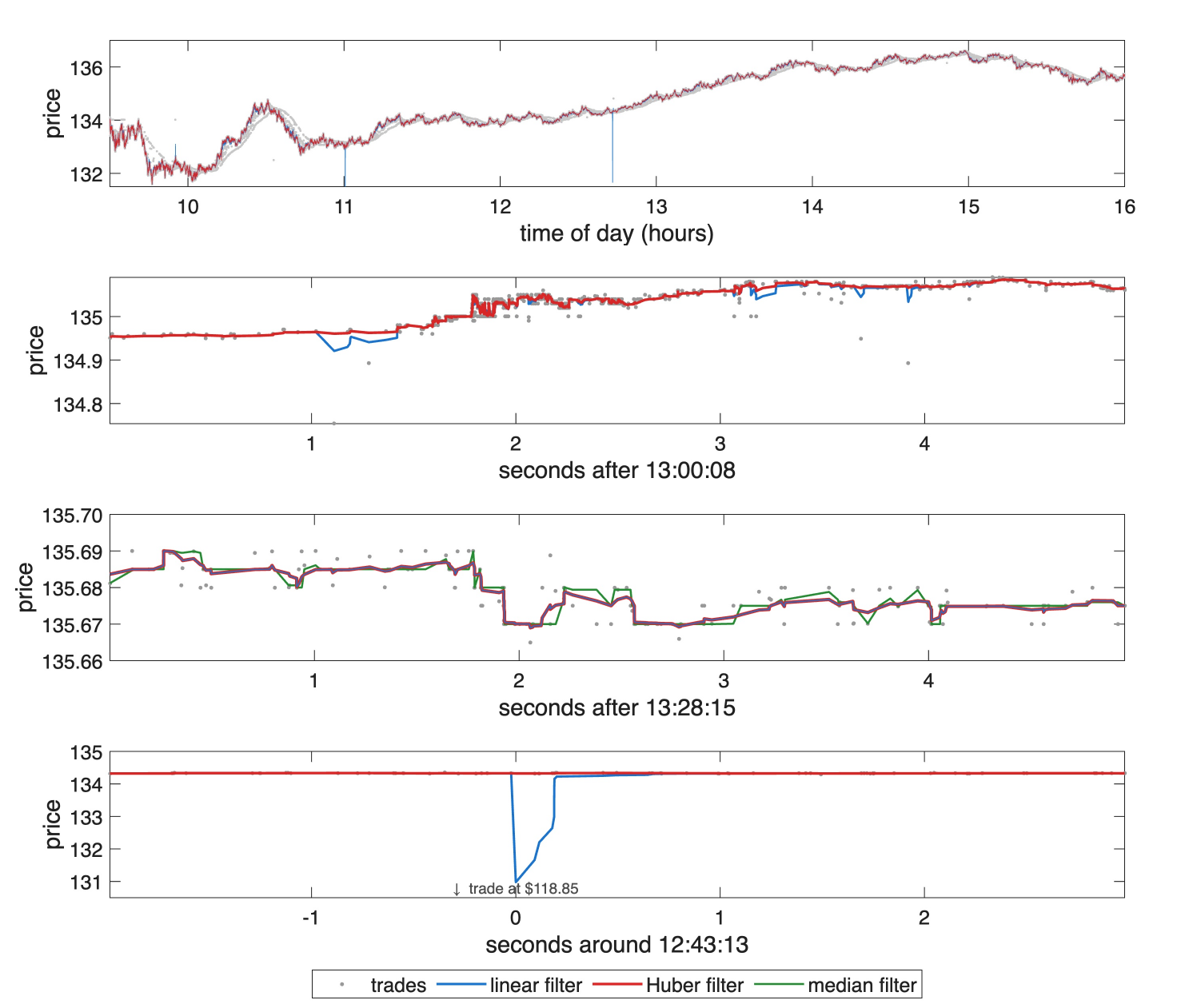}
    \vspace{-1mm}
     \caption{\texttt{NVDA} trades (grey points), the linear filter and the Huber filter, 16 Oct 2024. Top: full trading day. 2nd: $5$-second window at 13:00:08. 3rd: $5$-second window at 13:28:15 with the median filter $m_t$ added; the linear filter is drawn thin and appears on top of the Huber filter. Bottom: $5$-second window around the \$118.85 trade.}
    \label{fig:preav_path}
\end{figure}  

Table~\ref{tab:preav_diffs} reports summary statistics of the $t$-th trade $k$-th periodic difference of the estimates
$$\hat{X}_t-\hat{X}_{t-k},\quad k\in \{1,...,t-1\},\quad t \in \{2,3,...,T\}
$$ 
for the 3 filtered paths and the trade prices for $k$ up to 25. The filters move far less than the trades. One trade apart, the root mean squared difference is $3.00$ cents for the trades, and $0.43$ for the Huber path. The largest moves dominate these, and they differ across the columns. The \$118.85 trade is $30\%$ of the summed squared 1-trade differences of the linear filter and next to nothing of the Huber filter's, whose largest move, a fast repricing (a jump in the price level) at 09:45, is $3.8\%$. The first-order autocorrelation of the 1-trade differences is $-0.48$ for the trades, the bid-ask bounce. The linear filter removes it, $-0.02$, the median and the Huber keep $-0.24$ and $-0.18$ from the penny grid, and past a few lags are all near zero. The 3 filters share the same weights: it is the loss, not weights, driving these results. 

The Hill estimator of the tail exponent $\zeta$ of the $k$-trade returns is reported.  $\zeta$ controls how the tail probability decays; small values mean heavy tails, and below two the variance is infinite. Every Hill estimate in the table is below two, except for the trades at $k=25$; the Huber filter is the lowest. Between repricings the Huber filter barely moves, while at a repricing it moves in full, so its large differences are the biggest relative to its typical tiny ones.

\begin{table}[h!]
\centering
\begin{tabular}{@{}crrrrrrrr}
& \multicolumn{4}{l}{Mean of $\frac{|\hat{X}_t-\hat{X}_{t-k}|}{\sqrt{k}}$ in cents} & 
\multicolumn{4}{l}{Sqrt of Mean of $\frac{(\hat{X}_t-\hat{X}_{t-k})^2}{k}$ in cents} \\
\cmidrule(r){2-5} \cmidrule(r){6-9}
$k$ & median & Huber & linear & trades & median & Huber & linear & trades \\
\cmidrule(r){1-1} \cmidrule(r){2-2}
\cmidrule(r){3-3} \cmidrule(r){4-4} \cmidrule(r){5-5} 
\cmidrule(r){6-6}
\cmidrule(r){7-7} \cmidrule(r){8-8} \cmidrule(r){9-9}\\
1 & 0.115 & 0.090 & 0.106 & 0.442 & 0.537 & 0.431 & 0.493 & 2.995\\
2 & 0.114 & 0.103 & 0.125 & 0.362 & 0.467 & 0.389 & 0.487 & 2.159\\
5 & 0.136 & 0.123 & 0.150 & 0.281 & 0.418 & 0.366 & 0.460 & 1.405\\
10 & 0.151 & 0.139 & 0.161 & 0.241 & 0.382 & 0.348 & 0.418 & 1.018\\
25 & 0.162 & 0.153 & 0.166 & 0.207 & 0.333 & 0.315 & 0.351 & 0.682\\
[4pt]
& \multicolumn{4}{l}{ACF of $\hat{X}_t-\hat{X}_{t-1}$ at lag $k$} & \multicolumn{4}{l}{Tail index of $|\hat{X}_t-\hat{X}_{t-k}|$} \\
\cmidrule(r){2-5} \cmidrule(r){6-9} 
1 & $-$0.24 & $-$0.18 & $-$0.02 & $-$0.48 & 1.73 & 1.64 & 1.74 & 1.77\\
2 & 0.01 & 0.01 & $-$0.04 & $-$0.01 & 1.60 & 1.48 & 1.91 & 1.80\\
5 & $-$0.02 & $-$0.02 & $-$0.03 & 0.00 & 1.51 & 1.36 & 1.94 & 1.83\\
10 & $-$0.01 & $-$0.01 & $-$0.01 & 0.00 & 1.41 & 1.31 & 1.79 & 1.89\\
25 & $-$0.01 & $-$0.01 & 0.00 & 0.00 & 1.55 & 1.48 & 1.67 & 2.10\\
\end{tabular}
\caption{Summary of $\hat{X}_t-\hat{X}_{t-k}$ over $k$ trades of 4 price paths $\hat{X}_t$ for \texttt{NVDA} on 16 Oct 2024. $\hat{X}_t$ is the output from the filters for trade $t$, the median filter $m_t$, the Huber filter $\theta_t$ or the linear filter, each at the 3-trade half-life, or the raw trade price $Y_t$. 1st panel is divided by $\sqrt{k}$ and the second by $k$. The ACFs are the lag-$k$ sample ACFs of the 1-trade difference $\hat{X}_t-\hat{X}_{t-1}$. 
The Hill estimate for $|\hat{X}_t-\hat{X}_{t-k}|$ averages the Hill estimates computed over the largest $0.1\%$ to $1\%$ of observations.}
\label{tab:preav_diffs}
\end{table}

Table~\ref{tab:preav_diffs_time} repeats the exercise with the lag in time rather than trades, comparing the filtered price $\hat{X}_t$ with its own value $\hat{X}_{t,s}$ which happens at time $\tau_{t,s}=\max\{\tau_j: \tau_t-\tau_j \ge s\}$ at least $s$ seconds earlier. The resulting returns over $\tau_t - \tau_{t,s}\ge s$ seconds are 
$$
\hat{X}_t-\hat{X}_{t,s},\quad s>0,\quad t \in \{2,3,...,T\}.
$$

The first panel divides by $\sqrt{\tau_t-\tau_{t,s}}$. The second reports the subsampled, preaveraged \& bias correction annualized volatility measure 
\begin{align}\label{eqn:annVol}
\hat{\sigma}_{\text{Annual};s} = \sqrt{\frac{23{,}400}{23{,}400-\tau_1}} 100 \sqrt{252} \sqrt{\sum_{t=2}^T (\tau_t - \tau_{t-1}) c_{t,s} \frac{(\hat{X}_t-\hat{X}_{t,s})^2}{\tau_t - \tau_{t,s}}}
\end{align}
which may vary with $s$. The latter should be, on average, invariant to $s$, $T$ and $\tau_{1:T}$ under scaled Brownian motion prices. 
The $c_{t,s}=\mathbb{E}[(B_t-B_{t,s})^2]/{\mathbb{E}[(\hat{B}_t-\hat{B}_{t,s})^2]}$ is a finite $T$ correction for the damping of the variation of the signal caused by filtering.  Brownian motion $B\ind \tau_{1:T}$ is evaluated at the observed times $\tau_{1:T}$, yielding  $B(\tau_1),...,B(\tau_T)$ as the data.  The subsequent filtered price $\hat{B}_{1:T}$ is outputted by the Huber filter at the half-life and the threshold path $\delta_{1:T}$ of the day. Section~\ref{app:c} shows how to compute $c_{t,s}$ using simulation.

The correction is helpful because the filtered price is a robust average of the last few trades, so if only a few trades lie between the two return ends  $\hat{X}_{t,s},\hat{X}_t$, the two averages are built from mostly the same trades \& vary less than $X_t-X_{t,s}$ does, \& $c_{t,s}>1$. If $\hat{B}=B$, then $c_{t,s}=1$, while as $T$ goes to infinity $c_{t,s}\rightarrow 1$ for fixed $s$ for well behaved sequences of trade times. On 16 October the correction raises the Huber filter's estimate by $4.9\%$ in volatility at $s=1$ second and by $0.5\%$ at $s=10$ seconds.  

The differences $\tau_t-\tau_{t-1}$, telescope to $23{,}400-\tau_1$ (recall $\tau_1$ is the seconds into the day the market opening price $X_1$ is posted --- typically $\tau_1$ is around 1), so put the volatility on the scale of a trading day. Multiplying by $100\sqrt{252}$ makes this the annualized volatility in percent. The differences in equation (\ref{eqn:annVol}) are taken in the log of the filtered price, so $\hat{\sigma}_{\text{Annual};s}$ is in percent, while the first panel of Table~\ref{tab:preav_diffs_time} uses the price itself, in cents. The square rooted first ratio (\ref{eqn:annVol}) is so close to 1 it  could be dropped to avoid clutter.  

In both panels the trades column lies well above the filtered columns at small $s$. At this trading rate a lag of $0.1$s already spans about $7$ trades, and far more at the start of the day. At that $0.1$s lag the annualized volatility of the trades is roughly five times the Huber filter's. The difference decreases as the lag grows, and at $15$ minutes the 4 paths agree to within a fraction of a percent, so for price changes beyond a few minutes it hardly matters if one filters at all.  This result was apparent to the early realized volatility researchers \citep{AndersenBollerslevDieboldLabys(01),BarndorffNielsenShephard(02realised)}, who used returns measured over minutes and justifiably ignored the impact of noise on their estimators. 

The $0.1$s returns have negative correlations across all 4 methods, very large for the raw trades and quite large for the linear filter, $-0.40$; the Huber filter's is $-0.24$.  By 1 second the correlation between adjacent returns is near zero for all the filtered paths while still $-0.37$ for the raw trades, so the negative dependence in Table~\ref{tab:preav_diffs} lives at the trade scale.

For the filtered paths the Hill index rises from about 2 at 0.1 second to about 8 at 10 seconds. By then the fat tail is driven by the efficient price move rather than the noise. At 15 minutes the largest $1\%$   Huber difference is \$2 repeated 1000s of times, so beyond 10 seconds there is too little data in the tail to estimate it just using a day's data.      

\begin{table}[h!]
\centering
\begin{tabular}{@{}crrrrrrrr}
& \multicolumn{4}{l}{Mean of $\frac{|\hat{X}_t-\hat{X}_{t,s}|}{\sqrt{\tau_t-\tau_{t,s}}}$ in cents} & 
\multicolumn{4}{l}{Annualized volatility: $\hat{\sigma}_{\text{Annual};s}$} \\
\cmidrule(r){2-5} \cmidrule(r){6-9}
$s$ & median & Huber & linear & trades & median & Huber & linear & trades \\
\cmidrule(r){1-1} \cmidrule(r){2-2}
\cmidrule(r){3-3} \cmidrule(r){4-4} \cmidrule(r){5-5} 
\cmidrule(r){6-6}
\cmidrule(r){7-7} \cmidrule(r){8-8} \cmidrule(r){9-9}\\
0.1 sec & 2.50 & 2.34 & 2.58 & 3.54 & -- & 53.3 & 117.5 & 245.2\\
1 sec & 1.95 & 1.91 & 1.98 & 2.20 & 34.8 & 34.3 & 43.1 & 77.9\\
5 sec & 1.95 & 1.94 & 1.96 & 2.03 & 34.4 & 34.3 & 35.9 & 44.5\\
10 sec & 1.98 & 1.97 & 1.98 & 2.02 & 34.2 & 34.1 & 35.1 & 46.6\\
1 min & 1.93 & 1.93 & 1.92 & 1.94 & 33.4 & 33.3 & 33.4 & 35.7\\
15 min & 2.06 & 2.06 & 2.06 & 2.06 & 37.1 & 37.1 & 37.1 & 37.2\\
[4pt]
 & \multicolumn{4}{l}{$\mathrm{Cor}(\hat{X}_t-\hat{X}_{t,s}, \hat{X}_{t,s}-\hat{X}_{t,2s})$} &
 \multicolumn{4}{l}{Hill tail index of $|\hat{X}_t-\hat{X}_{t,s}|$}\\
 \cmidrule(r){2-5} \cmidrule(r){6-9}
0.1 sec & $-$0.26 & $-$0.24 & $-$0.40 & $-$0.81 & 1.83 & 1.81 & 2.05 & 2.31\\
1 sec & 0.05 & 0.07 & 0.00 & $-$0.37 & 3.32 & 3.38 & 3.81 & 2.64\\
5 sec & 0.08 & 0.09 & 0.06 & $-$0.03 & 7.57 & 7.63 & 6.62 & 4.87\\
10 sec & 0.04 & 0.04 & 0.03 & $-$0.12 & 8.42 & 8.44 & 8.16 & 6.19\\
1 min & 0.06 & 0.06 & 0.06 & $-$0.02 &&&&\\
\end{tabular}
\caption{Summary of $\hat{X}_t-\hat{X}_{t,s}$ over time lag $s$ of 4 price paths for \texttt{NVDA} on 16 Oct 2024. $\hat{X}_t$ is 1 of the 4 paths of Table~\ref{tab:preav_diffs} \& $\hat{X}_{t,s}$ is its value at the most recent trade at least $s$ seconds before trade $t$; when there is no trade at least $s$ seconds earlier that day, the reference is $\hat{X}_1$, the opening price. 1st panel is divided by $\sqrt{\tau_t-\tau_{t,s}}$. The median filter almost does not move at $s=0.1$, which makes the denominator of its correction $c_{t,s}$ close to 0. 3rd panel is the correlation of the return over the last $s$ seconds, $\hat{X}_t-\hat{X}_{t,s}$, with the return over the adjacent earlier window, $\hat{X}_{t,s}-\hat{X}_{t,2s}$, where $\hat{X}_{t,2s}$. The Hill estimate for $|\hat{X}_t-\hat{X}_{t,s}|$ averages the Hill estimates computed over the largest $0.1\%$ to $1\%$ of observations.}
\label{tab:preav_diffs_time}
\end{table}

\subsubsection{The noise can be very heavy tailed: evidence and implications}\label{sect:preav-tails}

Importantly from a theoretical perspective, the residual 
$\hat{\epsilon}_t = Y_t-\theta_t,$ for $ t\in \{1,...,T\},$ from the Huber filter, our proxy for the noise $\epsilon_t = Y_t-X_t$, has an infinite variance for the \texttt{NVDA} data. The Hill estimate of the tail index of the residual is about $1.4$, below the value $2$ above which the variance is finite.  The Appendix \ref{sect:tailindex} shows this conclusion holds up for \texttt{NVDA} data against various robustness checks, while Appendix \ref{app:xsec} shows it commonly holds for the most frequently traded U.S. stocks, but not for stocks with lower trading intensity.     

For a traditional model of the extreme tail $P(|\hat{\epsilon}_t|>x)\propto x^{-\zeta},$ with $ t\in \{1,...,T\}$ the variance is finite only if $\zeta>2$. The Hill estimator \citep{Hill(75)} of $\zeta$ from the $k$ largest residuals $|\hat{\epsilon}|_{(1)}\ge\cdots\ge|\hat{\epsilon}|_{(k+1)}$ is $\hat\zeta(k)=\bigl(\tfrac{1}{k}\sum_{i=1}^{k}\log(|\hat{\epsilon}|_{(i)}/|\hat{\epsilon}|_{(k+1)})\bigr)^{-1}$. Over the far tail, the largest $0.1\%$ to $1\%$ of residuals, $\hat\zeta$ is between roughly $1.0$ and $2.0$ across that range and averages $1.43$, holding under removal of the two largest residuals ($1.43$) and on 12 random subsamples of $25{,}000$ trades, where the index ranges from $1.33$ to $1.49$. 

The residual sums point the same way. The largest residual holds $35\%$ of $\sum_t \hat{\epsilon}_t^2$ but $0.3\%$ of $\sum_t |\hat{\epsilon}_t|$, and the largest $1{,}000$ ($0.07\%$ of the trades) hold $85\%$ against $7\%$, so a few trades dominate the sum of squares while the sum of absolute values stays spread, characteristic of an infinite variance with a finite mean \citep{EmbrechtsKluppelbergMikosch(97)}.

Robust preaveraging mitigates the microstructure noise by averaging, then subtracts a noise-bias term from the weighted sum of squared filtered returns. For a score $\psi=L'$, the derivative of the loss, the term is (e.g.~equation~(8) of \citet{MyklandZhang(16)}),
\begin{equation}\label{eqn:psi*}
\frac{\operatorname{Var}\!\left(\psi(\epsilon_t)\right)}{\bigl(\mathbb{E}\!\left[\psi'(\epsilon_t)\right]\bigr)^2},\quad t\in \{1,...,T\},
\end{equation}
under their assumption of independent noise; for stationary mixing noise they replace the numerator by the long-run variance of $\psi(\epsilon_t)$. For the linear filter, equation (\ref{eqn:psi*}) is $\operatorname{Var}(\epsilon_t)$, the noise variance. When that variance does not exist, the noise-bias correction for standard preaveraging does not exist either (formally invalidating that approach), and the sample variance of the residuals that estimates it grows without bound as the sample size increases. Because the Huber score $\psi_\delta(u)=\max\{-\delta,\min(\delta,u)\}$ is bounded by $\delta$, the numerator $\operatorname{Var}(\psi_\delta(\epsilon_t))\le\delta^2$ is finite and the denominator is positive, so (\ref{eqn:psi*}) is finite.

Table \ref{tab:errors} provides summaries of residuals based on the median, Huber \& linear filters of prices.  By construction, the median \& Huber cases have an infinite tail index.  

\begin{table}[h!]
\centering
\begin{tabular}{@{}crrrrrrrrrr}
& \multicolumn{3}{l}{Mean of $|\hat{\epsilon}_t|$ in cents} & & 
\multicolumn{3}{l}{s.d. of $\hat{\epsilon}_t$ in cents} & 
\multicolumn{3}{l}{Tail index of $\psi_{\delta_{t-1}}(\hat{\epsilon}_t)$}\\
\cmidrule(r){2-4} \cmidrule(r){6-8} \cmidrule(r){9-11}
& median & Huber & linear & & median & Huber & linear  & median & Huber & linear\\
\cmidrule(r){2-2}
\cmidrule(r){3-3} \cmidrule(r){4-4} 
\cmidrule(r){6-6}
\cmidrule(r){7-7} 
\cmidrule(r){8-8} 
\cmidrule(r){9-9}
\cmidrule(r){10-10} \cmidrule(r){11-11}
  & .30 & .38 & .41   & & 2.18 & 2.16 & 1.84   &  $\infty$ & $\infty$ & 1.74\\
  &&&&&&&&&&\\
[4pt]
$j$ & \multicolumn{3}{l}{ACF of $\psi_{\delta_{t-1}}(\hat{\epsilon}_t)$} & $\alpha$ & 
 \multicolumn{3}{l}{$100(1-\alpha)$\% C.I. for $\theta_t-X_t$} &
 \multicolumn{3}{l}{$\psi^*$: se$(\theta_t-X_t)$ in cents} \\
\cmidrule(r){1-1} \cmidrule(r){2-4} \cmidrule(r){5-5}
\cmidrule(r){6-8} \cmidrule(r){9-11} 

1  & .20 & .28 & $-$.02 & .20 & ---  & $-$.26,.26 & $-$.28,.28 & --- & .22 & $\infty$(.56)\\
2  & .08 & .17 & $-$.04 & .10 & --- & $-$.36,.36 & $-$.40,.40 &&&\\
5  & .02 & .06 & $-$.03 & .05 & --- & $-$.46,.46 & $-$0.54,0.54 &&&\\
10 & .00 & .02 & $-$.01 & .01 & --- & $-$.72,.75 & $-$1.23,1.15 &&&\\
25 & .00 & .00 & .00 &&&&&&&\\
\end{tabular}
\caption{Summary of the residuals $\hat{\epsilon}_t = Y_t - \hat{X}_t$ of the 3 filters for \texttt{NVDA} on 16 Oct 2024. The score is the subgradient of loss: $\mathrm{sign}(\hat{\epsilon}_t)$ for the median, the clipped $\psi_{\delta_{t-1}}(\hat{\epsilon}_t)$ for the Huber \& $\hat{\epsilon}_t$ for the linear filter. The median \& Huber scores are bounded, yielding infinite tail indices. The Hill estimator of the index averages the largest $0.1\%$ to $1\%$ of $|\psi(\hat{\epsilon}_t)|$. The s.e. of linear filter does not exist, as noise s.d. does not.  The sample value is in brackets.  
For the median neither the quantiles nor the s.e. are helpful due to discreteness of prices. 
}
\label{tab:errors}
\end{table}

Suppose $\theta_t$ is unique, as it is for the Huber loss when at least one trade with positive weight lies within its threshold of $\theta_t$. Then in the $\alpha=1$ case
\begin{align*}
P(\theta_t > X_t + \eta) &= P(\tilde{g}_t(\eta)> 0),\quad \tilde{g}_t(\eta) = \sum_{j=0}^{t-1} w_{t,j} \psi(Y_{t-j}-X_t - \eta) 
\end{align*}
For small $j$, the $X_t$ is close to $X_{t-j}$ as the $X$ process evolves in calendar time.  If the loss is continuously differentiable then $X_{t-j}$ should be close to $\hat{X}_{t-j}$ (discreteness of prices will impact the median filter, causing these two items to markedly differ).  So 
\begin{align*}
P(\theta_t - X_t > \eta) & \approx P(g_t(\eta)> 0),\quad g_t(\eta) = \sum_{j=0}^{t-1} w_{t,j} \psi(\hat{\epsilon}_{t-j}- \eta) \\
& \approx \frac{1}{T} \sum_{t=1}^T 1(g_t(\eta) > 0),
\end{align*}
yielding an approximation to the quantiles of $\theta_t-X_t$ which are used for confidence intervals in Table \ref{tab:errors}.  In some contexts, it might be useful to localize the average of indicators, allowing the quantiles to change through time.  The linear filter has wider confidence intervals for the estimation error than the Huber filter at every confidence level, by about a 0.1 at $80\%$ levels, 0.2 at $95\%$ and 0.6 at $99\%$.  For the median filter this approach to confidence intervals fails due to the discreteness of prices: algorithmically it would report 0,0 at $80\%$ levels, while -.20,.21 at $95\%$, \& -.51,1.00 at $99\%$, but these numbers are nonsense.      

It is tempting to use asymptotics to approximate the distribution of the error, which would work for the Huber filter yielding an approximate variance under strict stationarity, with the weights truncated at lag $h$: 
$$
\psi^{*2} = \frac{\sum_{j=0}^{h} \sum_{i=0}^{h} w_{t,j} w_{t,i} \gamma_{j-i}(\psi)}{\mathbb{E}[\psi'(\epsilon_t)]^2},
$$
where $\gamma_{j}(\psi)$ is the autocovariance function of the  $\psi(\hat{\epsilon}_1),...,\psi(\hat{\epsilon}_T)$ sequence (when the residual is i.i.d. and the weights are uniform, this produces the expression given in \cite{MyklandZhang(16)}).  We evaluate $\psi^{*2}$ with $h=W_t^{-1}(0.999)=29$, the head length of Remark~\ref{rem:strat}(d), and, since $\psi_{\delta}'(u)=1(|u|\le\delta)$, estimate $\mathbb{E}[\psi'(\epsilon_t)]$ by the share of residuals within the threshold, $\frac{1}{T}\sum_{t=1}^{T}1(|\hat{\epsilon}_t|\le \delta_{t-1})$.
But for the linear and median filters this approach for the standard error fails.  In the median filter case this is due to the discrete nature of the errors.  The linear filter crashes --- we will report an estimate of $\psi^{*}$, but its theoretical value is $\infty$ due to the heavy tails.  Section \ref{sect:TVpsi} reports empirical results for the slowly varying $\psi^*_{1:T}$ for the inputted Huber filter of the price.  It roughly halves through the day from its peak at the market opening.

\subsection{Volatility estimation using robust preaveraging}\label{sect:preav-vol}

Following the logic of the pre-averaged returns measure of variability  \citep{JacodLiMyklandPodolskijVetter(07)}, we replace their linear filter of prices with robust filters of prices.   This is inspired by \cite{MyklandZhang(16)}, but our details differ and our focus is on filtering.  We use the output of the robust filter on prices as an input into computing an estimator of the integrated volatility of the full day, and a filtered spot volatility through the day.

Again write $\hat{X}_t$ as the filtered (or raw) log-price after the $t$-th trade and $\hat{X}_{t,s}$ as the series value at the time of the most recent trade $s>0$ seconds before trade $t\in \{1,...,T\}$, taking $\hat{X}_{t,s}=\hat{X}_1$ when there is no such trade.   Then 
\begin{align*}
r_{t,s} & = \hat{X}_t - \hat{X}_{t,s},\quad s>0, \quad t \in \{2,3,...,T\},\\
& = X_t - X_{t,s} + (\hat{X}_t - X_t) - (\hat{X}_{t,s} - X_{t,s})
\end{align*}
is the corresponding filtered return over the last $\tau_t - \tau_{t,s}\ge s$ seconds. Hence $r^2_{t,s}$ minus
$$
\{[X,X](\tau_{t})-[X,X](\tau_{t,s})\} +2 \psi^{*2}_t
$$
has a conditional mean which is roughly zero given data up to time $\tau_{t,s}$, following the logic of \cite{JacodLiMyklandPodolskijVetter(07)} \& treating the 2 filter errors as uncorrelated, which is reasonable once the 2 filter windows share no trades, from about $0.5$ seconds on here. Thus the filter's $\psi^{*}_t$ induces a potentially meaningful bias to squared return volatility measures, which typically measure financial variability through $[X,X](\tau_{t})-[X,X](\tau_{t,s})$ type objects \citep[e.g.][]{AndersenBollerslevDieboldLabys(01),BarndorffNielsenShephard(02realised)}.     

In Table \ref{tab:errors} the Huber filter's estimated standard error $\psi_t^*$ is around 0.22 cents.  Using the $s=15$ min result of Table \ref{tab:preav_diffs_time} the volatility on 16 October is roughly $37\%$ annualized, which is $2.04\sqrt{s}$ cents over $s$ seconds. Hence, in terms of standard deviation type terms 
\begin{align}\frac{\sqrt{\{[X,X](\tau_{t})-[X,X](\tau_{t,s})\} +2 \psi^{*2}_t}}{\sqrt{[X,X](\tau_{t})-[X,X](\tau_{t,s})}} &= \sqrt{1 + \frac{2 \psi_t^{*2}}{[X,X](\tau_{t})-[X,X](\tau_{t,s})}} \nonumber \\
& \approx  1 + \frac{\psi_t^{*2}}{[X,X](\tau_{t})-[X,X](\tau_{t,s})},\label{eqn:scale}
\end{align}
deploying a binomial approximation.  
At the 1s level, the bias caused by the Huber filter's noise is around 1.2\% of the average level of volatility on 16 October (the corresponding result for the linear filter is 7.5\%, optimistically setting $\psi_t^*$ to  0.56 from Table \ref{tab:errors}). At the 0.1s level it would be around 12\% (around 75\% in the linear case).  During (i) periods of low volatility, (ii) periods with large market microstructure effects, (iii) taking $s\ll 1$, then the bias is likely to be empirically important \& it would be worthwhile to bias correct for the impact of the noise by dividing an initial spot volatility estimator based on $r_{t,s}$ by the left hand side of equation (\ref{eqn:scale}). See Appendix~\ref{app:robust} for several robustness checks on these results.

\subsubsection{Integrated volatility}

Using the $
r_{t,s}$ returns a daily integrated volatility estimator \citep{AndersenBollerslevDieboldLabys(01),BarndorffNielsenShephard(02realised)} with pre-robust-averaging \citep{JacodLiMyklandPodolskijVetter(07),MyklandZhang(16)} is, for a specific ``day'' of $23{,}400$ seconds   
$$
\hat{\sigma} := \hat{\sigma}^*(23{,}400),
$$
where (writing the sum forward \& backward in time to get two useful forms)
\begin{align*}
\hat{\sigma}^{2*}(u) &=  \sum_{t=2}^{T_u} (\tau_t-\tau_{t-1}) c_{t,s} \frac{r_{t,s}^2}{\tau_t-\tau_{t,s}},\quad u \in [0,23{,}400] \\  
&= \sum_{j=0}^{T_u-2} Z_{T_u-j,s}, \quad \text{where} \quad Z_{t,s} = (\tau_t-\tau_{t-1})c_{t,s} \frac{r_{t,s}^2}{\tau_t-\tau_{t,s}},  \quad Z_{1,s} := 0,
\end{align*}
a subsampled estimator of the price's $[X,X](u)$, recalling $T_u$ is the number of trades up to time $u$.    
The volatility estimator $\hat{\sigma}$ is scaled to make the volatility approximately invariant to $s$ under a Brownian motion model \& the differences $\tau_t-\tau_{t-1}$, summing to $23{,}400-\tau_1$ seconds, put the volatility on the scale of a trading day, so $\hat{\sigma}100\sqrt{252}$ is the annualized volatility in \%. This is the estimator used in (\ref{eqn:annVol}). In principle the estimator should be upscaled by $23{,}400/(23{,}400-\tau_1)$, but this term is so close to 1 in practice we avoid clutter by ignoring it below.      

Plotting the integrated volatility estimators, as a function of $s$, gives a preaverage version of the ``volatility signature'' plot.  The second part of Table \ref{tab:preav_diffs_time} shows this for a single day of data, evaluating $s$ at 6 different values. We will now see much more extensive results.  

Figure~\ref{fig:preav_signature} plots this annualized integrated volatility against $s$ for each of the 12 trading days, for the raw trades, the linear filter \& the Huber filter, with $\hat{X}_{t,s}$ as in Table~\ref{tab:preav_diffs_time}. The trade curves rise steeply below 1 second, to between about $2.5$ \& $7$ times the Huber level at a 0.1 second, depending on the day. The filtered curves rise far less and spread far less across days; from half a second on the Huber curves are nearly flat in $s$. As $s$ grows the gaps close, \& at 15 minutes the 3 curves of each trading day agree, at levels between $20$ and $57\%$ across the 12 days.

Our conclusion from this work is that it is conservative to view the volatility signature plot as flat for $s \ge 1$ when using the Huber filtered prices as inputs.

\begin{figure}[h!]
    \centering
    \includegraphics[width=0.95\linewidth]{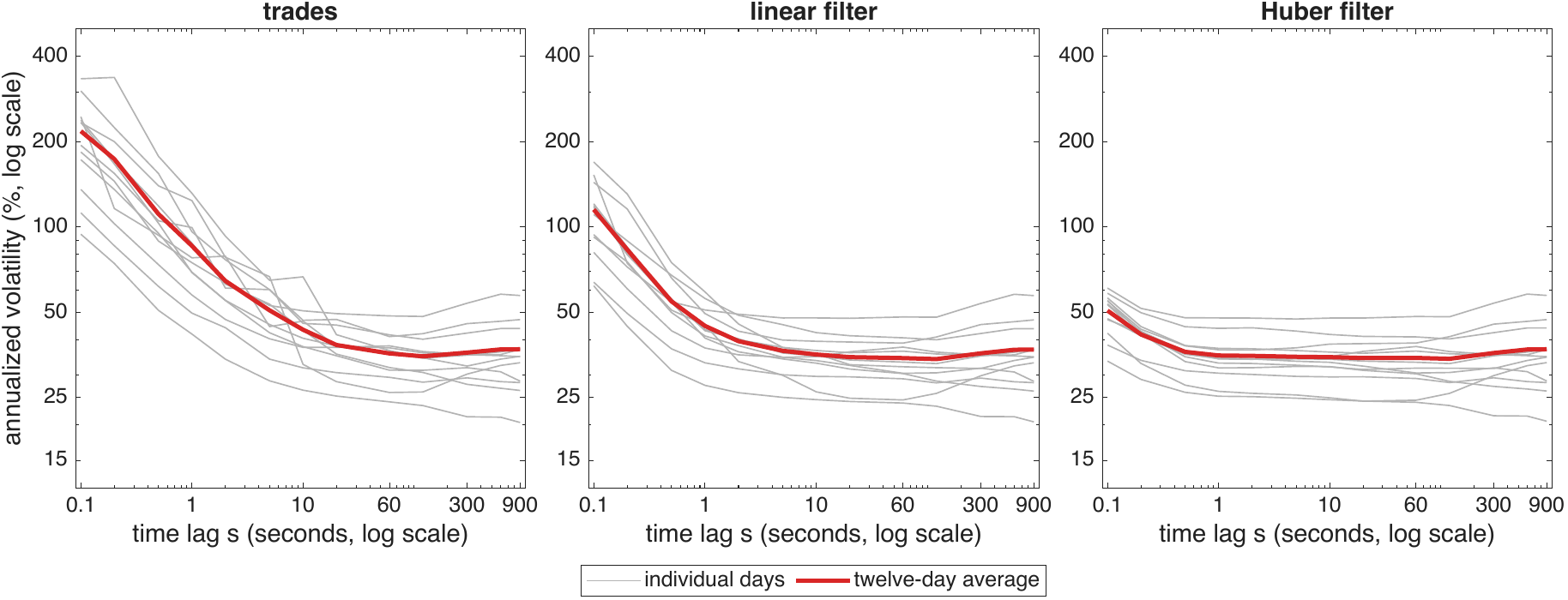}
    \caption{Volatility signature curves for \texttt{NVDA} over the 12 trading days in Oct 2024. The annualized vol $\hat{\sigma}100\sqrt{252}$ from the returns $\hat{X}_t-\hat{X}_{t,s}$ of Table~\ref{tab:preav_diffs_time}, against the lag $s$ in seconds, both axes on log scales, for: raw trades, linear filter \& Huber filter. Grey lines are the individual trading days \& the red line is the square root of the 12 day average of their squares. }
    \label{fig:preav_signature}
\end{figure}

\subsubsection{Filtered volatility}

Focus on the spot volatility.   
A simple localized spot variance estimator at time $u$ seconds (in the spirit of \cite{NelsonFoster(94)}) with preaveraging (in the spirit of \cite{MyklandZhang(16)}) is the one-sided numerical derivative of the estimated quadratic variation (QV) process, employing a bandwidth of $h\in (0,u]$ seconds, 
\begin{align*}
\hat{\sigma}^2(u) &= \frac{1}{h}\left\{\hat{\sigma}^{2*}(u) - \hat{\sigma}^{2*}(u-h)\right\},\quad u\in [h,23{,}400]  \\
& = \frac{1}{h} \sum_{j=0}^{n_{T_u,h}-1} Z_{T_u-j,s},\quad n_{T_u,h} = T_u-T_{u-h}, \\
&= \frac{n_{T_u,h}}{h} \sum_{j=0}^{T_u-2} w_{T_u,j}\, Z_{T_u-j,s},\quad w_{T_u,j} = \frac{1(j \le n_{T_u,h}-1)}{n_{T_u,h}},\quad j=0,1,...,T_u-2,
\end{align*} 
a uniform weight function applied to the past data using trade time where $h/n_{T_u,h}$ is the average time between the trades in the window.  As $h$ increases the $n_{T_u,h}$ increases, but the underlying spot volatility could move making the past data stale, implying a  trade-off if the goal is to recover the spot volatility \citep[e.g.][]{NelsonFoster(94),Kristensen(10)}.  

It is more effective empirically for the weights to decay smoothly, avoiding volatility down jumps as large older datapoints fall outside the window. Decaying weights have no $h$, but the average time between trades $h/n_{T_u,h}$ becomes the weighted average $\bar{\Delta}_{T_u}$, so now 
$$
\hat{\sigma}^2(u) = 
\frac{1}{\,\bar{\Delta}_{T_u}} \sum_{j=0}^{T_u-2} w_{T_u,j}\, Z_{T_u-j,s},\quad \text{where}\quad \bar{\Delta}_{T_u} = \sum_{j=0}^{T_u-2} w_{T_u,j}(\tau_{T_u-j}-\tau_{T_u-j-1})
$$ 

One set of weights is $w_{t,j} \propto \lambda^j$.  Then $\hat{\sigma}^2(u)$ is a scaled version of the EWMA of past $Z_{1:T_u,s}$, calculated in trade time. Here we take the hyperbolic weights from Ex.\ref{ex:hyperbolic}:  
$$
w_{t,j} = \frac{1}{\Psi(a+t) - \Psi(a)} \frac{1}{j+a},\quad a>0,\quad t\in \{2,...,T\},\quad j\in \{0,...,t-1\}, 
$$
recalling the ``Mandelbrot parameter'' $a$ \& $\Psi(a)$ is the digamma function.
Long-memory volatility is standard in finance \citep[e.g.][]{BaillieBollerslevMikkelsen(96),ComteRenault(98),AndersenBollerslevDieboldLabys(01),Corsi(09)} so a hyperbolic decay has appeal, although hyperbolic filtering is often computationally challenging with large datasets. Less so here. The filter is a weighted mean \& is computed using the stratified simulation estimator ($h = W_t^{-1}(0.999)$ \& $B=100$). With hyperbolic weights the head spans nearly the whole day, so its sums for all $t$ come from one FFT convolution of the $Z_{t,s}$ \& one of the clock gaps with the weights $1/(j+a)$, at $O(T\log T)$ flops for the day.

\subsection{End effects, diurnal effects and selecting ``$a$''}\label{sect:preav-ends}

Now focus on some practical aspects. First we need to deal with the noise scale $b_t$ for $t \le T_{60}$, the trades before 9.31am, as well as $b_0$. The first minute of trading is often intense \& quite different from  the rest of the day. Let $\tilde{\epsilon}^*_{1:T^*}$ be the $T^*$ non-zero median filter's residuals over the first minute's trading during the previous five openings. Then define, for  $ t \in \{0,...,T_{60}\}$,   
$$
b_t = \underset{\theta \ge 0}{\arg }\min \ \left(\frac{60-\tau_t}{T^* \times 60} \sum_{j=1}^{T^*} ||\tilde{\epsilon}^*_j| - \theta| 
+ \frac{\tau_t}{60} \frac{\sum_{j=0}^{t-1} w_{t,j} 
\mathbf{1}\{|\tilde{\epsilon}_{t-j}|>0\}\, \bigl||\tilde{\epsilon}_{t-j}|-\theta\bigr|}{\sum_{j=0}^{t-1} w_{t,j}\mathbf{1}\{|\tilde{\epsilon}_{t-j}|>0\}} \right), 
$$
recalling $\tilde{\epsilon}_{1:T_{60}}$ are the median filter's residuals on the day we are interested in, \& $\tau_t$ the seconds since the open of trade $t$, $\tau_0=0$; until the day's first nonzero residual the second term is dropped.  Thus $b_t$ is a weighted median of $\tilde{\epsilon}^*_{1:T^*}$ \& $\tilde{\epsilon}_{1:t}$, which uses the previous week's opening data but swaps it out as the new data becomes available.  This approach is not perfect.

Second, we need to deal with $\hat{\sigma}(u)$ for $u<15$, the first 15 seconds of trading \& set 
$$
\hat{\sigma}^2(u) = \frac{15-u}{15} \tilde{\sigma}^2 + \frac{u}{15} \frac{1}{\bar{\Delta}_{T_u}} \sum_{j=0}^{T_{u}-2} w_{T_{u},j} Z_{T_{u}-j,s},\quad u\in [0,15),
$$
where each of the previous 5 days gives the average spot variance over its first minute, \& $\tilde{\sigma}^2$ is the median of those 5 (pre-opening data was tried, but that had noticeably lower volatility than the first few minutes of the open data). $\hat{\sigma}$ shrinks the new datasource towards $\tilde{\sigma}$.  Hence $\hat{\sigma}(u)$ is initially determined by those 5 opening minutes, but as new data arises we swap out the average of those days' opening data with the new data --- trade by trade. Over the first $s$ seconds of a day there is no trade $s$ seconds earlier. For this filter we set $\hat{X}_{t,s}=X_1$, the price at the opening auction, and use the return only once that trade is at least 1 second old.

Diurnal volatility features are strong in financial markets, with intensive volatility at the start of the trading day, followed by declining volatility for much of the day \citep[e.g.][]{AndersenBollerslev(97jef)}.   For the {\tt NVDA} data, the volatility over the first 15 minutes is around 4 times higher than during the middle of the day.  For some markets, volatility rises near  the close.    

The diurnal feature can be exploited to improve the spot estimator.  Again start with  the Huber filtered $s=1$ second returns, employed to provide an initial estimate $\bar{\sigma}(u)$ using 10 seconds of data \& uniform weights.  Now we depart from the above.  For each value of $u$ we average these $\bar{\sigma}^2(u)$ over the previous 10 trading days, yielding $\grave{\sigma}^2(u)$. Then we fit a flexible Fourier functional \citep[e.g.][]{AndersenBollerslev(97jef)} form $g(u)$, with linear \& quadratic terms \& 4 cosine-sine pairs, on the $\log\{\grave{\sigma}^2(u)\}$, yielding an estimated diurnal fit of $\bar{\bar{\sigma}}(u)=\exp\{g(u)/2\}.$  
We then take our refined filtered volatility as 
$$
\hat{\sigma}^2(u) = \bar{\bar{\sigma}}^2(u)
\frac{1}{\bar{\Delta}_{T_u}} \sum_{j=0}^{T_u-2} w_{T_u,j}\, \frac{Z_{T_u-j,s}}{\bar{\bar{\sigma}}^2(\tau_{T_u-j})},\quad u \in [15,23{,}400].
$$
If the weights have very short memory then the impact of explicitly modeling the diurnal feature will be tiny, but with long memory effects it can make a difference.  Of course, more refined modeling can allow the diurnal feature to slowly change through time. 

We select $a$ using forecasting.  (1) we compute $\bar{\sigma}(u)$, a quite noisy estimate of the spot volatility $\sigma(u)$ using 10 seconds of the trade-by-trade Huber filtered price data with $s=1$ \& a uniform weight function.  This averages roughly 700 squared returns, but the returns are heavily overlapping. (2), we forecast $\bar{\sigma}(u)$ using $\hat{\sigma}(u-60;a)$, a hyperbolic filtered spot volatility 60 seconds before time $u$ noting that filter depends upon the ``$a$'' (note $\bar{\sigma}(u)$ \& $\hat{\sigma}(u-60;a)$ only share overlapping data through $\delta_t$, the Huber threshold, which has a 4-minute half-life, \& impacts the Huber filter).  We minimize the QLIKE objective function \citep[e.g.][]{patton2009evaluating,PattonSheppard(15)} to select ``$a$'', over the days in the sample
$$
\int_{75}^{23{,}400}\left\{\frac{\bar{\sigma}^2(u)}{\hat{\sigma}^2(u-60;a)} + \log \hat{\sigma}^2(u-60;a)\right\} {\mathrm d}u.
$$  
This is minimized at $a=563$ trades, roughly 8 seconds.   

We also fitted an Ex.\ref{ex:superposition} {\tt HExp}-type superposition structure, taking 4 exponential components and choosing their mixture weights by the same objective we use to choose $a$.  Figure~\ref{fig:preav_hyperhexp} shows the two fitted spot volatility paths: they are hard to tell apart, the correlation of their logarithms being $0.999$, so we keep the hyperbolic model, which is simpler.

\begin{figure}[h!]
    \centering
    \includegraphics[width=0.85\linewidth]{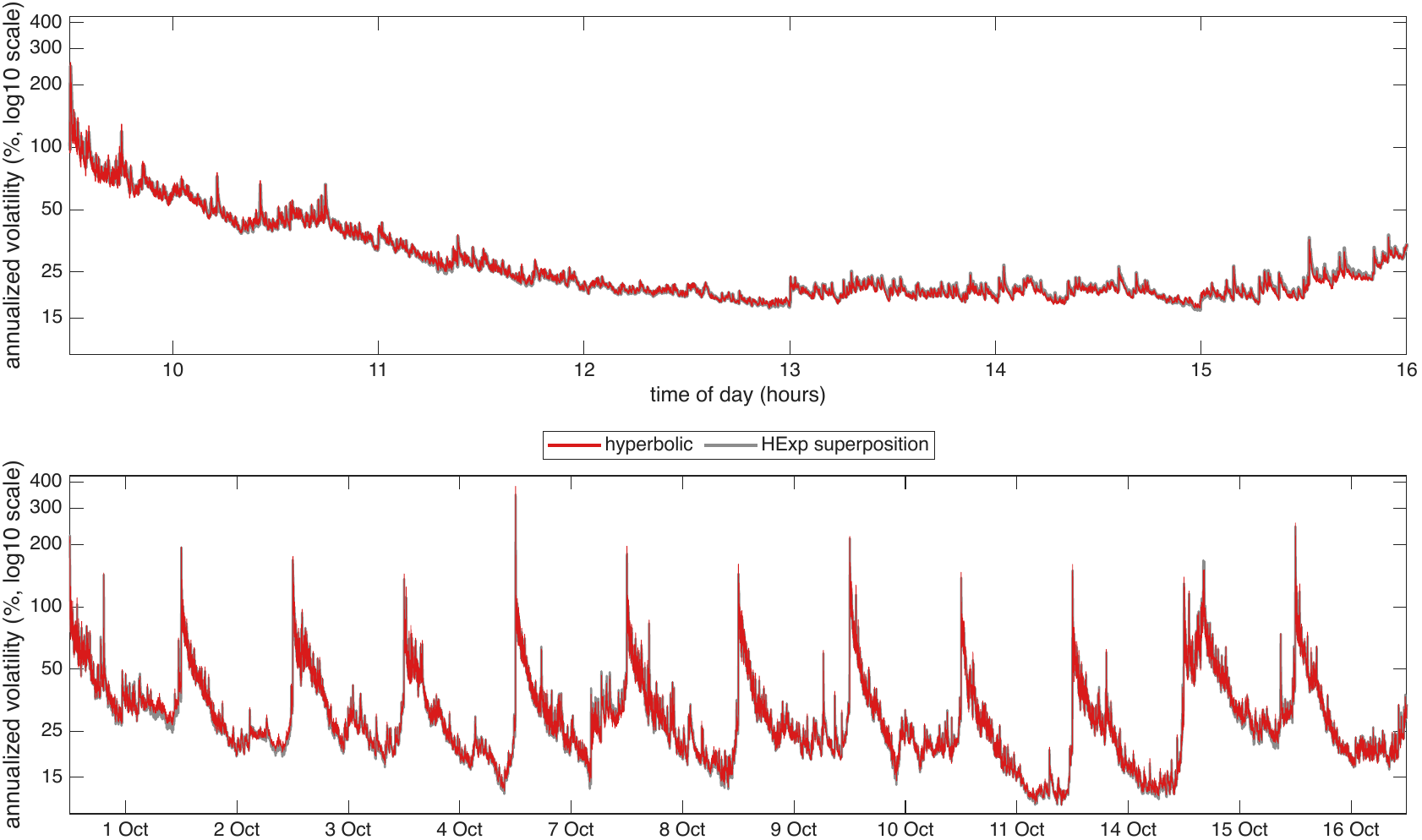}
    \caption{Hyperbolic against superposition weights for \texttt{NVDA} over 12 trading days in Oct 2024, the diurnally adjusted $\hat{\sigma}(u)100\sqrt{252}$ at $s=1$ second from the Huber filter, shown as annualized vol in \% on a log10 scale. The only difference between the 2 paths is the weight function: the hyperbolic weights of Ex.\ref{ex:hyperbolic} with $a=563$ in red, drawn on top, and the {\tt HExp} superposition of Ex.\ref{ex:superposition} in grey, 4 exponential components with half lives of 2, 20, 200 \& 2,000 seconds. Hyperparameters minimize the QLIKE objective. Top: 16 Oct, bottom: 12 days. }
    \label{fig:preav_hyperhexp}
\end{figure}

The robustness of the filtered volatility 
$
\hat{\sigma}(u)
$
is determined solely by the robustness of the price filter. We compute the path 3 ways: with the price filtered by the robust median filter, the robust Huber filter and the non-robust linear filter. The results for the median \& Huber filters are sufficiently close that we ignore the median case here.  

The top of Figure~\ref{fig:preav_volpath} shows the Huber- and linear-based paths as annualized volatilities, $\hat{\sigma}(u)100\sqrt{252}$ percent on 16 Oct 2024. Thus it can be read next to the top of Figure \ref{fig:preav_path} which showed the filtered path of the price. Recall it focused on the trade at \$118.85 around 12:43.

The spot volatility paths track each other through the quiet stretches, though the linear-based path lies about a quarter above the Huber-based one at the typical trade. That difference is the noise the linear filter leaves in (we saw this in the volatility signature plots in Figure \ref{fig:preav_signature} where the Huber filter allows us to use 1 second returns, but the linear filter would need roughly 10 second returns). At the \$118.85 trade the linear-based path jumps from around $23$ to around $273\%$ annualized and fades over the following half hour, while the Huber path stays where it is.
This difference is simply driven by the lack of robustness in the linear filter of the price. This distortion by the linear method causes a larger relative magnitude mistake in the spot volatility estimators later in the day.  It is systemic in this dataset.\footnote{Among the 54 assets, Appendix~\ref{app:xsec} shows that the difference mainly depends on the trading rate. The linear filter's spot volatility is 1.22 times the Huber filter's at the median level above 10 trades per second \& 1.01 times below 1. The slow stocks have relatively more outliers, not fewer, but moderate-sized ones, which do little harm to the linear filter.}       

\begin{figure}[h!]
    \centering
    \includegraphics[width=0.85\linewidth]{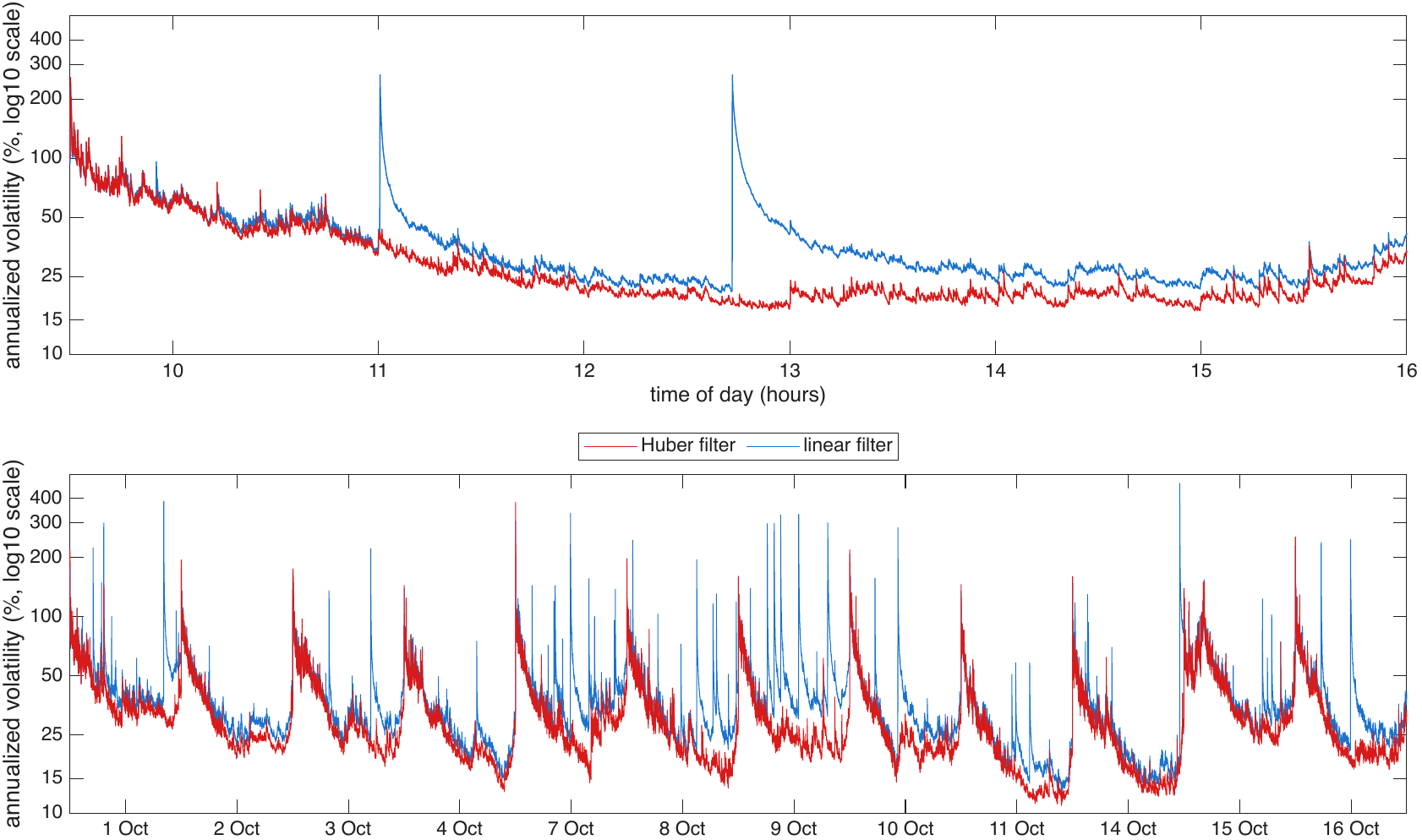}
    \caption{Filtered spot volatility for \texttt{NVDA} over 12 trading days in Oct 2024 ($19.6$M trades), the diurnally adjusted estimator $\hat{\sigma}(u)100\sqrt{252}$ at $s=1$ second with hyperbolic weights of Ex.\ref{ex:hyperbolic}, Mandelbrot parameter $a=563$, with the input prices computed using the Huber \& linear filters, shown as annualized volatility in percent on a log10 scale. The median \& Huber-filter based versions overlap almost everywhere. Top: 16 Oct, bottom: 12 days. }
    \label{fig:preav_volpath}
\end{figure}

The problematic nature of the linear filter can be seen in the bottom of Figure \ref{fig:preav_volpath}.  This shows the paths of the estimated spot volatility over the 12 trading days. The results based on the Huber filter of prices and the linear filter are roughly similar, but there are many upward blips caused simply by the lack of robustness in the linear filter.  It is not that there are no upward jumps in the spot volatility based on the Huber filter of prices: sometimes it happens responding to substantial news announcements.  It is that these jumps are much less common than in the linear filter case and that is due to the linear filter's severe vulnerability to outliers.

\subsection{Direction of extreme market microstructure noise: ``stale trades''}\label{sect:preav-direction}

The fifth filter is built on the residual of the Huber price filter of Section~\ref{sect:preav-path}, $\hat{\epsilon}_t = Y_t-\theta_t.$ We call a trade a statistical ``outlier'' when that residual is beyond the Huber threshold, $|\hat{\epsilon}_t|>\delta_{t-1}$, and there are $o=16{,}314$ of these on 16 October, roughly $1.1\%$ of the trades on that day. We view a statistical outlier as extreme market microstructure noise.  

The left of Figure~\ref{fig:preav_outdir} shows the Huber filtered price (black lines) \& the corresponding trades (dots), all recorded over a 12 minute period on 16 October.  Outliers above ($Y_t-\theta_t>\delta_{t-1}$) are shown in red, outliers below ($Y_t-\theta_t<-\delta_{t-1}$) in blue.  

As the robust filtered price falls, the trades outside the Huber band turn out to be mostly colored in red, being way above the filtered price.  When the filtered price rises the trades outside the Huber band are mostly blue, being way below the filtered price.      

\begin{figure}[h!]
    \centering
    \includegraphics[width=0.85\linewidth]{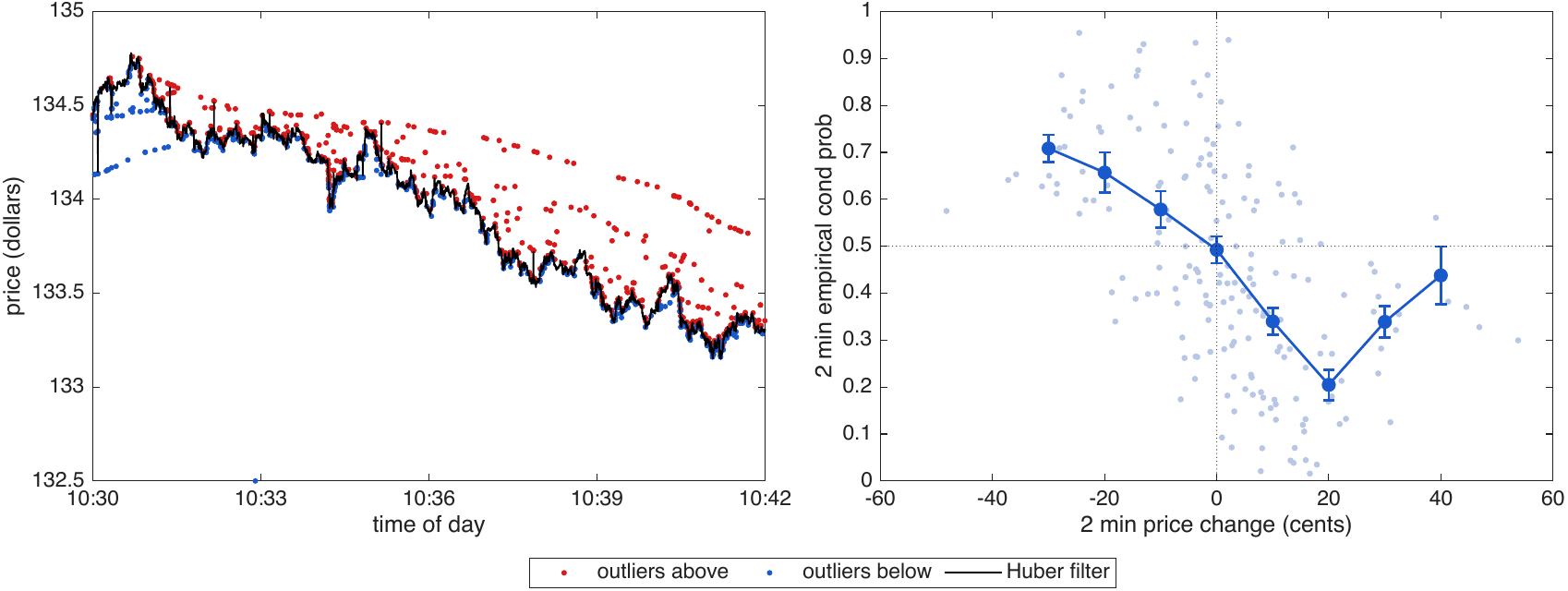}
    \caption{Outliers \& robust price for \texttt{NVDA} on 16 Oct 2024. Left: trades from 10:30 to 10:42; outliers above $Y_t-\theta_t>\delta_{t-1}$ in red, outliers below $Y_t-\theta_t<-\delta_{t-1}$ in blue, Huber filter in black. Right: over the trading day, the 2 minute empirical conditional probability that outliers are above (rather than below) the Huber band against the 2 minute change in the Huber filter: light points are single bins, the line \& bars are means \& s.e. over 10-cent ranges.}
    \label{fig:preav_outdir}
\end{figure}

An interpretation of these empirically overwhelming patterns, which systematically appear over all time periods in the dataset, is that some of the trades recorded  outside the Huber band are ``stale trades.'' Stale trades largely come from the late recording of off-exchange trades.  

To summarize this, the right of Figure \ref{fig:preav_outdir} bins the data into 2-minute intervals, again on 16 October. On the $y$-axis is the fraction of the trades outside the Huber band which are above the band (i.e. colored red) --- the conditional probability an outlier is above the price rather than being below.  On the $x$-axis is the filtered price at the end of the 2-minute bin minus the filtered price at the start of the bin.  The raw results are dots.  There is a strong negative correlation between the 2 variables.  To explore this in more detail we calculate sample means and standard errors, in blocks of 10-cent price move intervals.  Close to zero price changes there looks like a strong negative relationship.  For large absolute price changes there is some evidence the relationship may become weaker (i.e. a rotated $S$ shape), but there is not enough data to be more precise about this. The pattern holds on each of the 12 days.

We formalize these descriptive statistics by defining the ``directional conditional probability,'' the filtered probability through time that an outlier lies above the price: 
$$
q_t = \frac{p_t^+}{p_t^+ + p_t^-},\quad t \in \{T_{60}+1,...,T-1\},
$$
where $p_t^+$ is the filtered probability an outlier is above the Huber band and $p_t^-$ is the corresponding probability of being below the Huber band.  Of course $p_t^+ + p_t^-$ will vary around 0.01 as only around 1\% of the trades sit outside the band throughout this day. We exclude the opening minute warm-up, and the closing auction, which is never an outlier since $\theta_T=Y_T$. 

We build a filter for $p_t^+$ using the Bernoulli log-likelihood as the loss,
$
    L(p, d) = -d\log p - (1-d)\log(1-p),
$
where $p\in[0,1]$ is the candidate value of the filter and $d\in\{0,1\}$, taking the value of 1 if the data is above the Huber band ($Y_t-\theta_t>\delta_{t-1}$) and 0 otherwise. Then  
\begin{align*}p_t^+ &= \underset{p \in [0,1]}{\arg } \min \sum_{j=0}^{t-1} w_{t,j}\, L(p, d^+_{t-j})= \sum_{j=0}^{t-1} w_{t,j}\, d^+_{t-j},\quad  t \in \{T_{60}+1,...,T-1\}
\end{align*}
the ``binomial filter.''  Repeating this for $p_t^-$ yields
$$p_t^- = \underset{p \in [0,1]}{\arg } \min \sum_{j=0}^{t-1} w_{t,j}\, L(p, d^-_{t-j})  = \sum_{j=0}^{t-1} w_{t,j}\, d^-_{t-j},\quad  t \in \{T_{60}+1,...,T-1\},
$$
where $d^-_t$ takes the value of 1 if trade $t$ is below the Huber band ($Y_t-\theta_t<-\delta_{t-1}$) and 0 otherwise.  Then the filtered conditional probabilities are
$$
q_t = \frac{\sum_{j=0}^{t-1} w_{t,j}\, d^+_{t-j}}{\sum_{j=0}^{t-1} w_{t,j}\, (d^+_{t-j} + d^-_{t-j})}, \quad t \in \{T_{60}+1,...,T-1\}.
$$
The loss belongs to the canonical exponential family whose form drives the analytic solution (\citet{donkershephard2025CEF}). 

Throughout we have used weights decaying in clock time
$w_{t,j}\propto \lambda^{\tau_t-\tau_{t-j}},$ 
selecting $\lambda$ to maximize the sum of the Bernoulli (trade-by-trade) predictive log-likelihood over the 2 binary series $d_{T_{60}+1:T-1}^+$ and $d_{T_{60}+1:T-1}^-$, scoring trade $t$ by the filtered probabilities $p_{t-1}^{+}$ and $p_{t-1}^{-}$ at the previous trade,\footnote{The trade $t-1$ filters are the Definition~\ref{def:ewm} predictors for trade $t$. With $\alpha=1$ the predictor for trade $t$ minimizes $\sum_{j=1}^{t-1} w_{t,j} L(p,d^{+}_{t-j})$, \& these weights are the weights of the filter $p^{+}_{t-1}$ multiplied by the common factor $\lambda^{\tau_{t}-\tau_{t-1}}$, so the 2 minimizers coincide. Likewise for $p^{-}_{t-1}$. This relies on the exponential weights.} which yields a 13 second half life, although we note the log-likelihood is roughly flat for the half life in the range of 10 to 20 seconds.    

\begin{figure}[h!]
    \centering
    \includegraphics[width=0.85\linewidth]{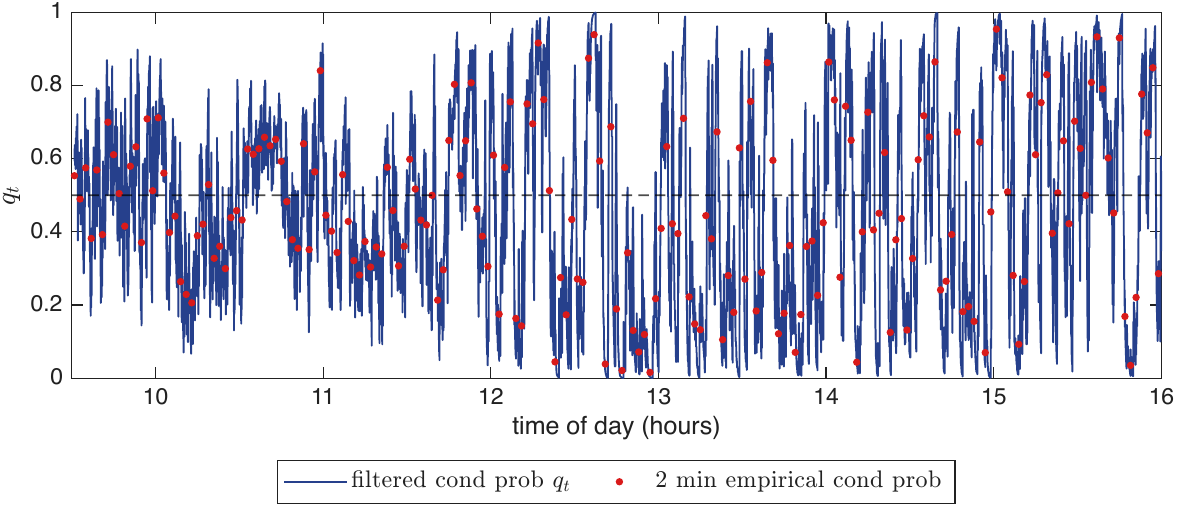}
     \caption{The directional conditional probability for \texttt{NVDA}, 16 Oct 2024. The filtered conditional probability $q_t$ that an outlier lies above the robust Huber band, under exponential weights decaying in clock time with half life of 13 seconds. The dashed line marks $0.5$. The red points are the empirical conditional probability using the 2-minute bins of Figure~\ref{fig:preav_outdir}.}
    \label{fig:preav_direction}
\end{figure}

Figure~\ref{fig:preav_direction} plots the filtered conditional probability  $q_t$ through time for \texttt{NVDA} throughout 16 October, along with the 2 minute empirical conditional probability from the right of Figure \ref{fig:preav_outdir}.  These are reasonably aligned, and $q_t$ is available trade by trade, but at a 13 second half-life it is estimated from roughly 12 outliers, so it is noisy.

\subsection{Other potential uses of the Huber filtered prices}

The focus has been on using the Huber filtered prices as an input into estimates of the time-varying volatility.  But the same type of filtered prices could be used as the basis for, for example, measuring the size of jumps in prices \citep{Mancini(01paper),BarndorffNielsenShephard(04jfe),BarndorffNielsenShephard(06test)}, estimating time-varying covariances \& betas \citep{BarndorffNielsenShephard(04multi),BarndorffNielsenHansenLundeShephard(11multi)}, factor pricing models \& asset allocation weights \citep{Pelger(19)}, statistical leverage \citep{WangMykland(14)} \& causal effects of monetary policy announcements on macro economic variables \citep{BauerSwanson(23)}.

\section{Conclusions}\label{sect:conc}

This paper suggests a non-recursive way of defining \& computing filters with general loss functions with the option to have a long-run anchor.  It  removes a computational bottleneck in the use of robust methods with flexible weight functions.  We provide some theory for this approach. A celebrated version of this setup is the exponentially weighted moving median. Losses and weights of the researcher's choosing can now be used to filter, predict and smooth at every observation; before, that was infeasible for long series outside uniform weights, exponential family likelihoods and linear filters.

The methods are illustrated on massive high frequency financial datasets, robustly filtering the price from which volatility measures are built. Robust methods provide rather simpler results than is often seen in the literature.  The consequent volatility signature plots are flat, day after day, even at the 1 second level. On the same trades, the linear filter's are not. We show that
the residuals, our proxy for the noise, have an infinite variance, which challenges the theoretical framework of the linear preaveraging methods developed in financial econometrics.  We also find an interesting predictable measure of the direction of extreme market microstructure noise.  

\baselineskip=12pt

\bibliographystyle{chicago}
\bibliography{neil.bib, Simon.bib}

\newpage

\baselineskip=20pt

\appendix

\begin{center}
\Large \textbf{Online appendix to: \\
``Filtering without recursion \\ and some of its uses in financial economics''}\\[3pt]
\normalsize
Simon Donker van Heel and Neil Shephard
\\
\today
\end{center}

\pagenumbering{arabic}
\renewcommand*{\thepage}{S\arabic{page}}
\setcounter{page}{1}
\setcounter{equation}{0}
\renewcommand{\theequation}{\Alph{section}.\arabic{equation}}
\setcounter{table}{0}
\renewcommand{\thetable}{\Alph{section}.\arabic{table}}
\setcounter{figure}{0}
\renewcommand{\thefigure}{\Alph{section}.\arabic{figure}}

\section{Robustness checks}\label{app:robust}

\subsection{Tail index}\label{sect:tailindex}

Could the Hill estimate of the tail index for the residuals being much less than 2, in Section \ref{sect:preav-tails} where it is around 1.4, be an artifact of special items rather than being systemic? A list of items which could be checked includes:
\begin{itemize}
\item effect of the intraday scale, 
\item effect of repricings, 
\item effect of the filter's estimation error, 
\item effect of the serial dependence in the residuals, 
\item does the result hold on other days, beyond 16 October?  
\end{itemize}

Our robustness checks reported in Table~\ref{tab:tailrobust} suggest the heavily tailed nature of the noise is systemic for the {\tt NVDA} trades.  The Hill estimate is below 2, on every one of the 12 days.  

\begin{table}[h!]
\centering
\begin{tabular}{@{}lcc}
& 16 October & Twelve days, min to max \\
\cmidrule(r){1-1} \cmidrule(r){2-2} \cmidrule(r){3-3}
All trades (as in the paper) & 1.43 & 1.23 to 1.81\\
Residual divided by $\delta_t$ & 1.51 & 1.31 to 1.85\\
Repricing windows removed & 1.43 & 1.14 to 1.64\\
Every 10th trade & 1.41 & 1.26 to 1.78\\
Six-trade half-life & 1.47 & 1.30 to 1.92\\
\end{tabular}
\caption{Hill estimator of the tail index of the Huber filter residual $\hat{\epsilon}_t = Y_t-\theta_t$, for \texttt{NVDA} over the twelve trading days 1 to 16 Oct 2024. The tail index estimator is the same as computed in Section~\ref{sect:preav-tails}. The twelve-day column gives the smallest and the largest of the twelve daily estimates. Second row: divides the residual by its own threshold $\delta_t$, removing scale variation through the day. Third row: drops trades where the filtered price moves more than 10 cents within $\pm 1$ second, 7\% to 25\% of a day, the windows where the residual picks up filter lag. Fourth row: thins out short-range dependence. Fifth row: doubles the filter's half-life with the thresholds held fixed, changing the filter's estimation error of equation (\ref{eqn:psi*}) with everything else unchanged.}
\label{tab:tailrobust}
\end{table}

In Appendix \ref{app:xsec} we repeat our exercises over 53 other stocks to see if the results for {\tt NVDA} generalize.  Over the 53 other stocks of Table~\ref{tab:xsec} the Hill estimator is between 1 \& 2 for 7 stocks, between 2 \& 3 for 28 stocks, between 3 \& 4 for 12 stocks and 4 or above for 6 stocks.  Whether or not the variance exists, heavy tails is important for high frequency data. Recall, a tail index over 2 would suggest that in principle linear filtering could be used for a preaveraged volatility estimator, but estimating the effect of the noise would likely be fragile, as for a tail index below 4 the estimator would converge very slowly due to the lack of a fourth moment.  Our results do suggest the empirical results for the Huber filter and linear filter are most materially distinct when the tail index is below 2.

\subsection{Comparing spot volatility estimators with $s=10$}

The impact of noise on the spot volatility scales with $\psi_t^{*2}/s$.  Figure \ref{fig:preav_volpath10} repeats Figure \ref{fig:preav_volpath}, comparing the Huber filter and the linear filter, but now with $s=10$, rather than $s=1$.  At $s=10$ the linear-based path lies $2.5\%$ above the Huber-based path at the typical trade, down from $23\%$ at $s=1$, and at the \$118.85 trade it jumps from around $20\%$ to around $85\%$ annualized, rather than to $273\%$.

\begin{figure}[h!]
    \centering
    \includegraphics[width=0.85\linewidth]{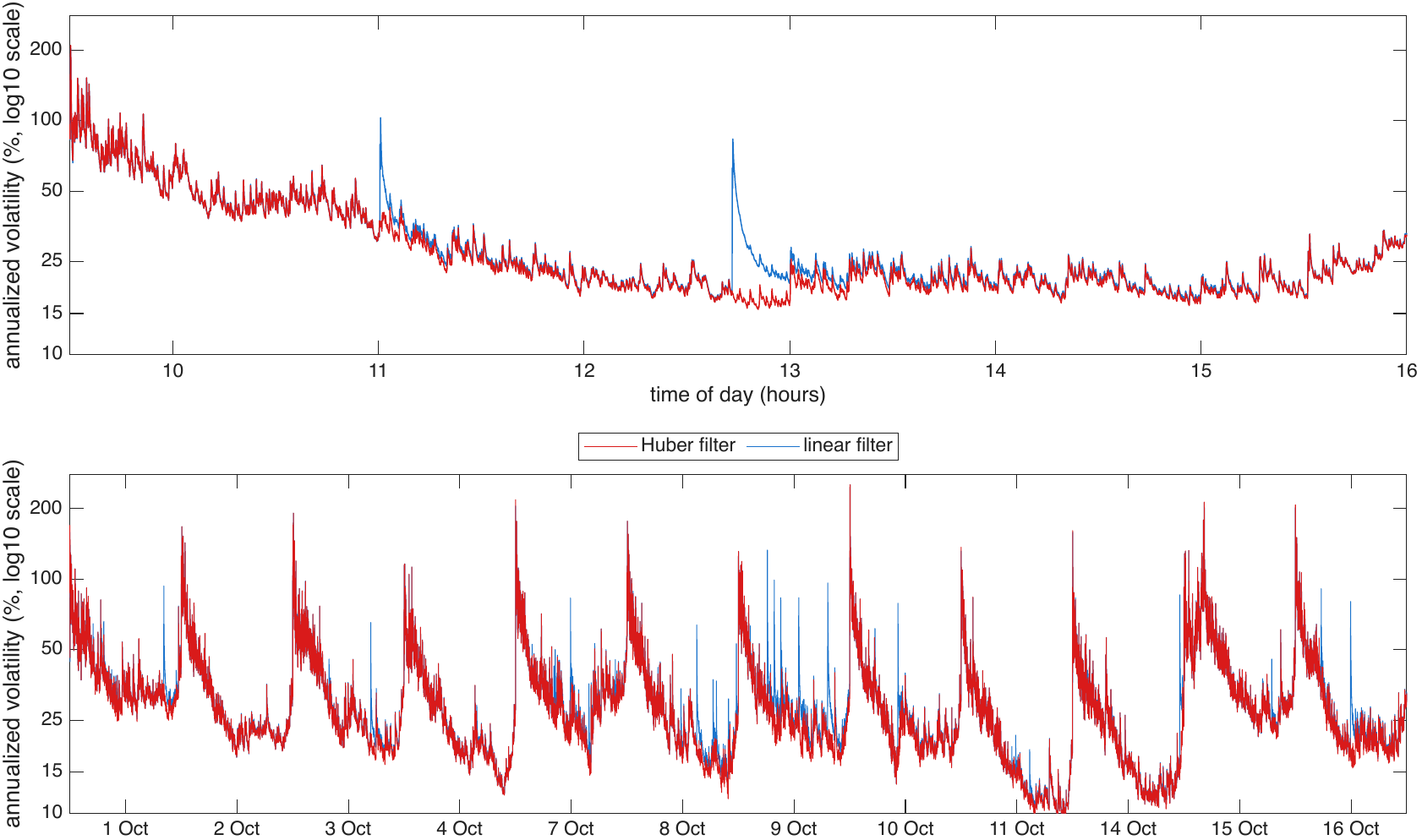}
    \caption{Filtered spot volatility for \texttt{NVDA} over 12 trading days in October 2024 ($19.6$ million trades), the diurnally adjusted estimator $\hat{\sigma}(u)100\sqrt{252}$ at $s=10$ seconds under the hyperbolic weights of Ex.\ref{ex:hyperbolic}, the Mandelbrot parameter $a=563$, with the input prices computed using the Huber filter and the linear filter, shown as annualized volatility in percent on a log10 scale. The median-filter based price path, not shown, overlaps the Huber-based one almost everywhere. Top: 16 October. Bottom: the 12 trading days.}
    \label{fig:preav_volpath10}
\end{figure}

\subsection{Volatility results using Huber filtered prices as $s$ varies}

Figure~\ref{fig:preav_volpath_s1s10} compares the filtered spot volatility based on the Huber price path at $s=1$ and $s=10$. The 2 paths track each other closely, with the $s=10$ path about $1\%$ lower at the typical trade.

Two effects work against each other. The averaging in the price filter shrinks a 1 second return more than a 10 second one, which lowers the $s=1$ path by more. The noise that remains after filtering scales with $1/(\tau_t-\tau_{t,s})$, so it raises the $s=1$ path by more. The correction $c_{t,s}$ removes the first effect but not the second. On a noise-free price the estimator has no dependence on $s$, so the difference here comes from the noise alone.

\begin{figure}[h!]
    \centering
    \includegraphics[width=0.85\linewidth]{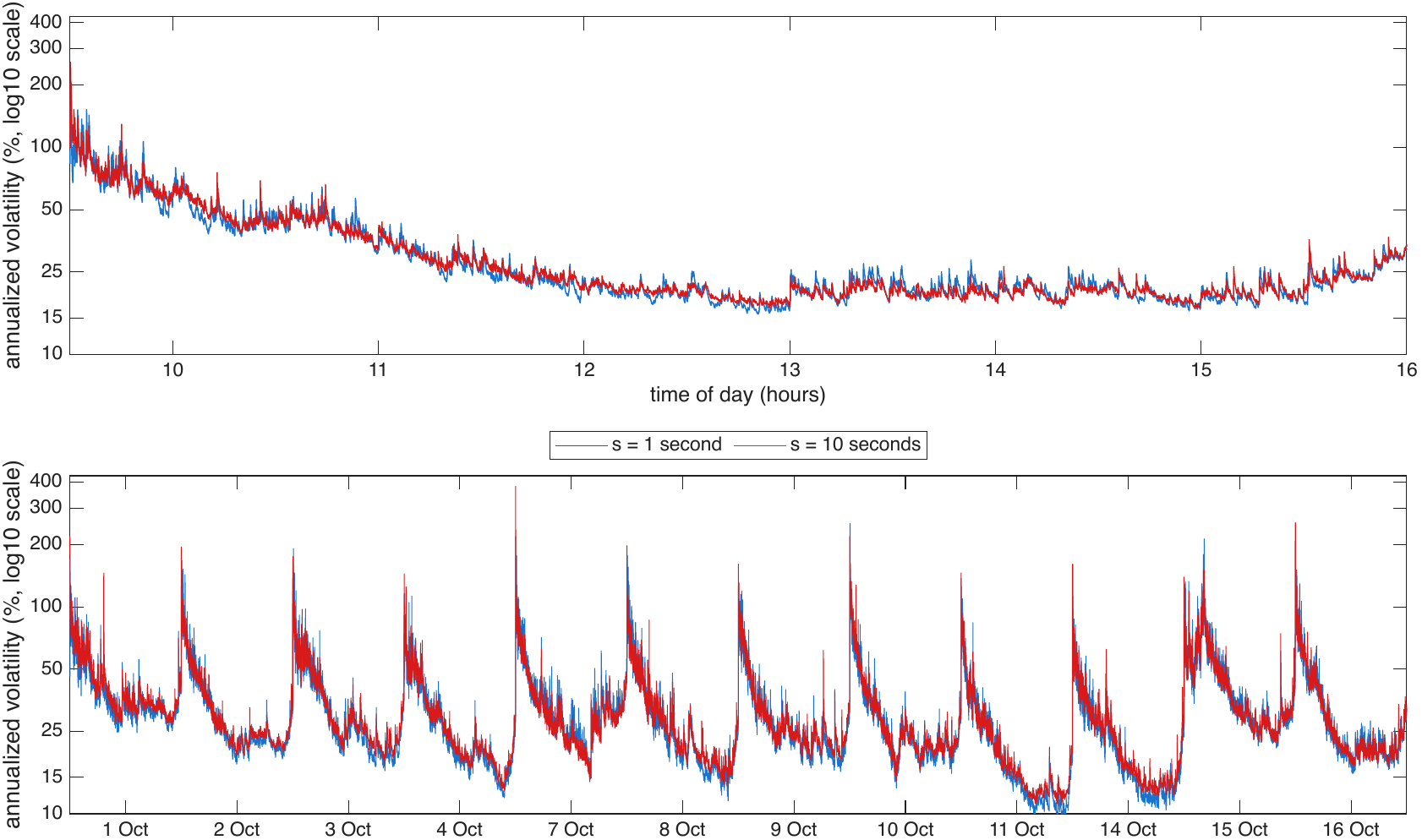}
    \caption{Filtered spot volatility for \texttt{NVDA} over 12 trading days in October 2024, the diurnally adjusted estimator at $s=1$ and $s=10$ seconds from the Huber filtered prices, under the hyperbolic weights with $a=563$, shown as annualized volatility in percent on a log10 scale. Top: 16 October. Bottom: the 12 trading days.}
    \label{fig:preav_volpath_s1s10}
\end{figure}

\subsection{Time-varying standard error of Huber filter}\label{sect:TVpsi}

To quantify how the Huber filter's standard error $\psi_t^*$ changes through time, we ran a uniform filter (a rolling estimate) based on 100,000 trades on 16 October. Figure~\ref{fig:preav_rollpsi} plots the result.  

\begin{figure}[h!]
    \centering
    \includegraphics[width=0.95\linewidth]{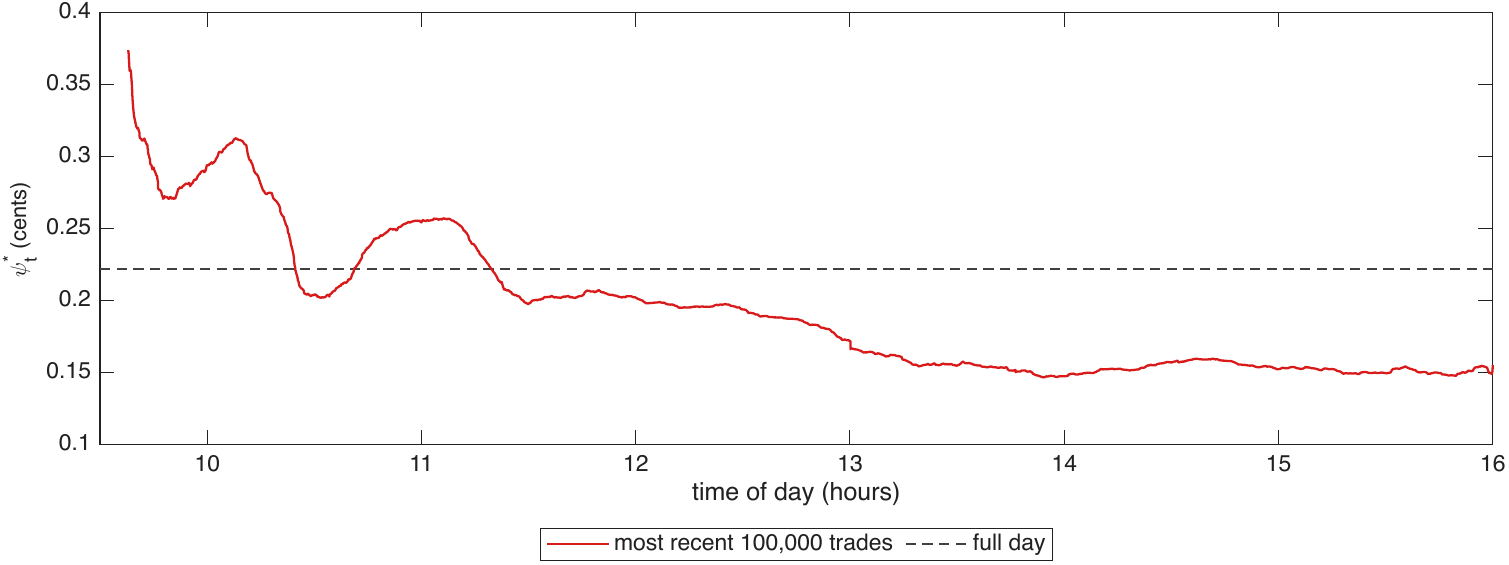}
    \caption{The Huber filter's standard error $\psi_t^*$ for \texttt{NVDA} on 16 October 2024, estimated on a uniform window of the most recent $100{,}000$ trades, moved forward $1{,}000$ trades at a time, the estimator of Table \ref{tab:errors} applied to each window. Dashed line:  full-day value, $0.22$ cents.}
    \label{fig:preav_rollpsi}
\end{figure}

There seems some substantial change, with the estimated standard error falling through the day.  In the second half of the day the standard error is around one half the value seen in the first hour of trading.  This halving is roughly the same as the change in values of the Huber threshold $\delta_t$ seen in Figure \ref{fig:preav_delta_hl}.

\subsection{The noise scale in the first minute}

Section~\ref{sect:preav-ends} starts the noise scale $b_t$ from the residuals of the previous 5 openings \& lets the current day's own residuals take over at the rate the clock runs, fully after one minute. The alternative is to let them take over at the rate the trades arrive. Trades arrive fastest in the first seconds of the day, when the residuals are also largest; on \texttt{NVDA} half of the first minute's trades arrive in its first 10 seconds. Figure~\ref{fig:six_delta_open} shows the Huber threshold $\delta_t = 2 \times 1.4826 \times b_t$ under both weighting rules for the 6 stocks of Figure~\ref{fig:six_signature}, each over the first 2 minutes of the opening on which the two differ most. By trade count $\delta_t$ jumps in the opening seconds and falls back; by the clock it moves gradually from the level of the previous openings to the day's own. After the first minute the two coincide, so nothing later in the day depends on the choice.

\begin{figure}[h!]
    \centering
    \includegraphics[width=0.95\linewidth]{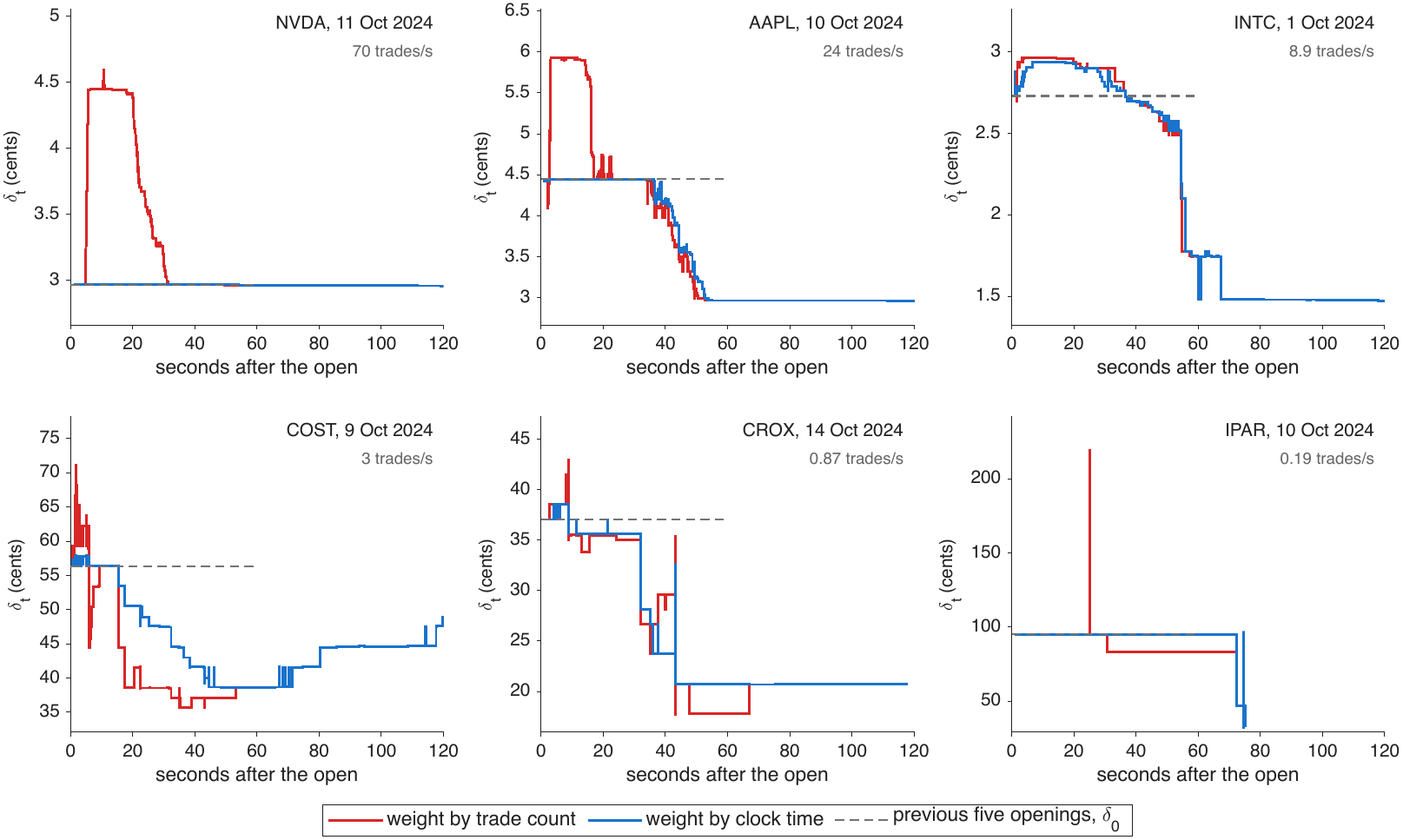}
    \caption{The Huber threshold $\delta_t$ in cents over the first 2 minutes of trading, for 6 stocks ordered by trading rate; each panel gives the stock's average trading rate over the 12 trading days. Each panel shows the day, of those 12, on which the trade-count weight moves $\delta_t$ most. Blue: the weights the paper uses, $(60-\tau_t)/(T^*\times 60)$ on the pooled previous openings and $\tau_t/60$ on the day's own residuals, from $b_t$ in Section~\ref{sect:preav-ends}. Red: those two replaced by $(T_{60}-t)/(T^*\times T_{60})$ and $t/T_{60}$, mixing by trade count rather than by the clock. Dashed: $\delta_0$, the level the previous five openings set. The two are identical for $t>T_{60}$.}
    \label{fig:six_delta_open}
\end{figure}

\subsection{Computing the finite sample correction}\label{app:c}

Simulate $M$ independent Brownian log-prices at the trade times $\tau_{1:T}$,
$$
B_m(\tau_t) - B_m(\tau_{t-1}) = \sigma \sqrt{\tau_t-\tau_{t-1}}\, U_{m,t}, \quad U_{m,t} \overset{iid}{\sim} N(0,1), \quad m \in \{1,...,M\},
$$
\& produce each $\hat{B}_m(\tau_{1:T})$ by the Huber filter at the half-life \& the threshold path $\delta_{1:T}$ of the day being measured; the median \& the linear columns of Table~\ref{tab:preav_diffs_time} instead use their own filter.\footnote{The $\sigma$ is the day's own volatility, the Huber filter's 15 minute level in Table~\ref{tab:preav_diffs_time}, converted to seconds. Its exact value matters little, since the threshold path $\delta_{1:T}$ is taken from the day and the correction for the Huber filter is within $1\%$ on average of the correction for the linear filter, which does not depend on $\sigma$. For computational speed we compute the expectation with the filter's window truncated at 30 lags, which costs at most $0.1\%$ in the estimate. We do this throughout the paper.} Then
$$
\hat{c}_{t,s} = \frac{\sum_{m=1}^{M} \{B_m(\tau_t)-B_m(\tau_{t,s})\}^2}{\sum_{m=1}^{M} (\hat{B}_{m}(\tau_t)-\hat{B}_{m}(\tau_{t,s}))^2}.
$$
The numerator has expectation $\sigma^2(\tau_t-\tau_{t,s})$ and the denominator $\mathbb{E}[(\hat{B}_{m}(\tau_t)-\hat{B}_{m}(\tau_{t,s}))^2]$, so $\hat{c}_{t,s}$ consistently estimates $c_{t,s}$ as $M\rightarrow \infty$ using the SLLN plus Slutsky's theorem. Taking both on the same $M$ paths helps as the $m$-th path of Brownian motion $B_m$ is a control variable for the $m$-th path's filtered version $\hat{B}_m$ \citep{Ripley(87)}. The 2 returns move together, so their ratio is far better determined than either sum, and where the filter changes the return little they are nearly equal and $\hat{c}_{t,s}$ is nearly 1 whatever the paths did.

One simulated price gives one draw of every return at every $s$, so $M$ prices give $M$ draws of all of them from one pass of the filter over each. The half-life and $\delta_{1:T}$ are inputs, so no tuning parameter is estimated inside the simulation.

\subsection{Volatility signatures on simulated prices}\label{app:simulated}

Here we look at simulated volatility signature plots based on the Huber filter.  The data takes the times of the trades as fixed at the $\tau_{1:T}$ seen on 16 October.  Then the prices are simulated using Brownian motion or the pure jump variance gamma process, either with or without very heavy tailed noise.  

\begin{table}[h!]
\centering
\begin{tabular}{lcc}
\toprule
 & $Y_t = X_t$ & $Y_t = X_t + \epsilon_t$\\
\midrule
$X$ Brownian motion & base case & adding noise\\
$X$ variance gamma & pure jump & noise \& pure jumps\\
\bottomrule
\end{tabular}
\caption{4 simulated cases. The rows are the price, the columns the trades. In all four $c_{t,s}$ is computed under the base case, the top left cell.}
\label{tab:fourcases}
\end{table}

First simulate a log-price at those trade times $\tau_{1:T}$: 
$$X_t-X_{t-1} = \sigma \sqrt{\tau_t-\tau_{t-1}} U_t,\quad  U_t \overset{iid}{\sim}  N(0,1),\quad t\in \{2,...,T\},
$$ 
with $\sigma$ set so the annualized volatility is constant at $40\%$, and set $Y_t=X_t$, without noise.  Here $X_t$ can be thought of as a simulated time-changed Brownian motion, using the empirical time-change, and 
$$
[X,X](23,400) = 23,400 \sigma^2.
$$  
The corresponding Huber filter's price $\hat{X}_{1:T}$ is produced by the median filter, followed by the scale filter and then the Huber filter --- all applied to $Y_{1:T}$.  In our experiments we run this 100 independent times, holding $\tau_{1:T}$ fixed, and draw the simulated price itself on the same 100.

The left of Figure~\ref{fig:preav_noisefree} computes $\hat{\sigma}_{\text{Annual};s}$ from the Huber filter's price $\hat{X}_{1:T}$. The purple line is the same estimator on the simulated price $X_{1:T}$ itself, drawn on the same 100 paths. For larger $s$ the day holds fewer nonoverlapping windows, which is why the replications vary more around the dashed line truth. Setting $c_{t,s}=1$, the green dotted line, leaves a large negative bias at small $s$, $38\%$ at $s=0.1$ seconds and $7\%$ at 1 second. That downward bias comes from the local robust averaging in the filters, not from the estimator. With $c_{t,s}$ it is gone from half a second on, where the red and purple lines coincide, and what remains below half a second is the bias of the estimated correction. A ratio of 2 averages is biased by a term of order $1/M$, and the returns are too heavy tailed for it to average away across the day. With $M=100$ it is $1.8\%$ of the truth at $s=0.1$ seconds.

Now add heavy tailed noise to the same paths,
\begin{equation}\label{eqn:noise}
Y_t = X_t + \epsilon_t, \quad \epsilon_t \overset{iid}{\sim} \omega \times t_{1.2},
\end{equation}
a scaled student-t random variable with 1.2 degrees of freedom and $\epsilon_{1:T}\ind X_{1:T}$.  We take $\omega=0.24$ cents at that day's opening price, which gives the simulated filter residual $Y_t-\theta_t$ the same median absolute size, $0.21$ cents, as the residual measured on 16 October. The right of Figure~\ref{fig:preav_noisefree} is the result. Setting $c_{t,s}=1$ the estimator is $39\%$ low at $s=0.1$ seconds against $38\%$ without the noise, so it barely responds to the noise. The noise pulls in 2 directions. It adds its own variation to the return, which pushes the estimate up. It also makes the residuals larger, so the scale filter raises $\delta_t$ from $0.16$ to $1.00$ cents, and with a higher threshold the filter clips fewer trades and smooths harder, which pushes the estimate down. The correction removes the effect of the smoothing, so what lies above the purple line at the shortest $s$ is the noise remaining in the filtered price. This is the case closest to the trades we observe, where noise is present and jumps are not.

\begin{figure}[h!]
    \centering
    \makebox[0.53\linewidth]{\textbf{Without noise}}\makebox[0.47\linewidth]{\textbf{With noise}}\\[+1mm]
    \includegraphics[width=\linewidth]{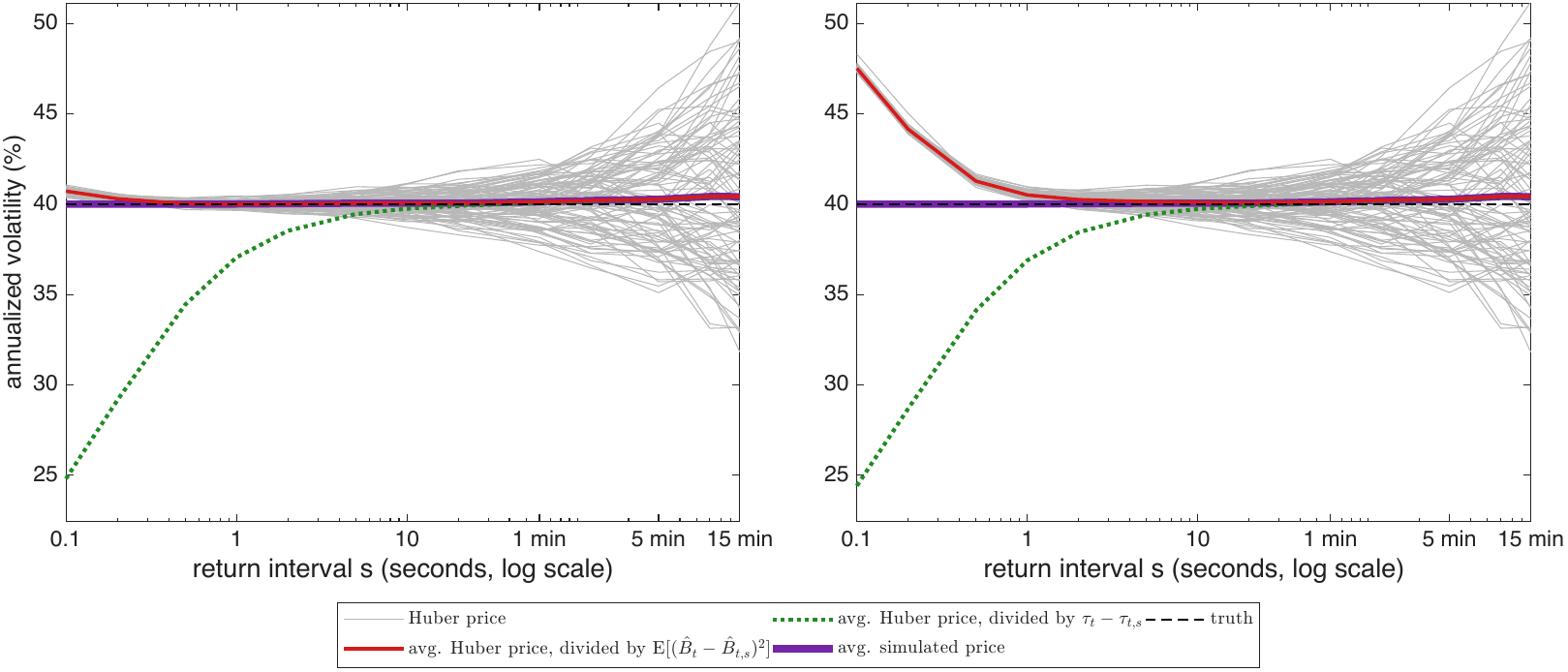}
    \caption{$\hat{\sigma}_{\text{Annual};s}$ applied to the Huber filter's price of the Brownian simulation, $Y_t=X_t$ on the left and $Y_t=X_t+\epsilon_t$ of (\ref{eqn:noise}) on the right, the same 100 paths in both. Each grey line is one replication; the red line is their average, the green dotted line the same average dividing by $\tau_t-\tau_{t,s}$ alone, the purple line the average over the simulated price $X_{1:T}$ itself, \& the dashed line the truth, $40\%$.}
    \label{fig:preav_noisefree}
\end{figure}

Now replace the Brownian price by the symmetric pure jump ``variance-gamma'' L\'{e}vy process \citep{MadanSeneta(90)}.  We start again with the no noise case.  This has discrete increments   
$$X_t-X_{t-1} = U_t \sqrt{Z(\tau_t) - Z(\tau_{t-1})}, \quad  U_t \overset{iid}{\sim}  N(0,1),\quad  
t\in \{2,...,T\},
$$
where $\{Z(u)\}\ind U_{1:T}$ is a gamma process with increments 
$$
Z(\tau_t) - Z(\tau_{t-1}) \overset{indep}{\sim} Ga(\sigma^2 \beta (\tau_t-\tau_{t-1}),\beta),\quad  
t\in \{2,...,T\},\quad Z(0)=0.
$$
For each day, the QV over the trade grid is different: $[X,X](23{,}400) = \sum_{t=2}^T U_t^2 \{Z(\tau_t) - Z(\tau_{t-1})\}$, a random variable with conditional mean $Z(23{,}400) \sim Ga(23{,}400\sigma^2\beta,\beta)$ \& unconditional mean $23{,}400 \sigma^2$. We take $Y_t=X_t$ \& apply our methods.  The LHS of Figure \ref{fig:preav_vg} plots 
$$
\hat{\sigma}_{\text{Annual};s} - 100\sqrt{252}\sqrt{[X,X](23,400)}
$$
against $s$, for 100 paths with $\beta=1$, measuring $X$ in percent, which makes the standard deviation of $Z(23{,}400)$ roughly $40\%$ of its mean.  

Averaged over the 100 paths the estimator is $3.3\%$ low at $s=1$ second, against $7.4\%$ on the Brownian price, and is within $0.6\%$ of the truth from 5 to 10 seconds. The bias is smaller because the variance gamma price barely moves between its larger jumps, so the residuals $Y_t-\theta_t$ are tiny there. The scale filter estimates $\delta_t$ from these residuals, so $\delta_t$ is in effect zero, the Huber loss is the absolute loss and the filter is a weighted median, which reproduces a jump as a step rather than averaging over it. The choice of half-life then hardly matters. The paths spread much more widely than in the Brownian case, and that spread comes from the price, not the filter. Only roughly 100 of the 1.5 million increments are larger than $0.01$ cents, and on a typical path 20 of them account for $99\%$ of the day's QV and 2 for half of it. The same estimator on the unfiltered $X_{1:T}$ spreads almost as much.

With $c_{t,s}$ the estimator is below the truth at the shortest $s$ and above it just beyond, so it is not exactly flat as it is in the Brownian case. The filter takes far less out of this price than out of a Brownian one, keeping $80\%$ of the estimate at $s=0.1$ seconds against $62\%$ there, because a jump survives a weighted median as a step. The correction is calibrated on the Brownian null, so here it is the wrong size, too small at the shortest $s$ and too large beyond them.

Finally, take the variance gamma price with the noise of (\ref{eqn:noise}), the right of Figure~\ref{fig:preav_vg}. Setting $c_{t,s}=1$ the estimator again barely moves when the noise is added, as in the Brownian case. With the correction the departures of the other 2 cases appear together, the noise remaining in the filtered price and the jump the filter clips as an outlier, and at the shortest $s$ the estimator is far above the truth.

Table~\ref{tab:preav_diffs} reports a first-order autocorrelation of $-0.18$ for the Huber filtered price on the data. These simulations say where that sign comes from. Averaging alone gives $+0.39$, its value on the Brownian price without noise, and jumps alone give $0.00$, its value on the variance gamma price without noise. Only the noise can turn it negative, so on \texttt{NVDA} the noise remaining in the filtered price outweighs the filter's averaging. Table~\ref{tab:xsecacf} shows this happens only on the most heavily traded stocks; where trading is slower the averaging dominates and the autocorrelation is positive.

\begin{figure}[h!]
    \centering
    \makebox[0.53\linewidth]{\textbf{Without noise}}\makebox[0.47\linewidth]{\textbf{With noise}}\\[+1mm]
    \includegraphics[width=\linewidth]{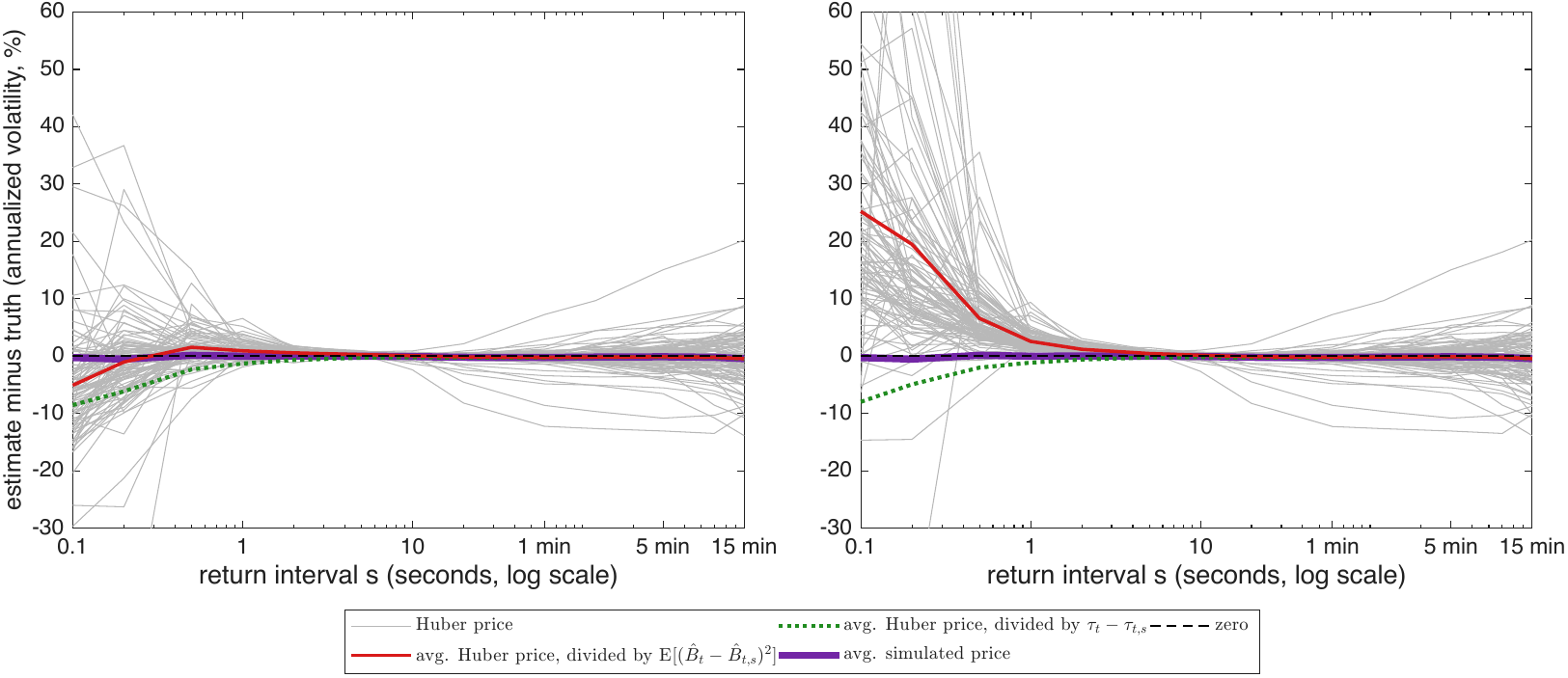}
    \caption{The same for the variance gamma price, drawn as estimate minus each path's own truth $100\sqrt{252}\sqrt{[X,X](23{,}400)}$ since the QV differs by path. Left: no noise, $Y_t=X_t$. Right: the same 100 paths with the noise of (\ref{eqn:noise}). Each grey line is one path, the red line their average, the green dotted line the same average dividing by $\tau_t-\tau_{t,s}$ alone, the purple line the average over the simulated price $X_{1:T}$ itself, \& the dashed line zero. Thirteen of the 200 paths leave the frame at the shortest lags.}
    \label{fig:preav_vg}
\end{figure} 

\setcounter{table}{0}\setcounter{figure}{0}

\section{Replicating results over 54 other assets}\label{app:xsec}

We repeated the whole calculation on another 53 stocks and 1 ETF, \texttt{SPY}.  Our Web Appendix, available at \href{https://simondonkervanheel.com/filtering-web-appendix/}{Donker van Heel's personal webpage}, gives detailed results for each asset.  It shows the trade median filter's optimal half-life does not change very much over stocks, nor does the scale filter's best half-life.  So throughout this Appendix we take a 3 trade half-life for the median filter, a 4 minute half-life for the scale filter \& estimate each stock's own Mandelbrot parameter $a$ following the procedure of Section~\ref{sect:preav-ends}. Table~\ref{tab:xsec} lists the assets, ordered by their average trading rate over the 12 days.  The results heavily depend on the trading rate.     

Figure~\ref{fig:six_signature} repeats Figure~\ref{fig:preav_signature} for 6 of the stocks.  It starts with \texttt{NVDA} at 70 trades a second and ends with \texttt{IPAR} at 1 trade every 5 seconds. The volatility signature of the trades lies well above the Huber filter's at $s=0.1$ seconds on every stock, by up to 10 times on a single day. The Huber filter's signature is flat from 1 second on the most heavily traded stocks, with a 1 minute level $0.92$ times its 1 second level, against roughly $0.6$ times on stocks with fewer than 3 trades per second. The advantage of the robust filter is largest where trading is fast and the noise is heaviest tailed. The price filter's half-life of 3 trades is a fraction of a second on \texttt{NVDA} and several seconds for stocks where trading is slow. The linear filter's curve for \texttt{IPAR} lies above the Huber filter's at the longer lags because of 1 extreme trade in the opening second of 14 October.

\begin{figure}[h!]
    \centering
    \includegraphics[width=0.95\linewidth]{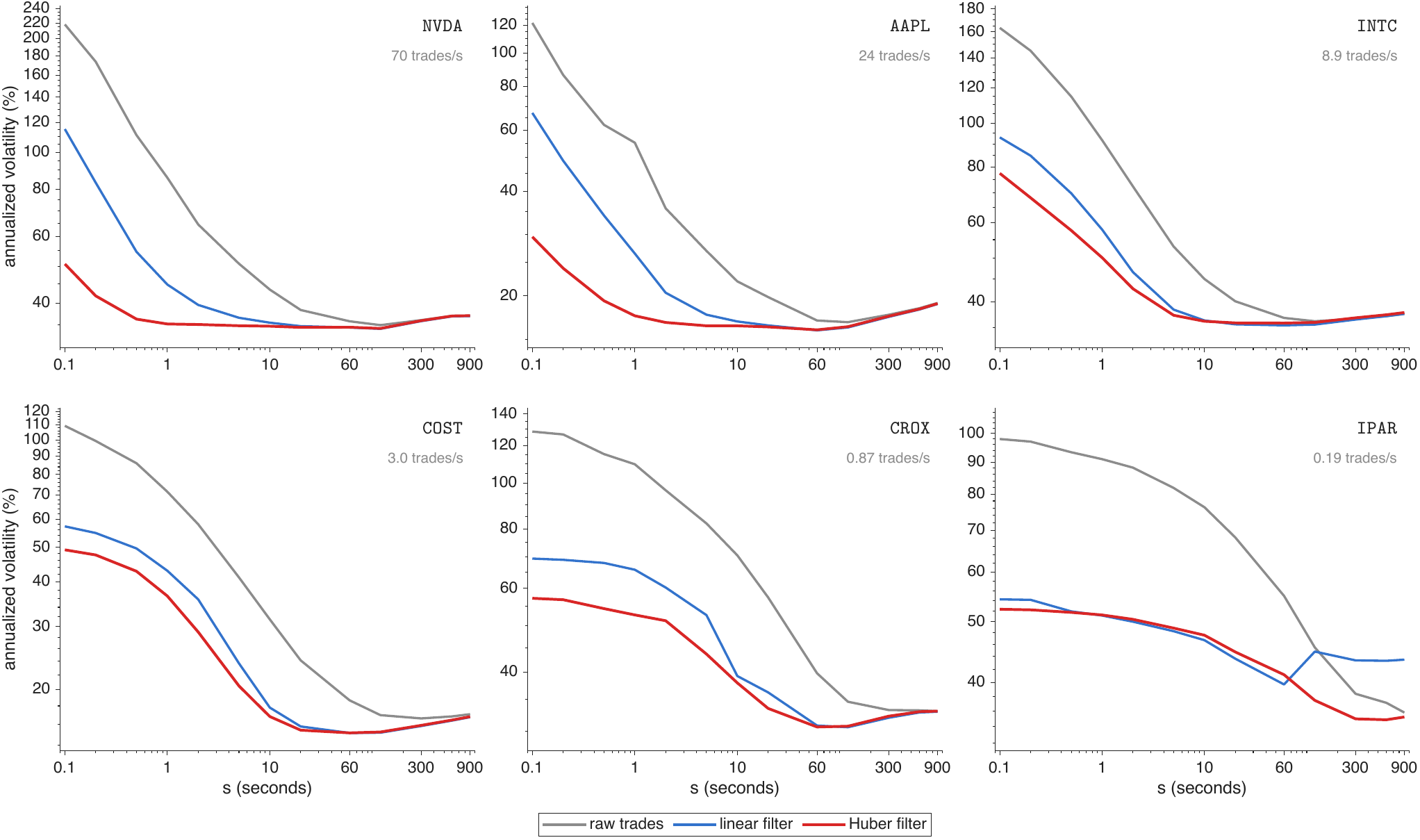}
    \caption{Volatility signature curves for six stocks ordered by trading rate, over the 12 trading days in Oct 2024; each panel gives the stock's average trading rate over those days. The annualized vol $\hat{\sigma}100\sqrt{252}$ from the returns $\hat{X}_t-\hat{X}_{t,s}$ of Table~\ref{tab:preav_diffs_time}, against the lag $s$ in seconds, both axes on log scales, for: raw trades, linear filter \& Huber filter. Each line is the square root of the 12 day average of the squares.}
    \label{fig:six_signature}
\end{figure}

At the largest Huber residual of 16 Oct 2024 the linear filter's spot volatility on \texttt{NVDA} jumps to 12 times its level in the minutes before (Figure~\ref{fig:preav_volpath}). Over the 54 stocks the median jump is to 4.6 times that level where trading is faster than 10 trades per second, to 1.2 times below 1 trade per second, and 19 of the 54 stocks jump by 3 times or more. 

The Hill index of the residual is at most 2.1 on every stock with more than 10 trades per second, so there the variance of the noise is infinite or at the boundary of existing, and above 2 on most of the others. The advantage of the Huber filter over the linear filter shrinks with it. The median spot volatility computed from linear filtered prices is 1.22 times the Huber filtered one with more than 10 trades per second \& 1.01 times with fewer than 1. On stocks with fewer than 1 trade per second the Huber filter treats 6\% of the trades as outliers, against 2\% on those with more than 10.\footnote{With a normal residual $\delta_t = 2 \times 1.4826 \times b_t$ is two standard deviations, since $1.4826$ turns the median absolute deviation into a normal standard deviation (Section~\ref{sect:preav-level}), and a normal distribution has $4.6\%$ of its mass outside two standard deviations; the slow stocks, whose residual has a finite variance, are close to that. On \texttt{NVDA} more than half of the residuals are exactly zero and $b_t$ is the median of the nonzero ones, so $\delta_t$ is large against the typical residual and few trades exceed it.} The linear filter only crashes in the cases with extremely heavy tails.

Table~\ref{tab:xsecacf} repeats the lag-1 ACF of Table~\ref{tab:preav_diffs} across the stocks. For the trades it is roughly $-0.45$ on every stock, the bounce between the bid and the ask. For the Huber filter's price its sign depends on the trading rate. The median over the 12 days is negative on 7 stocks only, all among the 8 most heavily traded, \texttt{NVDA}, \texttt{AAPL}, \texttt{AMZN}, \texttt{TSLA}, \texttt{AVGO}, \texttt{AMD} and \texttt{MU}, and positive on all the others. At lag 2 the bounce has gone from the trades, roughly $-0.02$ on every stock, while the Huber filter's price keeps a small positive dependence, roughly $0.05$ in every group.

The two half-lives were chosen by minimizing a loss for {\tt NVDA}, so it is fair to ask whether values picked on 1 stock suit the rest. They do. For the noise scale, the half-life minimizing the check loss of $b_{t,60}$ of Section~\ref{sect:preav-level} has a median of exactly 4 minutes over the 54 stocks, and imposing 4 minutes costs no stock more than 1\%. For the price filter, imposing the 3 trade half-life of Section~\ref{sect:preav-path} costs a median of 1.7\% of the $k$-trade-ahead Huber forecast loss at the worst horizon $k$ from 2 to 100 trades, and more than 5\% on 1 stock only, \texttt{GOOG}.

The ETF \texttt{SPY} is the fourth most heavily traded asset in the sample and has little noise relative to its price, $\psi^*$ of 0.06 basis points against 0.16 for \texttt{NVDA}. With so little noise the filter's averaging dominates the noise remaining in the filtered price, so the lag-1 autocorrelation is $+0.25$, while stocks with comparable trading rates are typically negative. The largest Huber residual of 16 Oct is a six-share trade at \$522.93, roughly 10\% below the market, and the linear filter's spot volatility jumps to 138 times its level in the minutes before; the Huber filter's does not.

\begin{table}[h!]
\centering
\begin{tabular}{@{}lrrrrrrr@{}}
\toprule
 & & \multicolumn{3}{c}{ACF of $\hat{X}_t-\hat{X}_{t-1}$ at lag 1} & \multicolumn{3}{c}{at lag 2}\\
\cmidrule(lr){3-5}\cmidrule(l){6-8}
 & & Huber & Huber & trades & Huber & Huber & trades\\
trades per second & stocks & 16 Oct & 12 days & 16 Oct & 16 Oct & 12 days & 16 Oct\\
\midrule
above 10 & 7 & $-$0.03 & $-$0.06 & $-$0.46 & 0.05 & 0.05 & $-$0.01\\
3 to 10 & 19 & 0.09 & 0.08 & $-$0.46 & 0.03 & 0.05 & $-$0.02\\
1 to 3 & 15 & 0.13 & 0.13 & $-$0.43 & 0.06 & 0.07 & $-$0.03\\
below 1 & 13 & 0.09 & 0.12 & $-$0.43 & 0.01 & 0.06 & $-$0.03\\
\bottomrule
\end{tabular}
\caption{The sample ACF of the 1-trade difference $\hat{X}_t-\hat{X}_{t-1}$ of Table~\ref{tab:preav_diffs} at lags 1 and 2, for the Huber filtered price and for the trades, medians over the stocks in each group of the average trading rate. 16 Oct is the day Table~\ref{tab:preav_diffs} uses; 12 days is each stock's median over the 12 trading days, then the median over the stocks in the group.}
\label{tab:xsecacf}
\end{table}

\begin{figure}[h!]
    \centering
    \includegraphics[width=0.95\linewidth]{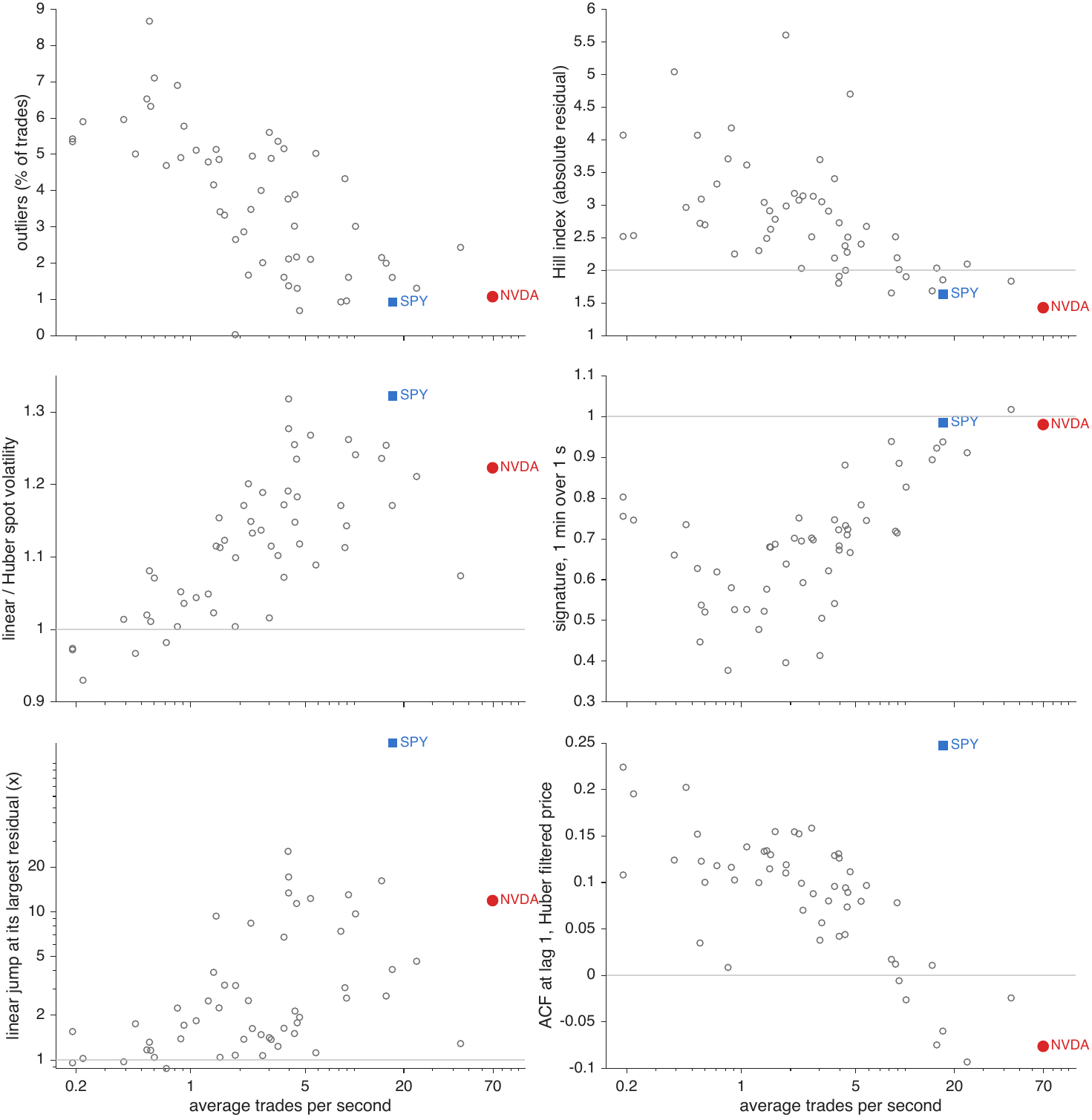}
      \caption{Six statistics of Table~\ref{tab:xsec} after the trading rate, plotted against it, \texttt{NVDA} in red and the ETF \texttt{SPY} in blue. The grey lines mark 2 for the Hill index, 1 for the 3 ratios and 0 for the ACF.}
    \label{fig:xsec}
\end{figure}

\begin{longtable}{@{}l r r r r r r r r r r@{}}
\toprule
 & trades & outlier & Hill & $\psi^*$ & $\psi^*$ & $\frac{\text{linear}}{\text{Huber}}$ & signat & linear & ACF & ACF\\
 & per sec & \% & estimate & cents & b.p. & & 60s/1s & jump & lag 1 & lag 2\\
\midrule
\endfirsthead
\toprule
 & trades & outlier & Hill & $\psi^*$ & $\psi^*$ & $\frac{\text{linear}}{\text{Huber}}$ & signat & linear & ACF & ACF\\
 & per sec & \% & estimate & cents & b.p. & & 60s/1s & jump & lag 1 & lag 2\\
\midrule
\endhead
\bottomrule
\endfoot
\bottomrule
\caption{Every stock and the ETF \texttt{SPY}, ordered by average trading rate, trades per second averaged over the 12 trading days in Oct 2024. Outlier \% is the share of trades the Huber filter treats as outliers on 16 Oct 2024. The Hill estimate is the tail index of the absolute residual $|\hat{\epsilon}_t|$, $\hat{\epsilon}_t=Y_t-\theta_t$, on that day, the average of the Hill estimates computed over the largest $0.1\%$ to $1\%$ of the absolute residuals, as in Section~\ref{sect:preav-tails}. $\psi^*$ is the standard error se$(\theta_t-X_t)$ of the Huber filter on 16 Oct 2024, in cents, computed as in Table~\ref{tab:errors}. $10^4\,\psi^*/X_T$ expresses the same standard error in basis points of the closing price $X_T$. Linear over Huber is the median spot volatility from linear filtered prices over the same from Huber filtered prices, sampled every 300 trades over the 12 days. Signature 1 min / 1 s is the 12 day average of the Huber volatility signature at $s=60$ seconds over the same at $s=1$ second, each the square root of the 12 day average of the squares. Linear jump is the linear filter's spot volatility around its largest Huber residual of 16 Oct, the largest value from 6 minutes before that trade to 12 minutes after it, over the median from 15 to 3 minutes before it. ACF lag 1 and lag 2 are the lag-1 and lag-2 sample ACFs of $\hat{X}_t-\hat{X}_{t-1}$ from the Huber filtered price, as in Table~\ref{tab:preav_diffs}, medians over the 12 days.}
\label{tab:xsec}
\endlastfoot
\texttt{\textbf{NVDA}} & 69.63 & 1.07 & 1.43 & 0.22 & 0.16 & 1.223 & 0.98 & 11.9 & $-$0.08 & 0.03\\
\texttt{TSLA} & 44.40 & 2.43 & 1.83 & 0.61 & 0.28 & 1.074 & 1.02 & 1.3 & $-$0.02 & 0.05\\
\texttt{AAPL} & 23.94 & 1.30 & 2.10 & 0.29 & 0.13 & 1.211 & 0.91 & 4.6 & $-$0.09 & 0.05\\
\texttt{SPY} & 17.06 & 0.92 & 1.63 & 0.35 & 0.06 & 1.322 & 0.99 & 137.6 & 0.25 & 0.20\\
\texttt{AMD} & 17.00 & 1.60 & 1.85 & 0.47 & 0.30 & 1.171 & 0.94 & 4.1 & $-$0.06 & 0.06\\
\texttt{AMZN} & 15.60 & 2.00 & 2.03 & 0.32 & 0.17 & 1.254 & 0.92 & 2.7 & $-$0.07 & 0.04\\
\texttt{MSFT} & 14.64 & 2.15 & 1.68 & 0.70 & 0.17 & 1.236 & 0.89 & 16.2 & 0.01 & 0.07\\
\texttt{AVGO} & 10.14 & 3.01 & 1.90 & 0.68 & 0.38 & 1.241 & 0.83 & 9.7 & $-$0.03 & 0.06\\
\texttt{MU} & 9.19 & 1.61 & 2.01 & 0.42 & 0.38 & 1.262 & 0.89 & 13.0 & $-$0.01 & 0.06\\
\texttt{INTC} & 8.94 & 0.96 & 2.19 & 0.12 & 0.52 & 1.143 & 0.71 & 2.6 & 0.08 & 0.03\\
\texttt{META} & 8.74 & 4.33 & 2.51 & 1.65 & 0.29 & 1.113 & 0.72 & 3.1 & 0.01 & 0.02\\
\texttt{GOOG} & 8.25 & 0.93 & 1.65 & 0.36 & 0.21 & 1.171 & 0.94 & 7.4 & 0.02 & 0.10\\
\texttt{COIN} & 5.81 & 5.03 & 2.67 & 1.74 & 0.83 & 1.089 & 0.74 & 1.1 & 0.10 & 0.04\\
\texttt{MRVL} & 5.39 & 2.10 & 2.40 & 0.43 & 0.52 & 1.268 & 0.78 & 12.3 & 0.08 & 0.08\\
\texttt{RIVN} & 4.62 & 0.69 & 4.70 & 0.10 & 0.94 & 1.118 & 0.67 & 1.9 & 0.11 & 0.04\\
\texttt{HOOD} & 4.47 & 1.30 & 2.51 & 0.19 & 0.72 & 1.183 & 0.72 & 1.8 & 0.09 & 0.07\\
\texttt{CELH} & 4.43 & 2.17 & 2.28 & 0.45 & 1.33 & 1.235 & 0.71 & 11.4 & 0.07 & 0.07\\
\texttt{QCOM} & 4.33 & 3.89 & 2.00 & 0.98 & 0.57 & 1.148 & 0.73 & 2.1 & 0.09 & 0.07\\
\texttt{PYPL} & 4.30 & 3.02 & 2.38 & 0.27 & 0.33 & 1.255 & 0.88 & 1.5 & 0.04 & 0.04\\
\texttt{CSCO} & 3.96 & 1.37 & 1.91 & 0.17 & 0.30 & 1.277 & 0.67 & 17.2 & 0.13 & 0.05\\
\texttt{PEP} & 3.96 & 2.11 & 2.73 & 0.48 & 0.27 & 1.318 & 0.68 & 13.4 & 0.04 & 0.07\\
\texttt{AMAT} & 3.93 & 3.77 & 1.80 & 0.91 & 0.49 & 1.191 & 0.72 & 25.5 & 0.13 & 0.09\\
\texttt{SBUX} & 3.71 & 1.61 & 2.19 & 0.34 & 0.36 & 1.172 & 0.75 & 6.7 & 0.13 & 0.07\\
\texttt{NFLX} & 3.71 & 5.15 & 3.40 & 3.31 & 0.47 & 1.072 & 0.54 & 1.6 & 0.10 & 0.01\\
\texttt{CRWD} & 3.41 & 5.36 & 2.91 & 2.17 & 0.71 & 1.102 & 0.62 & 1.2 & 0.08 & 0.03\\
\texttt{ADBE} & 3.10 & 4.89 & 3.05 & 2.86 & 0.57 & 1.115 & 0.51 & 1.4 & 0.06 & 0.02\\
\texttt{COST} & 3.02 & 5.60 & 3.70 & 3.82 & 0.43 & 1.016 & 0.41 & 1.4 & 0.04 & $-$0.01\\
\texttt{DKNG} & 2.75 & 2.01 & 3.13 & 0.25 & 0.67 & 1.189 & 0.70 & 1.1 & 0.09 & 0.09\\
\texttt{TXN} & 2.69 & 4.00 & 2.51 & 1.09 & 0.54 & 1.137 & 0.70 & 1.5 & 0.16 & 0.12\\
\texttt{PANW} & 2.38 & 4.95 & 3.14 & 2.69 & 0.72 & 1.133 & 0.59 & 1.6 & 0.07 & 0.02\\
\texttt{ABNB} & 2.33 & 3.48 & 2.03 & 0.84 & 0.62 & 1.149 & 0.69 & 8.4 & 0.10 & 0.07\\
\texttt{GILD} & 2.25 & 1.67 & 3.07 & 0.32 & 0.36 & 1.201 & 0.75 & 2.5 & 0.15 & 0.12\\
\texttt{DDOG} & 2.11 & 2.86 & 3.18 & 0.82 & 0.65 & 1.171 & 0.70 & 1.4 & 0.15 & 0.09\\
\texttt{ETSY} & 1.88 & 2.65 & 2.98 & 0.46 & 0.87 & 1.099 & 0.64 & 3.2 & 0.12 & 0.07\\
\texttt{LCID} & 1.87 & 0.03 & 5.60 & 0.08 & 2.37 & 1.004 & 0.40 & 1.1 & 0.11 & 0.00\\
\texttt{TTD} & 1.61 & 3.32 & 2.78 & 0.82 & 0.69 & 1.123 & 0.69 & 3.2 & 0.15 & 0.09\\
\texttt{ROKU} & 1.51 & 3.42 & 2.63 & 0.81 & 1.05 & 1.113 & 0.68 & 1.0 & 0.13 & 0.07\\
\texttt{TEAM} & 1.49 & 4.86 & 2.91 & 1.95 & 1.03 & 1.154 & 0.68 & 2.2 & 0.11 & 0.06\\
\texttt{INTU} & 1.43 & 5.13 & 2.49 & 4.12 & 0.68 & 1.115 & 0.58 & 9.3 & 0.13 & 0.04\\
\texttt{AMGN} & 1.38 & 4.16 & 3.04 & 1.82 & 0.57 & 1.023 & 0.52 & 3.9 & 0.13 & 0.05\\
\texttt{ISRG} & 1.28 & 4.79 & 2.30 & 3.19 & 0.67 & 1.049 & 0.48 & 2.5 & 0.10 & 0.02\\
\texttt{VRTX} & 1.08 & 5.11 & 3.61 & 3.95 & 0.81 & 1.044 & 0.53 & 1.8 & 0.14 & 0.05\\
\texttt{REGN} & 0.91 & 5.78 & 2.25 & 10.69 & 1.06 & 1.036 & 0.53 & 1.7 & 0.10 & 0.03\\
\texttt{CROX} & 0.87 & 4.91 & 4.18 & 1.35 & 0.98 & 1.052 & 0.58 & 1.4 & 0.12 & 0.06\\
\texttt{MELI} & 0.83 & 6.90 & 3.71 & 32.27 & 1.58 & 1.004 & 0.38 & 2.2 & 0.01 & $-$0.01\\
\texttt{CALM} & 0.71 & 4.69 & 3.32 & 1.13 & 1.20 & 0.982 & 0.62 & 0.9 & 0.12 & 0.09\\
\texttt{ORLY} & 0.60 & 7.11 & 2.70 & 12.75 & 1.06 & 1.071 & 0.52 & 1.0 & 0.10 & 0.06\\
\texttt{WING} & 0.57 & 6.32 & 3.09 & 6.20 & 1.64 & 1.011 & 0.54 & 1.2 & 0.12 & 0.05\\
\texttt{BKNG} & 0.56 & 8.67 & 2.72 & 70.21 & 1.62 & 1.081 & 0.45 & 1.3 & 0.03 & 0.00\\
\texttt{DUOL} & 0.54 & 6.53 & 4.07 & 5.08 & 1.77 & 1.020 & 0.63 & 1.2 & 0.15 & 0.08\\
\texttt{MIDD} & 0.46 & 5.01 & 2.96 & 1.52 & 1.10 & 0.967 & 0.73 & 1.8 & 0.20 & 0.11\\
\texttt{SAIA} & 0.39 & 5.96 & 5.04 & 11.51 & 2.50 & 1.014 & 0.66 & 1.0 & 0.12 & 0.05\\
\texttt{FIZZ} & 0.22 & 5.90 & 2.53 & 0.68 & 1.47 & 0.930 & 0.75 & 1.0 & 0.20 & 0.11\\
\texttt{LANC} & 0.19 & 5.35 & 2.52 & 3.29 & 1.83 & 0.972 & 0.76 & 1.0 & 0.11 & 0.01\\
\texttt{IPAR} & 0.19 & 5.42 & 4.07 & 2.29 & 1.94 & 0.974 & 0.80 & 1.6 & 0.22 & 0.10\\
\end{longtable}

Figure~\ref{fig:xsec} plots six statistics of Table~\ref{tab:xsec} against the average number of trades per second.  The percentage of outliers for the Huber filter declines with trading intensity.  A typical feature is that with high intensity the percentage is low but very heavy tailed.  With low intensity trading the outliers are quite common but thinner tailed.   

All but six Hill estimators are below 4, with the estimator typically declining with trading intensity.

\end{document}